%% file: paper.tex
\PassOptionsToPackage{table,x11names,dvipsnames}{xcolor}
\documentclass[acmsmall]{acmart}
\AtBeginDocument{%
  }

\setcopyright{acmlicensed}
\copyrightyear{2026}
\acmYear{2026}
\acmDOI{XXXXXXX.XXXXXXX}

\input{preamble}

\input{def}

\begin{document}

\title{BeePL: Correct-by-compilation kernel extensions}

\author{Swarn Priya}
\orcid{0000-0001-5820-6042}
\affiliation{%
  \institution{Virginia Tech}
  \city{Blacksburg}
  \country{USA}
}
\email{swarnp@vt.edu}

\author{Frédéric Besson}
\orcid{0000-0001-6815-0652}
\affiliation{%
  \institution{Inria}
  \city{Rennes}
  \country{France}
}
\email{frederic.besson@inria.fr}

\author{Connor Sughrue}
\orcid{}
\affiliation{%
  \institution{Virginia Tech}
  \city{Blacksburg}
  \country{USA}
}
\email{cpsughrue@vt.edu}

\author{Tim Steenvoorden}
\orcid{0000-0002-8436-2054}
\affiliation{%
  \institution{Open Universiteit of the Netherlands}
  \city{Heerlen}
  \country{Netherlands}
}
\email{tim.steenvoorden@ou.nl}

\author{Jamie Fulford}
\orcid{}
\affiliation{%
  \institution{University of Illinois Urbana-Champaign}
  \city{Urbana-Champaign}
  \country{USA}
}
\email{jamiehf2@illinois.edu}

\author{Freek Verbeek}
\orcid{0000-0002-6625-1123}
\affiliation{%
  \institution{Open Universiteit of the Netherlands}
  \city{Heerlen}
  \country{Netherlands}
}
\affiliation{%
  \institution{Virginia Tech}
  \city{Blacksburg}
  \country{USA}
}
\email{freek@vt.edu}

\author{Binoy Ravindran}
\orcid{0000-0002-8663-739X}
\affiliation{%
  \institution{Virginia Tech}
  \city{Blacksburg}
  \country{USA}
}
\email{binoy@vt.edu}

\renewcommand{\shortauthors}{}

\begin{abstract}
eBPF (extended Berkeley Packet Filter)  is a technology that allows developers to safely extend kernel functionality without modifying kernel source code or developing loadable kernel modules. Since the kernel governs critical system operations and enforces isolation boundaries between user space and privileged data, any mechanism that modifies its behavior must meet the highest standards of safety and correctness. To this end, the eBPF toolchain includes a verifier, which statically checks safety properties such as memory access validity, bounded loops, and type correctness before loading the program into the kernel. However, several recent studies have demonstrated that the existing verifier is both overly conservative in some cases—rejecting valid programs—and unsound in others, permitting unsafe behavior that violates the intended semantics of the kernel interface.

To address these challenges, we introduce BeePL, a domain-specific language for eBPF with a formally verified type system. The BeePL type system, along with the language design, statically enforces key safety properties such as type-correct memory access, safe pointer usage, absence of unbounded loops, and structured control flow. These guarantees are backed by formal type soundness proofs, ensuring that well-typed programs satisfy the safety invariants required by the eBPF execution environment. BeePL also proves that well-typed source programs meet critical eBPF-specific properties related to memory safety, termination, and control flow, enabling high-level reasoning prior to compilation. For properties not fully enforceable statically—such as dynamic bounds and undefined behavior—BeePL inserts semantics-preserving runtime checks during compilation. We develop a verified compilation strategy that extends CompCert to generate BPF bytecode from BeePL programs, establishing a principled foundation for an end-to-end verifiable toolchain for safe kernel extensions.

\end{abstract}

\maketitle

\input{sections/introduction}

\input{sections/framework}

\input{sections/formal_proofs}

\input{sections/relatedwork}

\bibliographystyle{ACM-Reference-Format}
\bibliography{references}

\input{sections/appendix}

\end{document}

%% file: preamble.tex
\usepackage{drawstack}
\usepackage{mathtools}
\usepackage{siunitx}
\usepackage{float} 
\usepackage{booktabs}
\usepackage{tikz}
\usepackage{multicol}
\usepackage[linguistics]{forest}
\usetikzlibrary{fit}
\usepackage{caption}
\usepackage{subcaption}
\usepackage{comment}
\usepackage{mdframed}
\usetikzlibrary{cd}
\usetikzlibrary{arrows,shapes,positioning,shadows,trees}
\usepackage{proof}
\usepackage{amsmath}
\usepackage{amsthm}
\usepackage{paralist,enumitem}
\usepackage{mathpartir}
\usepackage{coqdoc}
\usepackage{hyperref}
\usepackage{url}
\usepackage{lineno}
\usepackage{color}
\usepackage{diagbox}
\usepackage{pifont}
\usepackage{fancybox}
\usepackage{nccmath}
\usepackage{listings}
\usepackage{wrapfig}
\usepackage{wasysym}
\usepackage{colortbl} 
\usepackage{stmaryrd}
\usepackage{amsthm}
\usepackage{makecell}
\usepackage{adjustbox}
\usepackage{centernot}
\usepackage{pifont}
\usepackage{todonotes}

\lstdefinestyle{beepl}{
  basicstyle=\scriptsize\ttfamily,
  commentstyle=\color{green!50!black},
  keywordstyle=\color{blue},
  stringstyle=\color{black},
  showstringspaces=false,
  breaklines=true,
  frame=single,
  language=C,
  morekeywords={let, in, match, with, function, fun, type, module, struct, sig, open, rec, and, if, then, else, begin, end, try, raise, ref, for},
  literate=
    {:=}{{\textcolor{blue}{:=}}}2
    {!}{{\textcolor{blue}{!}}}1
    {!=}{{\textcolor{black}{!=}}}2
    {\#}{{\textcolor{black}{\#}}}2
    {alloc}{{\textcolor{red}{alloc}}}5
    {read}{{\textcolor{red}{read}}}4
    {write}{{\textcolor{red}{write}}}5
    {int}{{\textcolor{blue}{int}}}3
    {uint16}{{\textcolor{blue}{uint16}}}3
    {\_\_u16}{{\textcolor{blue}{\_\_u16}}}3
    {long}{{\textcolor{blue}{long}}}6
    {option}{{\textcolor{blue}{option}}}7
}

\setlist{left=.618\parindent .. 1.618\parindent}

\newif\iffullversion
\fullversionfalse

\newif\ifdiff
\difffalse

\hypersetup{
	colorlinks,
	linkcolor={red!70!black},
	citecolor={red!70!black},
	urlcolor={blue!70!black}
}

\definecolor{ao(english)}{rgb}{0.0, 0.5, 0.0}

%% file: def.tex
\newcommand*\figref[1]{Figure~\ref{fig:#1}}
\newcommand*\secref[1]{Section~\ref{sec:#1}}

\newcommand{\type}{\kw{type}}
\newcommand{\typev}{\kw{{prim_{\tau}}}}
\newcommand{\typeb}{\kw{basic_{\tau}}}
\newcommand{\kint}{\kw{int}}
\newcommand{\klong}{\kw{long}}
\newcommand{\tbool}{\kw{bool}}
\newcommand{\tint}[3]{\kint \ #1 \ #2 \ #3}
\newcommand{\tlong}[2]{\klong \ #1 \ #2}

\newcommand{\kunsigned}{\kw{unsigned}}
\newcommand{\ksigned}{\kw{signed}}

\newcommand{\keight}{\kw{i8}}
\newcommand{\ksixteen}{\kw{i16}}
\newcommand{\kthirtytwo}{\kw{i32}}
\newcommand{\tpointer}{\kw{ptr_\tau}}
\newcommand{\tunit}{\kw{unit_\tau}}
\newcommand{\tbstruct}[2]{\kw{struct} \ #1 \ #2}
\newcommand{\tbarray}[3]{#1 \ [#2] \ #3}

\newcommand{\tref}[1]{#1\rptr}
\newcommand{\toption}[1]{\kw{option}(#1)}
\newcommand{\tpstruct}[2]{\kw{struct} \ #1 \ #2*}
\newcommand{\tfunction}[3]{\overline{#1} \rightarrow_{#2} \ #3}
\newcommand{\tfun}[3]{{#1} \rightarrow_{#2} \ #3}
\newcommand{\tpfunction}[3]{*\overline{#1} \rightarrow_{#2} \ #3}
\newcommand{\bytes}{\kw{bytes}}
\newcommand{\effect}{\kw{effect}}
\newcommand{\efdivergence}{\kw{divergence}}
\newcommand{\efread}[1]{\kw{read} \ #1}
\newcommand{\efwrite}[1]{\kw{write} \ #1}
\newcommand{\efalloc}[1]{\kw{alloc} \ #1}
\newcommand{\efio}{\kw{io}}

\newcommand{\expr}{\mathsf{expr}}
\newcommand{\val}{\mathsf{Value}}
\newcommand{\const}{\mathsf{Const}}

\newcommand{\kw}[1]{{\ensuremath{\mathsf{#1}}}}

\newcommand{\kif}{\kw{if}}
\newcommand{\klet}{\kw{let}}
\newcommand{\kin}{\kw{in}}

\newcommand{\kthen}{\kw{then}}
\newcommand{\kelse}{\kw{else}}
\newcommand{\kunit}{\kw{unit}}

\newcommand{\kpop}{\kw{op}_p}
\newcommand{\kuop}{\kw{uop}}
\newcommand{\kbop}{\kw{bop}}

\newcommand{\kmass}{:=}
\newcommand{\kderef}{!}
\newcommand{\kref}{\kw{ref}}

\newcommand{\ecall}{\kw{ecall}}

\newcommand{\kmatch}{\kw{match}}

\newcommand{\kfor}{\kw{for}}
\newcommand{\kpat}{\kw{Pattern}}
\newcommand{\kpnone}{\kw{pnone}}
\newcommand{\kpsome}[1]{\kw{psome} \ #1}
\newcommand{\kpbytes}[4]{(#1,#2): \overline{(#3,#4)}}

\newcommand{\prim}[2]{{#1}(\overline{#2})}
\newcommand{\cond}[3]{\kif\ {#1}\ \kthen\ {#2}\ \kelse\ {#3}}

\newcommand{\app}[2]{#1(\overline{#2})}
 
\newcommand{\elet}[4]{\klet\ #1 : #2 = #3 \ \kin \ {#4}}

\newcommand{\loc}[2]{(#1, #2)}

\newcommand{\ef}[2]{\ecall\ #1\ #2}
\newcommand{\match}[3]{\kmatch\ #1\ \kw{with}\ \overline{{#2} \rightarrow #3}}
\newcommand{\structi}[3]{#1 | \overline{({#2}, {#3})}}
\newcommand{\structf}[2]{#1.#2}
\newcommand{\fore}[4]{\kfor\ (#1 \ldots #2, #3) \ #4}

\newcommand{\fundecl}{\{\kw{sec}: \kw{option \ string}; \kw{rt}:\tau;\ \kw{ef}:\eta; \kw{cc} : \kw{conv};\ \kw{args}:\overline{x,\tau};\ \kw{vars}:\overline{y,\tau};\kw{body}:e; \kw{flag}:\kw{bool}\}}
\newcommand{\globalvar}{\{\kw{gvar} : (x : \tau, v);  \kw{sec} : \mathsf{option \ string}\}}
\newcommand{\program}{\overline{\kw{id, decl}}}

\newcommand{\bstate}[3]{\prec#1,#2,#3\succ}
\newcommand{\eval}[4]{#1,#2 \rightarrow #3,#4}
\newcommand{\evals}[4]{#1,#2 \rightarrowtail #3,#4}

\newcommand{\meval}[5]{#1,#2 \twoheadrightarrow_{#5} #3,#4}
\newcommand{\mevals}[5]{#1,#2 \rightsquigarrow_{#5} #3,#4}

\newcommand{\typerule}[7]{#1,#2,#3,#4 \vdash #5 : #6, #7}
\newcommand{\typerules}[7]{#1,#2,#3,#4 \Vdash #5 : #6, #7}

\newcommand{\rptr}{\scalebox{0.6}{$\star$}}

\newcommand{\gtick}{\color{green}\checkmark}
\newcommand{\rcross}{\color{red}\times}

\DeclareUnicodeCharacter{00B7}{\cdot}
\DeclareUnicodeCharacter{00D7}{\ifmmode\times\else\texttimes\fi}
\DeclareUnicodeCharacter{00F7}{\ifmmode\div\else\textdiv\fi}
\DeclareUnicodeCharacter{03B3}{\ensuremath{\gamma}}
\DeclareUnicodeCharacter{03BA}{\ensuremath{\kappa}}
\DeclareUnicodeCharacter{03C4}{\ensuremath{\tau}}
\DeclareUnicodeCharacter{03C6}{\ensuremath{\varphi}}
\DeclareUnicodeCharacter{2081}{\ensuremath{{}_1}}
\DeclareUnicodeCharacter{2082}{\ensuremath{{}_2}}
\DeclareUnicodeCharacter{2113}{\ensuremath{\ell}}
\DeclareUnicodeCharacter{21A6}{\ensuremath{\mapsto}}
\DeclareUnicodeCharacter{21D2}{\implies}
\DeclareUnicodeCharacter{21D3}{\Downarrow}
\DeclareUnicodeCharacter{2200}{\forall}
\DeclareUnicodeCharacter{2203}{\exists}
\DeclareUnicodeCharacter{2208}{\in}
\DeclareUnicodeCharacter{2218}{\ensuremath{\circ}}
\DeclareUnicodeCharacter{2227}{\mathop{\land}}
\DeclareUnicodeCharacter{222A}{\cup}
\DeclareUnicodeCharacter{2260}{\neq}
\DeclareUnicodeCharacter{2264}{\leqslant}
\DeclareUnicodeCharacter{2291}{\sqsubseteq}

\def\substitute#1by#2{[#2/#1]}

\newcommand*\vtrue{\ensuremath{{t\kern-.4ex t}}}
\newcommand*\vfalse{\ensuremath{{f\kern-.7ex f}}}

\newtheorem{lemma}{Lemma}

\theoremstyle{definition}

%% file: sections/introduction.tex
\section{Introduction}
eBPF \cite{ebpfintro} is a revolutionary technology that enables safe, efficient, and dynamic instrumentation of the Linux kernel without requiring modifications to the kernel source code or the use of loadable kernel modules. It empowers developers to write user-defined programs that execute in kernel space, extending the kernel’s capabilities in areas such as performance monitoring, network packet filtering, system call tracing, and malware inspection. This flexibility has made eBPF a foundational component in modern observability and security tooling. It underpins a wide range of high-performance applications~\cite{ebpfcasestudy}, including networking frameworks (such as Meta’s Katran~\cite{facebook} and Google’s Cilium~\cite{google}), security enforcement platforms (like Cloudflare~\cite{cloudflare} and Falco~\cite{Falco}), observability solutions (like Netflix~\cite{netflix}), and performance monitoring tools (deployed at AWS~\cite{aws}). Initially, eBPF programs were written in a restricted subset of C, but with increasing attention from both industry and academia, support for writing eBPF programs in safer and more expressive languages like Rust \cite{Aya} has gained traction. To ensure safety and prevent buggy or malicious code from compromising the kernel, the eBPF toolchain includes a critical static analysis component known as the eBPF verifier. While there exists an informal specification for what constitutes a safe eBPF program—covering properties like memory safety, bounded execution, control flow integrity, valid helper usage, and safe pointer manipulation—this specification is neither formalized nor explicitly documented. It is understood implicitly within the eBPF community. Currently, the verifier enforces these checks, and only programs that pass are allowed to run in the kernel. Prior research and CVEs (Common Vulnerabilities and Exposures) \cite{eBPFCVElist} show that it can both accept unsafe programs and reject safe ones, revealing inconsistencies in its guarantees. These issues highlight the need for a more rigorous, formal foundation for verifying eBPF programs.


We formalize a set of key safety properties for eBPF programs and provide a trustworthy foundation for building compliant kernel extensions. Our language enforces many of these properties statically via its type system, while others—like dynamic memory safety—are ensured through semantics-preserving transformations by a verified compiler. The properties we capture include absence of uninitialized memory access, null dereferences, undefined behavior, guaranteed termination, memory bounds safety, and adherence to eBPF-specific typing rules.

\subsection{Problem Statement}
The existing eBPF ecosystem lacks a formally verified approach to guarantee the safety and security of programs loaded into the kernel. Here, we outline the key problem statements we aim to address. A detailed discussion of related efforts and prior work in these directions is provided in \secref{relatedwork}.
\paragraph{The eBPF verifier is limited and unsound.}
One of the major drawback's of the eBPF verifier is its unsoundness. Several reports~\cite{10278676}\cite{10.1145/3341301.3359641}\cite{10.1007/978-3-031-37709-9_12} have highlighted logical bugs in the verifier that lead to false negatives, allowing unsafe programs to pass verification. This can have severe consequences, including kernel memory corruption~\cite{CVE202339191} and security vulnerabilities~\cite{CVE20214159}~\cite{10.1145/3609510.3609822}. The complexity of a program depends not only on the number of instructions but also on the number of branching instructions in the program. For each branch, the verifier needs to fork the state to produce two states that need to be analyzed separately. The verifier has limited resources and can suffer from state explosion.
\paragraph{The absence of a type-safe language for eBPF development.}
Several approaches have explored using strongly typed languages for writing eBPF and user-space programs, including Rust~\cite{rust, Aya}, Go~\cite{go, Cilium}, Elixir~\cite{10.1145/3696443.3708923} and Lua~\cite{lua}. These languages aim to provide a safer and more accessible programming experience, either through built-in safety features or by relying on the verifier for enforcement. 
\begin{table}[h!]
\centering
\resizebox{0.55\textwidth}{!}{\begin{tabular}{|c|c|c|c|c|c|c|c|c|}
\hline
\rowcolor{Lavender} 
\textbf{eBPF Properties} & \multicolumn{7}{|c|}{\cellcolor[HTML]{CCE5FF}\textbf{Languages}}  \\ \hline
\rowcolor[HTML]{E5E5E5} 
 \textbf{Checked at compile-time} & \textbf{C} & \textbf{Rust} & \textbf{Elixir} & \textbf{Go} & \textbf{Python}  & \textbf{Lua} & \textbf{BeePL}   \\ \hline
Type safety     & $\LEFTcircle$  & $\gtick$ & $\LEFTcircle$ & $\gtick$   & $\rcross$   & $\rcross$ & $\gtick$            \\ \hline
Termination   & $\rcross$   &  $\rcross$ &  $\gtick$ & $\rcross$  & $\rcross$  & $\rcross$ & $\gtick$   \\ \hline
Stack usage    & $\rcross$   &  $\LEFTcircle$ & $\gtick$ & $\rcross$  & $\rcross$   & $\rcross$ & $\rcross$   \\ \hline
Register state safety & $\rcross$   &  $\gtick$ & $\LEFTcircle$ &  $\rcross$  & $\rcross$  & $\rcross$ & $\gtick$   \\ \hline
Memory safety   & $\rcross$   &  $\gtick$ & $\LEFTcircle$ & $\rcross$  & $\rcross$   & $\rcross$ & $\gtick$   \\ \hline
Helper functions safety    & $\rcross$    &  $\LEFTcircle$  & $\LEFTcircle$ &  $\rcross$  & $\rcross$   & $\rcross$ & $\LEFTcircle$    \\ \hline
Absence of undefined behavior   & $\rcross$   &  $\gtick$ & $\LEFTcircle$ &  $\rcross$  & $\rcross$   & $\rcross$ & $\gtick$   \\ \hline
Certified compilation  & $\LEFTcircle$   &  $\rcross$ & $\rcross$ &  $\rcross$  & $\rcross$  & $\rcross$ & $\LEFTcircle$   \\ \hline
\end{tabular}}
\caption{Languages and their support for properties of interest}
    \vspace{-1.5em}
\label{fig:languages-state-of-art}
\end{table}

Table~\ref{fig:languages-state-of-art} compares several languages—C, Rust, Elixir, and Lua for eBPF programs, and Go and Python for user-space tools—based on their ability to statically reason about eBPF properties. Symbols denote whether a language cannot ($\rcross$), can ($\gtick$), or partially ($\LEFTcircle$) verify a given property statically. C remains dominant for eBPF development due to its performance and integration with Clang/LLVM, but its type system cannot reason about termination, stack usage, uninitialized registers, or memory safety, and relies heavily on the eBPF verifier. Rust offers memory safety, bounds checking, and safe map access, but lacks termination analysis and certified BPF compilation. Honey Potion, an eBPF backend for Elixir, uses language features and code generation to reason about termination and stack use, but lacks formal guarantees. Lua, Go, and Python provide high-level abstractions but offer limited safety guarantees. The core challenge is the lack of a type system and a formal approach that fully captures eBPF safety properties—a gap BeePL aims to fill.

\paragraph{Lack of certified compilation.}
Ensuring eBPF safety properties at the source level alone is insufficient, as compiler optimizations can inadvertently break these guarantees when generating BPF bytecode. For instance, while Rust can statically enforce several eBPF safety properties (as shown in Table~\ref{fig:languages-state-of-art}), there is no assurance that these properties are preserved in the compiled BPF bytecode due to the lack of a formally verified, property-preserving compiler. The work in~\cite{10.1145/3691621.3694943} highlights several attack patterns that exploit the unsoundness of Rust’s compiler, demonstrating how ``safe" Rust can fail to uphold critical safety properties. Their analysis reveals vulnerabilities in Rust’s borrow checker, type system, and lifetime inference, which can lead to violations of its safety guarantees. 

\subsection{Contributions}
We present an end-to-end framework for building safe and secure eBPF programs using a strong type system and certified compilation. Our toolchain automates safety enforcement through static checks and runtime-preserving transformations, eliminating the need for manual formal reasoning. The framework includes BeePL—a verification-friendly language with a rich type system and a compiler that translates BeePL to C. The generated C code is subsequently compiled to BPF bytecode using CompCert. We extend CompCert with a new backend that targets eBPF bytecode, and this backend is formally verified in the Rocq theorem prover to ensure preservation of functional correctness. BeePL enforces memory safety by disallowing unsafe pointer operations and requiring null checks via option types and pattern matching. It supports structured memory access and guarantees no out-of-bounds errors through safe byte-pattern matching. Control-flow safety is ensured by only allowing terminating, bounded for-loops. By shifting safety checks to compile time and using a verified compiler, BeePL reliably produces correct C and BPF bytecode.

\textit{Contributions in nutshell}
\begin{itemize}
    \item A language called BeePL, featuring a type system that reasons about side effects, unsafe pointer manipulations, termination, and type signatures of helper functions.
    \item A compiler that translates BeePL to C while automatically enforcing safety properties through program transformations, without requiring manual checks.
    \item A CompCert extension for BPF bytecode generation, with mechanized proofs in Rocq ensuring preservation of functional correctness.
    \item Formal pencil-and-paper proofs establishing the \textit{type soundness} of the BeePL type system.
    \item BeePL enforces a set of eBPF safety properties, supported by formal pencil-and-paper proofs.
    \item A study illustrating how BeePL can eliminate several previously reported eBPF vulnerabilities.
\end{itemize}

\subsection{Motivating examples}
In this section, we present some motivating examples showcasing the features of BeePL language and its compiler. We examine a series of well-documented vulnerabilities and bug reports related to the eBPF verifier, many of which have led to critical security issues and assigned CVEs. We revisit some of these vulnerabilities and demonstrate how BeePL, by design, avoids these pitfalls through its strong type system and precise program transformations.

CVE-2021-3600 \cite{CVE20213600} highlights a critical flaw in the eBPF verifier related to the incorrect handling of 32-bit division and modulo operations. Specifically, the verifier fails to properly account for the truncation of the source register when performing 32-bit div or mod operations, which can result in unexpected values being propagated. This misbehavior can lead to out-of-bounds memory accesses—both reads and writes—posing a serious threat to system integrity.  \figref{ebpfprogc1} shows the C representation of the eBPF bytecode \cite{CVE20213600Blog} corresponding to the CVE-2021-3600 \cite{CVE20213600}, where a division or modulo operation is performed with a divisor that may become zero due to truncation. In this example, the check in line 6 guards the 64-bit register \texttt{r0}, but the actual operation in line 4 uses its truncated 32-bit counterpart \texttt{w0}. The verifier incorrectly assumes that the non-zero value of \texttt{r0} ensures safety, but fails to account for the fact that truncating \texttt{r0} to \texttt{w0} yields zero, leading to a divide-by-zero error—an instance of undefined behavior. 
\begin{figure}[h]
\centering
\begin{minipage}[t]{0.45\textwidth}
\begin{lstlisting}[basicstyle=\scriptsize]
SEC("xdp")
int prog1(struct xdp_md *ctx) {
    long r0 = 0x100000000;    // 2^32 - 64 bits can fit, but not 32
    int w0 = (int)r0;   // truncates to 0 (lower 32 bits)
    int w1 = 3;
    if (r0 != 0) { w1 = w1 % w0; } // Dangerous: w0 == 0 : division/modulo-by-zero
    return XDP_PASS;
}
char LICENSE[] SEC("license") = "GPL";
\end{lstlisting}
\caption{eBPF program in C for CVE-2021-3600 (allows division or modulus by zero)}
\label{fig:ebpfprogc1}
\end{minipage}
\hfill
\begin{minipage}[t]{0.50\textwidth}
\begin{lstlisting}[basicstyle=\scriptsize,mathescape=true]
#section "xdp"
fun bprog1(option(struct xdp_md$\rptr$) ctx) : int {
    let r0 : long = 0x100000000 in 
        let w0 : int = (int)r0 in 
            let w1 : int = 3 in 
                if (r0 != 0) then let w1 = w1 % w0 in XDP_PASS else XDP_PASS
}
char LICENSE[] #section "license" = "GPL";
// Division by zero eliminated due to program transformation done by BeePL compiler
\end{lstlisting}
\caption{BeePL program for CVE-2021-3600}
\label{fig:ebpfprogb1}
\end{minipage}
\end{figure}

In contrast, \figref{ebpfprogb1} presents the same logic written in BeePL. In BeePL, the compiler explicitly tracks truncation and enforces safety through a program transformation (done by BeePL compiler) that rewrites operations like \texttt{w1 \% w0} into a guarded conditional form: \lstinline[basicstyle=\normalsize\ttfamily]|if w0 == 0 then 0 else w1 

CVE-2020-8835~\cite{CVE20208835} reports a vulnerability where the eBPF verifier failed to properly track 32-bit register values used as shift operands. For example, an instruction like \(r = r >> r\), with \(r\) initialized to 808464432 (a \texttt{long}), may result in undefined behavior if the shift amount is \(\ge\)64. The verifier should enforce an upper bound (i.e., \(\le\)63 for 64-bit values), but in this case~\cite{CVE20208835Blog}, it allowed the program to pass despite the unsafe operation. BeePL prevents such undefined behaviors through compile-time transformations. In particular, for shift operations, the BeePL compiler inserts explicit bounds checks to eliminate unsafe execution paths.

CVE-2022-23222~\cite{CVE202223222} exposes a critical flaw in the eBPF verifier’s handling of $\kw{*\_OR\_NULL}$ pointers. These may be NULL or valid, and the verifier must ensure safety before any arithmetic. However, due to a missing sanity check in $\kw{adjust\_ptr\_min\_max\_vals}$ within \texttt{kernel/bpf/verifier.c}, it incorrectly allows pointer arithmetic on possibly NULL values, leading to out-of-bounds access and potential privilege escalation.

\begin{figure}[h]
\centering
\begin{minipage}[t]{0.40\textwidth}
\begin{lstlisting}[basicstyle=\scriptsize]
SEC("xdp")
struct { ... } counter_table SEC(".maps");
int cprog2(struct xdp_md *ctx) {
     long uid; 
     long *p; 
     uid = bpf_get_current_uid_gid() & 0xFFFFFFFF;  
     p = bpf_map_lookup_elem(&counter_table, (&uid)); 
     return (int)*p // Dangerous: dereferencing may be NULL
}
char LICENSE[] SEC("license") = "GPL";
\end{lstlisting}
\caption{eBPF program in C for CVE-2022-23222 (may lead to null pointer dereference)}
    \vspace{-1em}
\label{fig:ebpfprogc2}
\end{minipage}
\hfill
\begin{minipage}[t]{0.55\textwidth}
\begin{lstlisting}[basicstyle=\scriptsize,mathescape=true]
#section "xdp"
struct { ... } counter_table #section ".maps"
fun bprog2(option(struct xdp_md$\rptr$) ctx) : int {
    let  uid : long$\rptr$ = ref(0) in 
        let _ = uid := bpf_get_current_uid_gid() & 0xFFFFFFFF in 
            let p : option(long*) = bpf_map_lookup_elem(counter_table, uid) in 
                (int)!p
}
char LICENSE[] #section "license" = "GPL";
// Rejected by the BeePL type system as dereferencing(!) an option type is not allowed 
\end{lstlisting}
\caption{BeePL program for CVE-2022-23222 (safe BeePL is present in \figref{ebpfprogb3})}
\label{fig:ebpfprogb2}
    \vspace{-1em}
\end{minipage}
\end{figure}
Figure~\ref{fig:ebpfprogc2} shows an eBPF program vulnerable to null pointer dereference. The helper function $\kw{bpf\_map\_lookup\_elem}$ returns a pointer to a map value or NULL if the key is not found. However, the returned pointer p is dereferenced without any null-check. This can lead to undefined behavior or a crash if p is NULL. This is the core issue behind CVE-2022-23222~\cite{CVE202223222Blog} missing safety check before pointer dereference. Figure~\ref{fig:ebpfprogb2} shows the same logic written in BeePL. Here, \texttt{bpf\_map\_lookup\_elem} returns \texttt{option(long$\rptr$)}, making the possibility of a NULL value explicit in the type. The dereference \texttt{!p} is rejected by the type checker, as BeePL requires such operations to be guarded by pattern matching. Unlike C, all pointer-returning helper functions in BeePL return an option type, enabling null safety at the language level. During compilation, option types are translated back to regular pointers in C, with pattern matches ensuring that null dereferences cannot occur. This approach eliminates the need for manual null checks and ensures that programs like the one in Figure~\ref{fig:ebpfprogb2} are rejected by the type checker, preventing verifier bugs.
\begin{figure}[h]
\centering
\begin{minipage}[t]{0.50\textwidth}
\begin{lstlisting}[basicstyle=\scriptsize,mathescape=true]
#section "xdp"
struct { ... } counter_table #section ".maps"
fun bprog3(option(struct xdp_md$\rptr$) ctx) : int {
    let  uid : long$\rptr$ = ref(0) in 
        let _ = uid := bpf_get_current_uid_gid() & 0xFFFFFFFF in 
            let p : option(long*) = bpf_map_lookup_elem(counter_table, uid) in 
                match p with 
                    | pnone => -1
                    | psome p' => (int)!p'
}
char LICENSE[] #section "license" = "GPL";
\end{lstlisting}
\caption{Safe BeePL program for CVE-2022-23222}
    \vspace{-1em}
\label{fig:ebpfprogb3}
\end{minipage}
\hfill
\begin{minipage}[t]{0.45\textwidth}
\begin{lstlisting}[basicstyle=\scriptsize,mathescape=true]
SEC("xdp")
struct { ... } counter_table #section ".maps";
int cprog3(struct xdp_md *ctx) {
     long uid; 
     long *p; 
     long *p';
     uid = bpf_get_current_uid_gid() & 0xFFFFFFFF;  
     p = bpf_map_lookup_elem(&counter_table, (&uid)); 
     if (p == (long *) 0) { return -1; } 
        else { p' = p;return (unsigned int) *p'; }
}
char LICENSE[] SEC("license") = "GPL";
\end{lstlisting}
\caption{Safe C program generated by BeePL compiler for CVE-2022-23222}
    \vspace{-1em}
\label{fig:ebpfprogc3}
\end{minipage}
\end{figure}

Several CVEs involving NULL pointer dereference—CVE-2022-0433~\cite{CVE20220433}, CVE-2022-3606~\cite{CVE20223606}, CVE-2025-21852~\cite{CVE202521852}, and CVE-2024-38566~\cite{CVE202438566}—can be eliminated using BeePL’s typechecker and program transformations performed by the compiler. \figref{ebpfprogb3} shows the safe BeePL version of the unsafe program from \figref{ebpfprogb2}. Here, the possible null dereference is handled via a \texttt{match} expression (lines 6–8) on the \texttt{option} type. If \texttt{p} is \texttt{Pnone}, the function returns $-1$; otherwise, it safely dereferences the pointer inside \texttt{Psome}. This resolves the type checker error by ensuring dereferencing applies only to valid pointers. Additionally, no null check is needed on \texttt{uid}, as \texttt{ref} guarantees non-null pointers. \figref{ebpfprogc3} shows the generated C code, where an explicit null check appears on line 9, inserted automatically by the BeePL compiler. This demonstrates how BeePL enforces safe pointer usage and removes the burden of manual null checks.

\begin{table}[h]
\centering
\scriptsize
\resizebox{\linewidth}{!}{%
\begin{tabular}{|c|p{7.5cm}|c|p{4.6cm}|}
\hline
\textbf{CVE-ID} & \textbf{Vulnerability} & \textbf{BeePL} & \textbf{BeePL features} \\
\hline
CVE-2021-3600 & Division/modulo with zero due to truncation of 64-bit to 32-bit & \ding{55} & \makecell{Program transformation} \\
\hline
CVE-2020-8835 & Out-of-bounds reads and writes due to improper bound calculation in shift operator & \ding{55} & \makecell{Type safety + \\ Program transformation} \\
\hline
CVE-2022-3606 & Manipulation in the program operations leads to null pointer dereference & \ding{55} & \makecell{Type safety + \\ Program transformation} \\
\hline
\makecell{CVE-2022-0433, \\ CVE-2022-23222 \\ CVE-2025-21852, \\ CVE-2024-38566} & Null pointer dereference flaw & \ding{55} & \makecell{Type safety + \\ Program transformation} \\
\hline
CVE-2021-3490 & Improper ALU32 Bounds Tracking & \ding{55} & \makecell{Type safety + \\No pointer arithmetic} \\
\hline
\makecell{CVE-2021-31440,\\
CVE-2020-27194} & Wrong assumption about bounds of 32-bit register values & \ding{55} & \makecell{Type safety + \\ Program transformation} \\
\hline
\end{tabular}%
}
\caption{CVE mitigated by BeePL's verified compilation}
\label{tab:beepl-cve}
\end{table}
We examined a collection of known CVEs~\cite{eBPFCVElist} related to the eBPF ecosystem to assess whether BeePL can eliminate them through its design and guarantees. Table~\ref{tab:beepl-cve} summarizes these findings: each CVE, its vulnerability type, a \ding{55} indicating vulnerabilities prevented by BeePL, and BeePL features (type system, verified transformations) responsible for prevention. This case study shows BeePL’s ability to eliminate critical bugs and highlights our goal to enable developers to safely write eBPF programs without relying on the kernel verifier.

%% file: sections/framework.tex
\section{Formally verified end-to-end framework}
\figref{beeplframework} presents an overview of our approach. We introduce BeePL, a programming language designed with the expressive power necessary to write realistic eBPF programs. BeePL supports multiple programming paradigms and features a type system that statically enforces key eBPF properties—including type safety, termination, side effects, and correct usage of helper function signatures—directly at the source level (see Table~\ref{fig:languages-state-of-art}). BeePL also employs a range of program transformations that eliminate undefined behavior and prevent out-of-bounds memory accesses. The framework also includes a compiler that translates BeePL source code to C, which is then compiled to BPF bytecode using CompCert~\cite{10.1007/s10817-009-9155-4}.
\begin{figure}[htbp]
  \centering
  \begin{subfigure}[t]{0.58\textwidth}
    \centering
    \includegraphics[width=\linewidth]{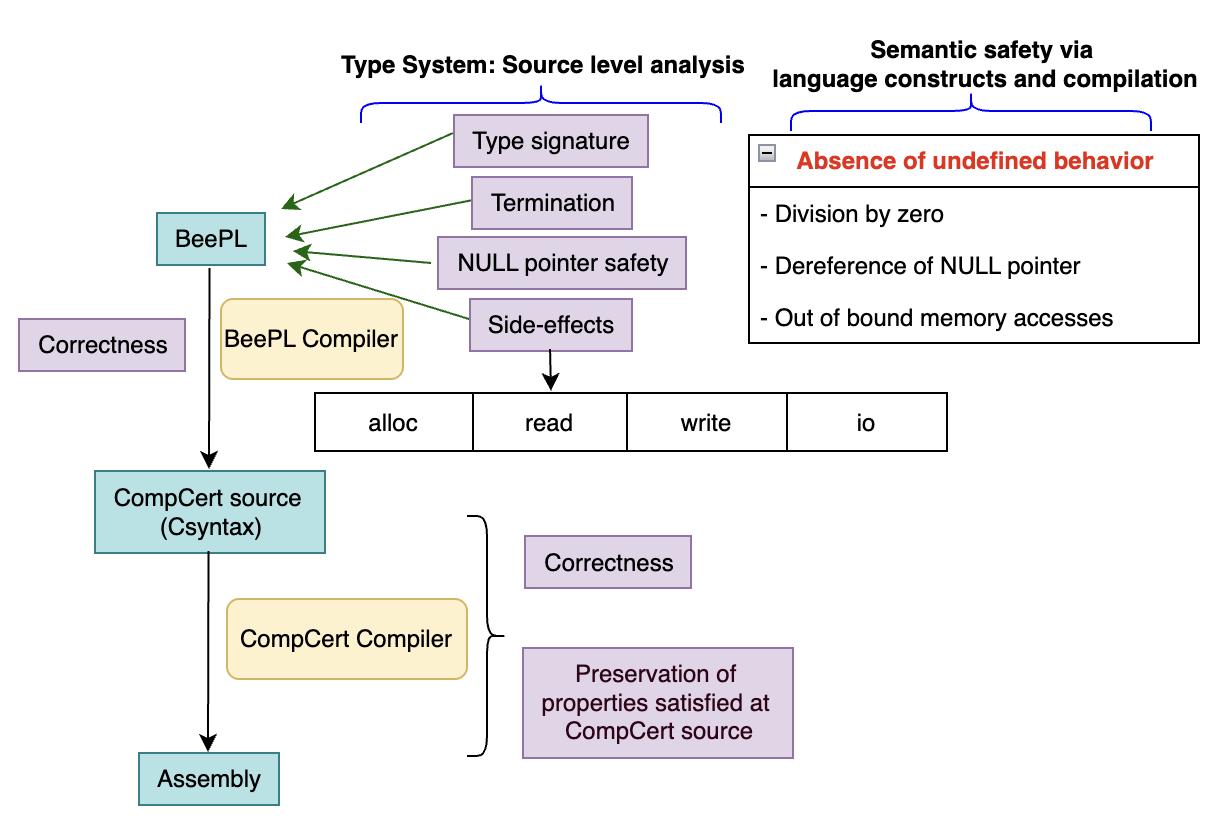}
    \caption{BeePL framework}
    \label{fig:beeplframework}
  \end{subfigure}
  \hfill
  \begin{subfigure}[t]{0.38\textwidth}
    \centering
    \includegraphics[width=\linewidth]{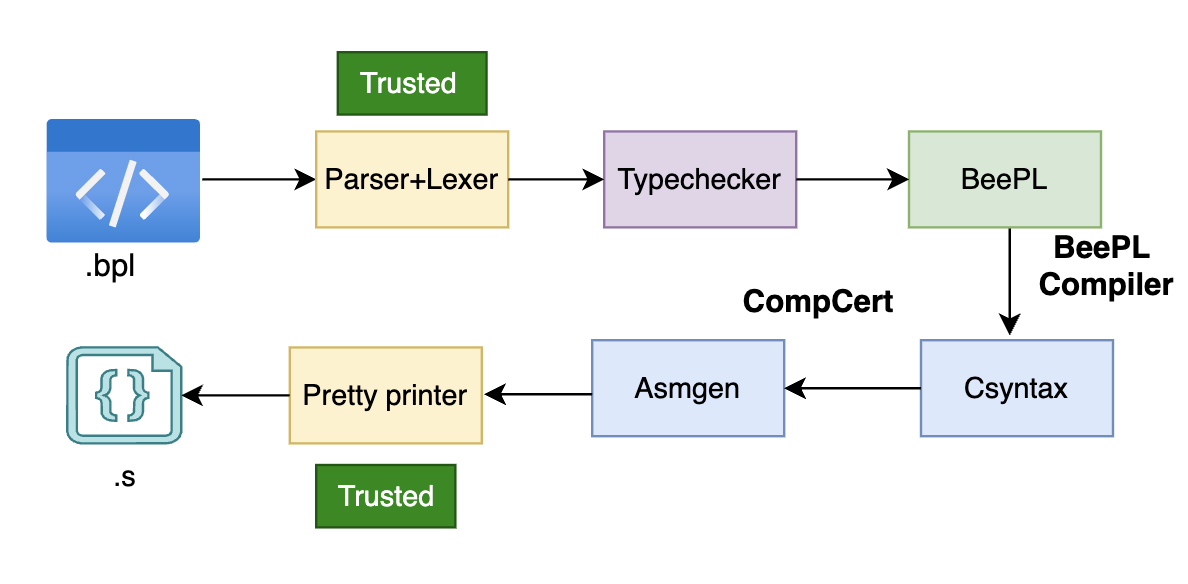}
    \caption{BeePL toolchain}
    \label{fig:beepltoolchain}
  \end{subfigure}
\end{figure}

A key component of our framework is the formally verified BeePL compiler (\figref{beepltoolchain}), which ensures that safe source programs compile to safe and secure target code. It parses BeePL source into an AST, type-checks it, and compiles it to C, which is then compiled to BPF bytecode using CompCert~\cite{10.1007/s10817-009-9155-4}, a verified C-to-assembly compiler. The BeePL-to-C compiler is implemented in the Rocq proof assistant, with passes extracted to OCaml. While the core compiler is verified (BeePL using paper-based formal proofs and CompCert (with extension of eBPF backend) in Rocq), the parser, lexer, and pretty-printer remain part of the trusted computing base. Our development contributes 33,635 lines of Rocq code as a CompCert extension.

\subsection{Formal specification for a set of eBPF properties}\label{sec:ebpf-spec}
The primary objective behind the design of BeePL and its type system is to enable formal reasoning about a well-defined set of safety properties critical to eBPF programs. We introduce a high-level formal specification for a set of eBPF safety properties, and in \secref{ebpf-proof}, we prove that all well-typed BeePL programs satisfy them. A program is considered $\kw{eBPF_{safe}}$ if it adheres to the following properties (text shown in green highlights which specific component of the framework is responsible for enforcing each property). The formal specification provides a high-level judgment involving a program $P$, an initial state $s$, a final state $s'$, a result value $v$, and a step count $n$. It states that if $P$ is well-typed and the initial state $s$ is well-formed, then the evaluation of $P$ transitions from $s$ to $s'$ in $n$ steps, producing the value $v$. The detailed definitions of well-typed programs and well-formed states are presented later in \secref{beepllang} and \secref{formalproofs}.

\begin{itemize}
\item \textbf{No uninitialized memory access ($\textmd{\textcolor{ForestGreen}{\textsf{Language feature}}})$:} All memory accesses involving pointers (represented as location $l$ with offset $o$ in memory $\Theta$ and ! representing dereference) generated within BeePL are guaranteed to reference initialized memory locations. 
\small
\begin{align*}
\forall P, s, s', v, n.\ &\vdash P : \text{well-typed} \wedge \text{well-formed}(s) \wedge \mevals{s}{P}{s'}{v}{n} \implies \\
&\forall (l, o),\ \text{if } !(l, o) \text{ evaluated during execution, then } \exists i.\ s_\Theta[l, o] = i
\end{align*}
\item \textbf{No null pointer dereferences ($\textmd{\textcolor{ForestGreen}{\textsf{Type system and runtime check}}}$):} Every memory access is ensured to be safe using the type system and runtime checks, ruling out null pointer dereferences.
\small
\begin{align*}
\forall P, s, s', v, n.\ &\vdash P : \text{well-typed} \wedge \text{well-formed}(s) \wedge \mevals{s}{P}{s'}{v}{n} \implies \\
&\forall (l, o),\ \text{if } !(l, o) \text{ evaluated during execution, then } l \neq \kw{NULL}
\end{align*}
\item \textbf{Memory bounds safety ($\textmd{\textcolor{ForestGreen}{\textsf{Program transformation and runtime check}}}$):} All memory accesses are verified to stay within bounds through program transformations and run-time checks, primarily for structured access patterns. $\Sigma$ maps memory locations to their types and is used to compute type-specific sizes.
\small
\begin{align*}
\forall P, s, s', v, n,\ l, o, \tau.\ 
&\vdash P : \text{well-typed} \wedge \text{well-formed}(s) \wedge \mevals{s}{P}{s'}{v}{n} \ \wedge \\
&\text{struct of type } \tau \text{ is read at } (l, o) \Rightarrow o + \kw{sizeof}(\tau) \leq \kw{sizeof}(\Sigma(l))
\end{align*}

\item \textbf{Compliance on helper function signatures and section attributes ($\textmd{\textcolor{ForestGreen}{\textsf{Type \ system}}}$):} All helper function calls respect their declared signatures, ensuring correct usage and argument types. BeePL’s type system statically enforces that all global variables and eBPF functions are annotated with the appropriate section attributes. This ensures conformance with the kernel's runtime placement and calling conventions.
\item \textbf{Termination ($\textmd{\textcolor{ForestGreen}{\textsf{Type system and language feature}}}$):} Every program written in BeePL terminates, ruling out infinite loops or non-terminating behavior (ensuring no divergence effect in the inferred effect $\eta$ for the program \textit{P}).
\small
\[
\forall P, s, \eta. \ \vdash P : \text{well-typed} \wedge \text{well-formed}(s) \implies \exists n, s', v. \ \mevals{s}{P}{s'}{v}{n} \wedge \kw{divergence} \notin \eta
\]

\item \textbf{Absence of undefined behavior ($\textmd{\textcolor{ForestGreen}{\textsf{{Type system and program transformation}}}}$):} 
Operations such as division, modulo, or arithmetic on unbounded integers are assigned well-defined semantics. When their operands may lead to undefined behavior—such as division by zero or overflow—the BeePL compiler inserts guards and rewrites them into safe expressions. Any result outside the expected semantics is modeled as a special case called $\kw{undef}$.
\small
\[
\forall P, s, s', v, n.\quad 
\vdash P : \text{well-typed} \wedge \text{well-formed}(s) \wedge 
\mevals{s}{P}{s'}{v}{n} \implies v \neq \kw{undef}
\]
\end{itemize}
\section{BeePL language and its type system}\label{sec:beepllang}
\paragraph{Expressions}\label{sec:expressions}
Expressions in BeePL (shown in \figref{BeePL}) are built from variables, constants, function applications, operators, and let-bindings. The language includes control-flow constructs such as conditionals and for-loops, supports external function calls, and allows initialization of user-defined structs and fixed-size arrays. Pattern matching over optional values and binary data enables expressive, type-safe deconstruction. The presence of the \texttt{option} type also facilitates pattern matching over nullable references via the \kw{match} construct, whose semantics and translation are described in \secref{semantics}. This mechanism ensures that all cases—both Psome and Pnone—are explicitly handled, further eliminating the possibility of undefined behavior due to null dereferencing. BeePL allows the use of safe abstractions over direct pointer manipulations. By utilising the pattern matching constructs, BeePL allows operations on packed data without relying on unsafe pointer arithmetic. For instance, rather than casting raw pointers and adding offsets—a common practice in low-level C code—BeePL code operates on $\kw{bytes}$. Pattern matching on $\kw{bytes}$ in BeePL allows direct extraction of structured data from binary sequences. This enables precise control over byte alignment, size, and type during parsing, without explicit pointer manipulation. Operators include standard unary and binary forms, as well as primitive operators for mutable effects (e.g., $\mathsf{ref}$, !, :=). A BeePL program consists of a sequence of declarations—functions (local or external) and global variables. Each function includes a signature specifying the return type, effects, parameters, local variables (with types), and a body. A special $\mathsf{flag}$ field marks whether the function is an eBPF program, enabling additional checks on section attributes and argument types to ensure consistency with the associated eBPF hook. Both function and global variable declarations may include an optional $\kw{sec}$ field for section attributes, which is $\kw{None}$ for regular functions.
\input{figures/language}
\paragraph{Types}\label{sec:types}
BeePL features a rich type system that encompasses a diverse set of type constructs: primitive types ($\tau_v$), pointer types ($\tau*$), effectful function types ($\tfunction{\tau}{\eta}{\tau}$), user-defined struct types annotated with attributes (omitted in paper but incorporated in Rocq development), array types (fixed-size, contiguous collections of elements of type $\tau$ and size $n$), a bytes type (representing raw byte sequences that encapsulate data for structured access), and the unit type (indicating the absence of meaningful data). The complete set of supported types is shown in \figref{Btypes}. Primitive types ($\tau_v$) capture low-level data representations, including booleans, integers parameterized by bitwidth, signedness, and attributes, as well as long integers with similar qualifiers. Size is defined as follows: $sz \in \kw{isize} ::= \keight \mid \ksixteen \mid \kthirtytwo$ and signedness as follows: $s \in \kw{sign} ::= \kunsigned\ \mid \ksigned\ $.

Pointer types ($\tau*$) denote memory references. In particular, $\tref{\tau_b}$ denotes a safe, statically trackable reference to an allocated value (either local or global). References are introduced internally through safe constructs like $\kw{ref}$. This design is crucial in eBPF settings to differentiate between pointers created inside the program and those originating externally. BeePL disallows returning the addresses of local variables and it is ensured by the BeePL type system (present in \figref{typesystem3}). $\toption {\tau*}$ represents references that may be null. This optional pointer type models a nullable reference—a pointer that either refers to a valid value of type $\tau$, or denotes the absence of a value (i.e., null). Such types are essential for capturing partiality, lookup failures, and conditionally available data in low-level systems programming, especially in languages like C where nullability is pervasive but not tracked by the type system.

\input{figures/types}

In BeePL, optional pointers play a crucial role in maintaining memory safety while enabling expressive pointer usage. BeePL distinguishes between two sources of pointers: (1) those produced by the \kw{ref} expression, which allocates memory and returns a valid address, and (2) those obtained from external functions. By design, the \kw{ref} construct is guaranteed to allocate and return a valid non-null pointer; thus, null-pointer checks are only necessary for external function calls that may return invalid or null pointers. Unlike traditional C-style programming where null checks such as if (p != NULL) are manually inserted and easy to forget, BeePL uses $\toption {\tau*}$ to enforce nullability at the type level. This approach enables the compiler—rather than the programmer—to insert necessary null-pointer checks during translation to C, ensuring that no unchecked dereference can occur. Moreover, the optional type clearly distinguishes between definitely valid pointers and potentially null pointers, making pointer usage both explicit and verifiable.

\textit{Extension of type system to capture effects.} 
BeePL extends its type system with an effect system (\(\eta\)) to track the computational behavior of expressions, including allocation, reading, writing, I/O, and divergence: $\eta \in \effect ::= \efdivergence \mid \efread \mid \efwrite \mid \efalloc \mid \efio.$ Each expression or function is annotated with a list of effects summarizing its possible side effects. This allows the compiler to enforce purity or track memory operations. Making effects explicit enables compile-time reasoning about stateful or external operations, improving security and predictability. The effect concatenation is written as \(\dotplus\). Effects serve as a static approximation of runtime behavior and align with operational semantics, where the tracked effects correspond to actual execution steps. Currently, BeePL tracks abstract effects like read, write, alloc, and io. In future work, we plan to extend this to richer properties—such as memory regions, access permissions, and data flow tags—to support more precise analysis and modular safety reasoning.

\subsection{Type System}
\figref{typesystem} presents the typing rules for BeePL. The typing judgment in BeePL is of the form $\typerule{\Gamma}{\Sigma}\Pi \Psi {e}{\tau}{\eta}$, where $\Gamma$ represents the typing environment that maps variables to their types, and $\Sigma$ denotes the store typing environment, which maps memory locations to their types. $\Pi$ denotes the struct environment, which encapsulates the type information associated with all struct declarations in the program. This environment is consulted during type checking to validate the use of struct expressions, such as field access and initialization, ensuring that the fields conform to their declared types and attributes. On the other hand, $\Psi$ represents the external function signature environment. It records the expected type signatures of all external functions invoked in the program, including eBPF helper functions. Instead of relying on standard BPF signature files, we construct $\Psi$ explicitly as a map within BeePL. This design choice is motivated by the fact that the types of eBPF helper functions, as interpreted in BeePL, may differ from their low-level C counterparts—particularly due to BeePL’s stricter typing and effect system. A more detailed treatment of this distinction is provided in the subsequent discussion of the typing rules. Extensions to the typing environment are denoted using the notation $x \rightarrow \tau, \Gamma$, where the first placeholder represents the new variable, the second denotes its associated type, and the third refers to the existing typing context being extended. Access to the typing environment is denoted by the notation $\_ \_= \_$, where the first placeholder represents the typing environment, the second indicates the variable or location being queried, and the third denotes the associated type. The operation yields either $\kw{Some}~\tau$ if the variable or location exists in the context with type $\tau$, or $\kw{None}$ if it is not present.

It is important to highlight that $\Sigma$ is not required in the executable type checker, as BeePL’s surface syntax does not permit programmers to write explicit memory locations. Concrete locations are introduced only during the evaluation of expressions as part of intermediate computation steps. Consequently, type information for these locations becomes relevant only during the reduction phase and not during the initial type-checking phase.
While the store typing environment is critical for proving meta-theoretical properties such as progress and preservation, it is not involved in the type-checking process for programs written by the user. BeePL disallows implicit casting, ensuring that the type of a location is determined by its initial value at the time of allocation. This guarantees that a location’s type remains fixed throughout execution, and thus, $\Sigma$—which maps each location to its type—is sufficient for type-checking locations during evaluation. 

For clarity of presentation, we omit details such as size and signedness for primitive types like $\kw{int}$ and $\kw{long}$ in the typing rules. The typing rule for variables retrieves the type of the variable $x$ from the typing context $\Gamma$ and produces an empty effect $\phi$. Typing rules for values are generally straightforward, with the exception of locations. In the case of locations, the type of a location $l$ is looked up in the store context $\Sigma$. The typing rule $\kw{TAPP}$ governs function application in BeePL, covering both user-defined and external functions. The judgment $\typerule{\Gamma}{\Sigma}{\Pi}{\Psi}{e}{\tau_e}{\eta}$ checks the function expression’s type, while $\typerules{\Gamma}{\Sigma}{\Pi}{\Psi}{\overline{e}}{\overline{\tau}}{\eta'}$ checks the argument expressions. A predicate $\kw{isExt}(e)$ distinguishes between internal and external functions. If $e$ refers to an external function, its type is retrieved from the external environment $\Psi$ and must match the inferred type $\tau_e$. In both cases, $e$ must have a function type $\tfunction{\tau}{\eta_e}{\tau_r}$, denoting its argument type, effect, and return type. The total effect of the application is the union of: $\eta$ (evaluating the function expression, reducing to its declaration), $\eta'$ (evaluating the arguments), and $\eta_e$ (the function’s declared effect). This rule ensures both type and effect safety, with explicit handling of external functions. While $\kw{isExt}$ and the lookup in $\Psi$ are typically resolved during elaboration or typechecking, we make them explicit in the rule to clarify how external calls are integrated into the type system. For example, in the BeePL function
\lstinline[basicstyle=\normalsize\ttfamily]|fun foo() : int, <alloc,read> { let x : int* = ref(2) in let r : int = !x + 1 in r;}|
the function’s effect captures both the allocation from \texttt{ref(2)} and the read from dereferencing \texttt{x}.

\input{figures/type_system}

The rule $\kw{TMASSGN}$ defines the typing rule for memory assignment expressions of the form $e_1 := e_2$. It ensures that $e_1$ is of a pointer type, which may originate either from within the BeePL program via the $\kw{ref}$ operator or from external sources, such as pointers returned by eBPF helper functions. \lstinline[basicstyle=\normalsize\ttfamily]|fun bar() : int, <alloc,read,write,read> { let x : int* = ref(2) in let _ : _ in x := !x + 1 in !x;}|. The effect annotation associated with the function $\kw{bar}$ captures the effects produced by its body, which includes the use of $\kw{ref}$, dereferencing ($!$), and memory assignment ($:=$). The rule enforces that the memory assignment expression has the unit type as its result, reflecting that the purpose of the expression is to produce a side effect rather than a value.

The typing rules for unary and binary operators ($\kw{TUOP}$ and $\kw{TBOP}$) are straightforward. They accumulate the effects produced during the evaluation of their operands and ensure that these operations are performed only on primitive data types. In BeePL, pointer arithmetic is disallowed, so such operators cannot be applied to pointer types. The typing rules for conditional and let-bindings are also pretty straightforward. Type checking of struct initialization ($\kw{TSINIT}$) ensures that all fields are initialized with values matching their declared types and appear in the correct order as defined in the struct environment $\Pi$ (ensured by the function $\kw{fields_\tau}$). For field access ($\kw{TFIELD}$), the type checker looks up the struct type in $\Pi$ to verify that the accessed field exists and returns the corresponding field's type. 

The type checking rule $\kw{TFOR}$ for the for loop in BeePL ensures that the loop always terminates by enforcing strict control over its structure and dependencies. The loop has the form $\fore {e_1} {e_2} d e$, where $e_1$ and $e_2$ represent the lower and upper bounds, $d$ is the fixed direction (up or down), and $e$ is the loop body. The typing rule requires that $e_1$ and $e_2$ are integer/long expressions, and that the loop body $e$ does not reference any variables used in $e_1$ or $e_2$. This is enforced by checking that the set of free variables in the body is disjoint from those in the bounds. As a result, the number of iterations is fixed and independent of the loop body, ensuring that the loop cannot introduce unbounded behavior. Consider the following code snippet: \lstinline[basicstyle=\normalsize\ttfamily]|fun loop() : int, <alloc,read,write,read> { let x : int* = ref(2) in let _  = for(1 ... 5, Up) {x := !x + 1} in !x;}|. The typing rule for the for loop captures the effect of the loop body—for example, <read, write, read> for the expression like x := !x + 1—and also tracks the free variables and constants involved (such as x, 1, and 5). This information is used to ensure that the loop bounds remain independent of the body, allowing the type system to statically guarantee the absence of the divergence effect. We conservatively include the per-iteration effect of the loop body once in the overall effect. This ensures soundness without requiring static knowledge of the number of iterations. 
\input{figures/type_system1}
\figref{typesystem1} presents additional typing rules emphasizing external calls and match expressions. The typing rule for pattern matching on an option expression in BeePL ensures that both the None and Some cases are handled explicitly and safely. Given an expression $m$ of type option $\tau_m*$, the rule typechecks two branches: one for \kw{Pnone}, which produces a result of type $\tau$, and one for $\mathsf{Psome} \ x$, where $x$ is bound to $\tau_m*$, and the body is typechecked under this binding to also yield type $\tau$. Both branches must return the same type, and the total effect includes the effect of evaluating $m$ as well as those of both branches. This rule enforces that a potentially null-like value must be explicitly unpacked before use, guaranteeing that no dereferencing occurs unless the value is known to be present, thus eliminating null pointer dereference errors at the type level.

For example, the eBPF helper function $\mathsf{bpf\_map\_lookup\_elem}$ returns a pointer to a value associated with a key in the eBPF map, but this pointer may be NULL if the key is not found. To model this uncertainty safely, BeePL wraps the return type of such functions in an option type. The type system then enforces, through the match expression and the typing rule $\kw{TMATCHO}$, that any use of such a pointer must first explicitly pattern match on the result. This guarantees that the Pnone case (corresponding to a null pointer) is handled before any dereferencing can occur, ensuring memory safety through compile-time checks. Revisiting the typing rule for function application, in the case of external function calls, BeePL uses the environment $\Psi$ to typecheck the call against the externally defined function signature. The environment $\Psi$ maps each external function to its type, including argument types, return type, and associated effects. For example, the signature of the helper function $\textcolor{blue}{\mathsf{bpf\_map\_lookup\_elem}}$ is defined in $\Psi$ as $\tfun {(\textcolor{blue}{{\toption {\kw{map*}}}}, \textcolor{blue}{\toption {\kw{long*}}})} {<\kw{{\textcolor{red}{read}}}, \kw{\textcolor{red}{io}}>} \textcolor{blue}{\toption{\kw{long*}}}$, indicating that it takes optional pointers to a map (a pointer to a struct, which creates map in eBPF) and a key, may perform \kw{read} and \kw{io} effects, and returns an optional pointer to a long integer. The $\kw{io}$ effect distinguishes operations that involve interaction with the external environment, such as calls to external functions or system-level resources. It allows the type system to clearly separate internal computation from externally observable behavior. Similarly, the eBPF helper function $\textcolor{blue}{\mathsf{bpf\_get\_current\_uid\_gid}}$ is assigned the type $\tfun{\textcolor{blue}{()}}{<\kw{\textcolor{red}{io}}>}{\textcolor{blue}{\kw{long}}}$ in BeePL. This signature indicates that the function does not perform any memory reads or writes, but it interacts with the external environment—such as querying kernel state—which results in an $\kw{io}$ effect.

The typing rule $\kw{TMATCHB}$ handles pattern matching on byte sequences when attempting to parse structured data, such as C-style headers, from raw bytes. In a pattern of the form $\kpbytes{x}{\tau_x}{y}{\tau_y}$, the byte buffer m (e.g., a field like data in \kw{xdp\_md} BPF structure, illustrated in the example present in \figref{ebpfprogb4}) is being parsed into a structure of type $\tau_x$ (e.g.,eth), which is referred to by the variable $x$. The result of this parsing is a sequence of bindings represented by the list of fields y with types $\tau_y$, which correspond to the fields of the struct type $\tau_x$. The rule first checks that $m$ has type $\kw{bytes}$, with effect $\eta_m$. Then, each branch expression is typechecked in an extended typing context where $x$ is bound to the struct type $\tau_x$ and the individual field bindings in $\overline{y : \tau_y}$ are available. All branch expressions must return the same result type $\tau_r$, and their effects are combined into $\eta_s$. This rule ensures that parsing raw bytes into a structured type is only allowed when it is type-safe, and that all resulting fields are well-typed and locally available for use in the corresponding match branch.

\input{figures/typing_fundecl}
\figref{typesystem3} defines the typing rules for BeePL programs. It enforces that function declarations must return non-pointer types—preventing functions from returning addresses—and ensures that each argument's type is compatible with the declared section attribute through the \textsf{section\_ok} predicate. The typing environment \(\Gamma\) initially starts empty and is extended with function arguments and local variables when checking the function body. The environment \(\Psi\) contains the signatures of external functions, while \(\Sigma\) and \(\Pi\) track the store and global contexts, respectively.
The predicate $\mathsf{section\_ok}(\tau, \kw{s})$ ensures that the type of a function argument is compatible with the section attribute of the function. If no section is specified ($\kw{s} = \texttt{None}$), no constraint is enforced. Otherwise, the predicate checks that specific argument types (e.g., $\kw{struct \ xdp\_md, struct \ \_\_sk\_buff)}$ are only used in functions annotated with the corresponding section (e.g., "xdp", "socket"). For all other types or sections, the predicate defaults to allowing the combination.

\subsection{Operational Semantics}\label{sec:semantics}
The small-step operational semantics of expressions are defined over a state $\bstate {\Delta} {\Omega} {\Theta}$, where $\Delta$ denotes global environment, comprising function definitions, global variable declarations, and struct type definitions, $\Omega$ is the variable environment, mapping local variables to memory locations and their respective types: ($\Omega := [x_1 \rightarrow (l_1, \tau_1), ..., x_n \rightarrow (l_n, \tau_n)])$, and $\Theta$ represents the memory mapping from locations to values($\Theta := [l_1 \rightarrow v_1, ..., l_n \rightarrow v_n]$). BeePL uses the memory model of CompCert~\cite{inproceedings}. We assume that $\mathsf{dom}(\Omega) \cap \mathsf{dom}(\Delta) = \Phi$, indicating that the domains of $\Omega$ and $\Delta$ are disjoint—that is, no variable name appears in both the local environment and the global environment. Access to environments is expressed using the notation $\_[\_]$, where the outer component refers to the environment or store, and the inner component denotes the key being looked up. Updates to environments are denoted using the notation $[\_ \rightarrow \_]$, which represents a mapping from a key to a new value. New allocations in the environment are represented using the notation $[|\_ \rightarrow \_|]$, which depends on the type of the variable being allocated. Type information is crucial for reserving the correct amount of memory. While types are omitted in the rules for clarity of presentation, they are fully taken into account in the actual implementation. Each component of the program state is accessed using a projection notation. For instance, $s_\Theta$ denotes the memory component of the state $s$. This operational model supports effectful computation while preserving a clear separation between name binding, memory layout, and program behavior, facilitating both reasoning and formal verification. The operational semantics of the language is present in \figref{exp_sem}.
\input{figures/expr_semantics}
The evaluation of a local or global variable in the state $\bstate{\Delta}{\Omega}{\Theta}$ proceeds by first resolving the variable $x$ to its corresponding location $l$ in $\Omega$ (for local variables) or $\Delta$ (for global variables), and then dereferencing $l$ in memory $\Theta$ to obtain the associated value. Rule $\mathsf{REFV}$ and $\mathsf{REF}$ define the operational semantic of $\mathsf{ref}$ construct. The construct $\mathsf{ref}$ in BeePL explicitly allocates a value to memory and returns the new location. The programmers are not intended to write expressions involving explicit, concrete locations: such expression will arise only as intermediate results during the operational semantics. This ensures that all pointers originate from well-defined allocation sites within the program, making them tractable and easier to reason about. Unlike C, where the address of any variable can be taken arbitrarily, BeePL guarantees that every address corresponds to memory that has been explicitly allocated. Similarly, the rules $\mathsf{DEREF}$ and $\mathsf{DEREFV}$ define the operational semantics of the dereference construct $!$. They evaluate a pointer by locating the address $l$ with offset $o$ in the memory $\Theta$, and produce the value stored at that specific location. The memory assignment construct $e_1 := e_2$ updates the value stored in the memory location pointed out by $e_1$ with the result of the evaluation $e_2$. Unlike let-binding, which introduces an immutable binding and does not alter memory, $:=$ performs a side-effecting update, modifying the program state. 

Let-binding introduces an immutable variable bound to the result of evaluating an expression, without performing memory updates—it simply extends the environment. Unlike \texttt{:=}, it is purely functional and side-effect free. In \(\elet{x}{\tau}{e_1}{e_2}\), the result of \(e_1\) is bound to \(x\) in \(e_2\). Unary and binary operators follow standard evaluation rules, but BeePL handles operations like division, modulo, and shifts explicitly to avoid undefined behavior common in C. Division and modulo return zero on division-by-zero or overflow; shift operations also return zero when exceeding bit-width or on overflow. This ensures deterministic semantics. The rule \(\mathsf{BOPV}\) uses an auxiliary function \(\mathsf{unsafe}\) to detect such cases based on the operator, operand values, and signedness. Struct initialization (\(\mathsf{STRUCTV}\)) allocates memory and assigns values to fields according to the struct layout. Field access retrieves the value at the correct offset, determined using type information from the global environment, even if types are not shown explicitly in the rule.

The operational semantics of the for loop in the form $\fore {e_1} {e_2} d {e}$ define a bounded iteration over a range of values determined by the evaluated results of $e_1$, $e_2$, and the direction $d$ (which can be either increasing or decreasing). Both $e_1$ and $e_2$ are evaluated exactly once before the loop begins and remain unchanged throughout the execution of the loop body $e$. This design guarantees termination, as the number of iterations is fixed and entirely determined by the initial bounds. By ensuring the immutability of the range and preventing the loop body from influencing the control flow, the construct avoids common issues such as unintended infinite loops or mutation-induced errors—offering a safer and more analyzable alternative to traditional C-style loops. The translation of BeePL for-loop to C for-loop is discussed in \secref{appendix}. The operational semantics of function application and external call is pretty straightforward.

\input{figures/expr_semantics1}

\figref{exp_sem1} presents the additional operational semantics for pattern matching in BeePL. BeePL supports three kinds of patterns, each serving a distinct purpose. Patterns over option types, such as $\kw{Pnone}$ and $\kw{Psome} \ x$, ensure that programmers explicitly handle cases where the evaluation of an expression $e$ results in either $\kw{None}$ or $\kw{Some}$. In the case of $\kw{None}$, the semantics proceed with the expression corresponding to the $\kw{Pnone}$ branch, closely resembling the handling of null pointers in C. In the case of $\kw{Some} \ x$, the value is extracted and execution continues with the expression associated with the $\kw{Psome}$ branch. The $\kw{MBYTES}$ rule defines how BeePL evaluates pattern matching over byte buffers. It checks that the value being matched is of type bytes and that the buffer is large enough to hold a value of the target type $\tau_x$. The pattern $\kpbytes{x}{\tau_x}{y}{\tau_y}$ attempts to parse the byte buffer into a value $v_x$ of type $\tau_x$, binding it to $x$. If $\tau_x$ is a struct, its individual fields are also extracted and bound to variables $\overline{y}$. The match succeeds if extraction is successful, and the body expression $e_1$ is evaluated with all these bindings substituted in. This rule handles both flat and structured types uniformly, ensuring that byte parsing is size-checked and safe before access. This form of pattern matching enforces memory bounds checks directly from type information and safe-semantics conditions, avoiding any reliance on low-level pointer arithmetic. The corresponding translation to C is provided in \secref{appendix}. The $\kw{EMPTY}$ rule states that evaluating an empty sequence returns an empty sequence without modifying the state. The $\kw{SEQH}$ rule evaluates the head of a non-empty sequence if it is not yet a value, while the $\kw{SEQT}$ rule allows evaluation to proceed to the tail only after the head has evaluated to a value.

%% file: figures/language.tex
\begin{figure}[h]
		\footnotesize
		\begin{minipage}[t]{.40\columnwidth}
			\[\begin{array}{@{}r@{\ \ }l@{\quad}l}
				e  \in \expr ::= & x                           & \mbox{variable}\\
				               \mid & c   			            & \mbox{constant} \\
				              \mid & \app e e                    & \mbox{application} \\
                                \mid & \prim {\kw{op}} e             & \mbox{primitive operations}\\
				              \mid & \elet x \tau e e            & \mbox{let-binding} \\
				              \mid & \cond e e e                 & \mbox{conditional} \\ 
                                \mid & \structi x {f} {e}           & \mbox{struct initialization}\\
                                \mid & \structf e f                 & \mbox{field access}\\
                                \mid & \kw{none}                    & \mbox{none}\\
                                \mid & \kw{some}\ e                 & \mbox{some}\\
                                \mid & \match e p e                 & \mbox{match}\\
                                \mid & \fore e e d e               & \mbox{for}\\
                                \mid & \kw{unit}                   & \mbox{unit}\\[5ex]
			\end{array}\]
                
                \[\begin{array}{@{}r@{\ \ }l@{\quad}l}
				fdecl  \in \kw{fundecl} ::= & \ef {id} {esg}                         & \mbox{external}\\
                                \mid & \kw{lcall} \ {id} \ {lsg}         & \mbox{internal}\\

                gdecl \in \kw{globdecl} ::= & \globalvar        & \mbox{global}\\[3ex]

			\end{array}\]

                \[\begin{array}{@{}r@{\ \ }l@{\quad}l}
				esg  \in \kw{sig} ::= & \{\kw{args_{type}} : \overline{\tau};\ \kw{res_{type}} : \tau;\ \kw{ef} : \eta;  \kw{cc}
                : \kw{conv} \}      \\
                lsg  \in \kw{sig} ::= & \fundecl\\[2ex]
			\end{array}\]
            
\[
\begin{array}{@{}l@{\qquad}l@{}}
\begin{array}{@{}r@{\ \ }l@{\quad}l}
co \in \text{\kw{composite}} ::= & \overline{x, \overline{y,\tau}} & \text{composite definitions}
\end{array}
&
\begin{array}{@{}r@{\ \ }l@{\quad}l}
\mathsf{P} \in \text{\kw{prog}} := \{ \text{decl} : \overline{\text{fdecl} \mid \text{gdecl}};\ \text{cenv} : \overline{co} \} & \text{program}
\end{array}
\end{array}
\]

		\end{minipage}%
		\begin{minipage}[t]{.70\columnwidth}
         \[\begin{array}{@{}r@{\ \ }l@{\quad}l}
				op  \in \kpop ::= & \kderef                         & \mbox{dereference}\\
                                      \mid & \kmass                          & \mbox{allocation}\\
                                      \mid & \kref                         & \mbox{reference} \\
                                      \mid & \kuop                       & \mbox{unary operation}\\
                                      \mid & \kbop                       & \mbox{binary operation}\\[2ex]
			\end{array}\]

             \[\begin{array}{@{}r@{\ \ }l@{\quad}l}
			      p  \in \kpat ::= & \kpnone                         & \mbox{None}\\
                                     \mid  & \kpsome x                  & \mbox{Some}\\
                                     \mid & \kpbytes {x} {\tau} {y} {\tau}   & \mbox{Bytes}\\[2ex]
			\end{array}\]
            
			\[\begin{array}{r@{\ \ }l@{\quad}l}
				c \in \const ::= & \kw{int} & \mbox{32-bit integers}\\
				               \mid & \kw{long}   & \mbox{64-bit integers}\\[2ex]

				v \in \val := & \kunit                     & \mbox{unit}\\
                              \mid & \kw{bool}             & \mbox{bool}\\
				           \mid & \kw{int}                  & \mbox{int}\\
				           \mid & \kw{long}                 & \mbox{long}\\
                              \mid & \loc l n                 & \mbox{location}\\
                              \mid & \kw{option} \ v            & \mbox{option}\\
				
			\end{array}\]
		\end{minipage}
		
	\caption{Syntax of BeePL}
        \vspace{-1em}
    \label{fig:BeePL}
\end{figure}

%% file: figures/types.tex
\begin{figure}[h]
    \centering
    \footnotesize
    \begin{minipage}[t]{0.3\textwidth}
        \begin{math}
        \begin{array}{@{}r@{\ \ }l@{\quad}l}
           \tau  \in \type ::= & \tau_v                                   & \mbox{prim type} \\
                                \mid & \tau*                              & \mbox{pointer}\\
                                \mid & \tfunction {\tau} {\eta} {\tau}    & \mbox{function}\\
                                \mid & \tbstruct {id}                     & \mbox{struct}\\
                                \mid & \tbarray {\tau_v} {n}              & \mbox{array}\\
                                \mid & \bytes                              & \mbox{bytes}\\
                                \mid & \tunit                              & \mbox{unit}\\[2ex]  
        \end{array}
        \end{math}
    \end{minipage}%
    \hspace{0.5em}
    \begin{minipage}[t]{0.3\textwidth}
        \begin{math}
        \begin{array}{@{}r@{\ \ }l@{\quad}l}
           \tau_v  \in \typev ::= & \tint {sz} {s}              & \mbox{int}\\
                                  \mid & \tlong {s}              & \mbox{long} \\
                                  \mid & \tbool                  & \mbox{boolean} \\[2ex]
           \tau_b  \in \typeb ::= & \tau_v                      & \mbox{value}\\
                                  \mid & \tbstruct {id}          & \mbox{struct} \\
                                  \mid & \tbarray {\tau_v} {n}   & \mbox{array}\\[2ex]
        \end{array}
        \end{math}
    \end{minipage}%
    \hspace{0.5em}
    \begin{minipage}[t]{0.3\textwidth}
    \vspace{-12ex}
        \begin{math}
        \begin{array}{@{}r@{\ \ }l@{\quad}l}
           \tau*  \in \tpointer ::= & \tref {\tau_b}            & \mbox{ref type}\\
                                  \mid & \toption {\tau*}       & \mbox{option type}\\
                                  \mid & \tpfunction {\tau} {\eta} {\tau} & \mbox{fun ptr}\\[2ex]
        \end{array}
        \end{math}
    \end{minipage}
    \caption{Types supported by BeePL}
    \label{fig:Btypes}
\end{figure}

%% file: figures/type_system.tex
\begin{figure}
 \footnotesize
  \[
\begin{array}{@{}c@{}}
   \inferrule*[left=\kw{TVAR}]{\Gamma \ x = \tau}{\typerule \Gamma \Sigma \Pi \Psi x \tau \phi}\ \ \ \ \ 
   \inferrule*[left=\kw{TCONSI}]{~}{\typerule \Gamma \Sigma \Pi \Psi i {\kw{int}} \phi}\ \ \ \ \
   \inferrule*[left=\kw{TCONSL}]{~}{\typerule \Gamma \Sigma \Pi \Psi l {\kw{long}} \phi}\\[2ex]
   \inferrule*[left=\kw{TCONSB}]{~}{\typerule \Gamma \Sigma \Pi \Psi b {\kw{bool}} \phi}\ \ \ \ \ 
   \inferrule*[left=\kw{TLOC}]{\Sigma \ l = \tref \tau}{\typerule \Gamma \Sigma \Pi \Psi {(l,o)} {\tref \tau} \phi}\\[2ex]
    \inferrule*[left=\kw{TREF}]{\typerule \Gamma \Sigma \Pi \Psi {e} {\tau_b} {\eta} \ \quad \kw{isBasic \ \tau_b} }
                               {\typerule \Gamma \Sigma \Pi \Psi {\kw{ref}(e)} {\tref {\tau_b}} (\kw{alloc} :: \eta)}\ \ \ \ \
    \inferrule*[left=\kw{TDEREF}]{\typerule \Gamma \Sigma \Pi \Psi {e} {\tau_e*} {\eta} \ \quad \tau_e* \neq \toption{\_*} }
                               {\typerule \Gamma \Sigma \Pi \Psi {!e} {\tau_e} (\kw{read} :: \eta)}\\[2ex]
    \inferrule*[left=\kw{TMASSGN}]{\typerule \Gamma \Sigma \Pi \Psi {e_1} {\tau_1*} {\eta_1} \ \quad \tau_1* \neq \toption{\_*}
                                    \quad \typerule \Gamma \Sigma \Pi \Psi {e_2} {\tau_1} {\eta_2}}
                               {\typerule \Gamma \Sigma \Pi \Psi {e_1 := e_2} {\kw{unit}_\tau} (\eta_1 \dotplus \eta_2 \dotplus \kw{write})}\\[2ex]
    \inferrule*[left=\kw{TBOP}]{\typerule \Gamma \Sigma \Pi \Psi {e_1} \tau {\eta_1} \quad 
                                \typerule \Gamma \Sigma \Pi \Psi {e_2} \tau {\eta_2} \quad \kw{isPrim \ \tau} }
                               {\typerule \Gamma \Sigma \Pi \Psi {\kw{bop}(op,e_1,e_2)} \tau {\eta_1 \dotplus \eta_2}}\ \ \ \ \
    \inferrule*[left=\kw{TUOP}]{\typerule \Gamma \Sigma \Pi \Psi e \tau \eta \quad \kw{isPrim \ \tau}}
                               {\typerule \Gamma \Sigma \Pi \Psi {\kw{uop}(op,e)} \tau \eta}\\[2ex]
    \inferrule*[left=\kw{TCOND}]{\typerule \Gamma \Sigma \Pi \Psi e \tau \eta \quad 
                                 \typerule \Gamma \Sigma \Pi \Psi {e_1} {\tau'} {\eta_1} \quad
                                 \typerule \Gamma \Sigma \Pi \Psi {e_2} {\tau'} {\eta_2} \quad
                                 \kw{isBool \ \tau}}
                                {\typerule \Gamma \Sigma \Pi \Psi {\cond e {e_1} {e_2}} \tau' {\eta \dotplus \eta_1 \dotplus \eta_2}}\\[2ex]
   \inferrule*[left=\kw{TAPP}]{\typerule \Gamma \Sigma \Pi \Psi {e} {\tau_e} {\eta} \quad 
                               \typerules \Gamma \Sigma \Pi \Psi {\overline{e}} {\overline{\tau}} {\eta'} \\
                               (\kw{if} \ \kw{isExt \ e} \ \kw{then} \ (\Psi(e) = \tau_f \quad \tau_e = \tau_f \quad \tau_e = {\tfunction {\tau} {\eta_e} {\tau_r}}) \ 
                                \kw{else} \ \tau_e = {\tfunction {\tau} {\eta_e} {\tau_r}}} 
                              {\typerule \Gamma \Sigma \Pi \Psi {\app e e} {\tau_r} {\eta \dotplus \eta' \dotplus \eta_e}}\\[2ex]
    \inferrule*[left=\kw{TBIND}]{\typerule \Gamma \Sigma \Pi \Psi {e_1} {\tau} {\eta_1} \quad 
                                 \typerule {(x \rightarrow \tau, \Gamma)} \Sigma \Pi \Psi {e_2} {\tau_2} {\eta_2}}
                                {\typerule \Gamma \Sigma \Pi \Psi {\elet x \tau {e_1} {e_2}} {\tau_2} {\eta_1 \dotplus \eta_2}}\\[2ex]
    \inferrule*[left=\kw{TSINIT}]{\Gamma \ x = \kw{struct} \ {id}* \quad \Pi \ id = co \quad 
                                  \mathsf{fields_\tau}(co) = \tau_s \quad \typerules {\Gamma} {\Sigma} {\Pi} {\Psi} {\overline {e}} {\tau_s'} {\eta_s} \ \quad \ \tau_s = \tau_s'}
                                {\typerule \Gamma \Sigma \Pi \Psi {\structi x {f} {e}} {\tpstruct {id} {}} {\eta_s}}\\[2ex]
    \inferrule*[left=\kw{TFIELD}]{\typerule {\Gamma} {\Sigma} {\Pi} {\Psi} {e} {\tau} {\eta} 
                                  \quad \tau = \kw{struct} \ {id}* \lor \tbstruct {id} \quad \Pi \ id = co 
                                  \quad \mathsf{field_\tau}(co, f) = \tau_f }
                                {\typerule \Gamma \Sigma \Pi \Psi {\structf e f} {\tau_f} {\eta}}\\[2ex]
    \inferrule*[left=\kw{TFOR}]{\typerule {\Gamma} {\Sigma} {\Pi} {\Psi} {e_1} {\tau_1} {\eta_1} 
                                  \quad \typerule {\Gamma} {\Sigma} {\Pi} {\Psi} {e_2} {\tau_2} {\eta_2}
                                  \quad \typerule {\Gamma} {\Sigma} {\Pi} {\Psi} {e} {\tau} {\eta}
                                  \quad \tau_1 = \tau_2 \quad \tau_1 \in \{\kw{int}, \kw{long}\}\\ 
                                  \quad \ \mathsf{fvar}(e) \cap \mathsf{fvar}(e_1) = \phi 
                                  \quad \mathsf{fvar}(e) \cap \mathsf{fvar}(e_2) = \phi}
                                {\typerule \Gamma \Sigma \Pi \Psi {\fore {e_1} {e_2} d e} {\tau} {\eta_1 \dotplus \eta_2 \dotplus \eta}}\\[2ex]
   \inferrule*[left=\kw{TSEQ}]{\typerule \Gamma \Sigma \Pi \Psi {e} \tau \eta \quad 
                               \typerules \Gamma \Sigma \Pi \Psi {es} {\tau_s} {\eta'}}
                               {\typerules  \Gamma \Sigma \Pi \Psi {e :: es} {\tau :: \tau_s} {\eta \dotplus \eta'}}\\[2ex]

  \end{array}
 \]
 \caption{Typing rules for expressions.}
\label{fig:typesystem}
\end{figure}

%% file: figures/type_system1.tex
\begin{figure}
 \footnotesize
  \[
\begin{array}{@{}c@{}}
   \inferrule*[left=\kw{TNONE}]{~}
                               {\typerule  \Gamma \Sigma \Pi \Psi {\kw{None}} {\toption {\tau*}} {\phi}}\ \ \ \ \
   \inferrule*[left=\kw{TSOME}]{\typerule \Gamma \Sigma \Pi \Psi e {\tau*} \eta}
                               {\typerule  \Gamma \Sigma \Pi \Psi {\kw{Some} \ e} {\toption {\tau*}} {\eta}}\\[2ex]
   \inferrule*[left=\kw{TMATCHO}]{\typerule \Gamma \Sigma \Pi \Psi {m} {\toption {\tau_m*}} {\eta_{m}} \quad 
                                   \quad \overline{p,e} = \{(\kw{Pnone}, e_1), (\kw{Psome \ x}, e_2), \ldots\} \quad \\
                                  \typerules {(x \rightarrow \tau_m*, \Gamma)} \Sigma \Pi \Psi {\overline{e}} {\tau_s} {\eta_s} \quad
                                   \forall \tau_1, \tau_2 \in \tau_s, \tau_1 = \tau_2 }
                                 {\typerule  \Gamma \Sigma \Pi \Psi {\match {m} p e} {\tau} {\eta_{m} \dotplus \eta_s}}\\[2ex]
  \inferrule*[left=\kw{TMATCHB}]{\typerule \Gamma \Sigma \Pi \Psi {m} {\kw{bytes}} {\eta_{m}} \quad 
                                  \overline{p,e} = \{(\kpbytes {x} {\tau_x} {y} {\tau_y}, e_1), \ldots\} \quad\\ 
                                  \typerules {(x \rightarrow \tau_x, \overline{y \rightarrow \tau_y}, \Gamma)} \Sigma \Pi \Psi {\overline{e}} {\tau_s} {\eta_s} \quad
                                   \exists \tau_r : \forall \tau \in \tau_s, \tau_r = \tau }
                                 {\typerule  \Gamma \Sigma \Pi \Psi {\match {m} p e} {\tau_r} {\eta_{m} \dotplus \eta_s}}
  \end{array}
 \]
 \caption{Typing rules for additional expressions.}
\label{fig:typesystem1}
\end{figure}

%% file: figures/typing_fundecl.tex
\begin{figure}[h]
 \footnotesize
  \[
\begin{array}{@{}c@{}}
\inferrule*[left=\kw{TFDECL}]{
  \typerule{\overline{x : \tau},\ \overline{y : \tau},\ \Gamma}{\Sigma}{\Pi}{\Psi}{e}{\tau}{\eta} \wedge 
  \kw{rt} = \tau \wedge 
  \kw{ef} = \eta \wedge 
  \neg(\tau*) \wedge
  \forall x \in \overline{x}.\ \mathsf{section\_ok}(\tau,\ \kw{s})
}{
  \Gamma, \Sigma, \Pi, \Psi \vdash\ \{\kw{sec}: \kw{s}; \kw{rt}:\tau;\ \kw{ef}:\eta; \kw{cc} : \kw{conv};\ \kw{args}:\overline{x,\tau};\ \kw{vars}:\overline{y,\tau};\kw{body}:e; \kw{flag}:\kw{bool}\}
}\\[2ex]
\inferrule*[left=\kw{TGDECL}]{
  \kw{typeof}(v) = \tau \wedge \mathsf{section\_ok}(\tau,\ \kw{s})
}{
  \Gamma, \Sigma, \Pi, \Psi \vdash\ \{\kw{gvar} : (x : \tau, v);  \kw{sec} : s\}
}\quad 
\inferrule*[left=\kw{TPROG}]{
  p = \overline{\kw{id, decl}} \wedge \forall d \in \kw{decl}. \vdash d
}{
  \Gamma, \Sigma, \Pi, \Psi \vdash\ p}
\end{array}
 \]
 \caption{Typing rules for BeePL program.}
\label{fig:typesystem3}
\end{figure}

%% file: figures/expr_semantics.tex
\begin{figure}
 \footnotesize
  \[
\begin{array}{@{}c@{}}
    \inferrule*[left=\kw{LVAR}]{s_{\Omega}[x] = l, \tau_x \wedge s_{\Theta}[l] = v}{\eval s x s v} \ \ \ \ \ \
   \inferrule*[left=\kw{GVAR}]{x \notin dom(s_{\Omega}) \wedge s_{\Delta}[x] = l \wedge s_{\Theta}[x] = v}{\eval s x s v}\\[2ex]
    \inferrule*[left=\kw{REF}]{\eval s e {{s'}} e'}{\eval s {\kref(e)} {{s'}} {\kref(e')}}\ \ \ \ \ \
    \inferrule*[left=\kw{REFV}]{\kw{fresh}(l) \wedge \Theta' = s_{\Theta}[l \rightarrow v] \wedge \Sigma[l \rightarrow *\kw{typeof}(v)]}{\eval s {\kref(v)} {\prec s_{\Delta},\Omega,\Theta'\succ} l} \\[2ex]
    \inferrule*[left=\kw{DREF}]{\eval s e {{s'}} e'}{\eval s {!(e)} {{s'}} {!(e')}}\ \ \ \ \ \
    \inferrule*[left=\kw{DREFV}]{\Theta[(l,\kw{o})] = v}{\eval s {!(l,\kw{o})} s v}\\[2ex]
    \inferrule*[left=\kw{MASSGN1}]{\eval s {e_1} {s'} {e_1'}}{\eval s {e_1 := e_2} {s'} {e_1' := e_2}}\ \ \ \ \ \
      \inferrule*[left=\kw{MASSGN2}]{\eval s e {s'} e'}{\eval s {l := e} {s'} {l := e'}}\ \ \ \ \ \
      \inferrule*[left=\kw{MASSGNV}]{\Theta' = \Theta[l \rightarrow v]}{\eval s {l := v} {\prec s_{\Delta},s_{\Omega},\Theta'\succ} \kw{unit}}\\[2ex]
    \inferrule*[left=\kw{LET}]{\eval s {e_1} {s'} {e_1'}}{\eval s {\elet x \tau {e_1} {e_2}} {s'} {\elet x \tau {e_1'} {e_2}}}\ \ \ \ \ \
      \inferrule*[left=\kw{LETV}]{{e_2}[x \leftarrow v] = {e_2'}}{\eval s {\elet x \tau {v} {e_2}} s {e_2'}}\ \ \ \ \ \\[2ex]
    \inferrule*[left=\kw{UOP}]{\eval s e {s'} e'}{\eval s {\mathsf{uop}(op,e)} {s'} {\mathsf{uop}(op,e')}}\ \ \ \ \ \
      \inferrule*[left=\kw{UOPV}]{\mathsf{uop_{sem}}(op,v) = v'}{\eval s {\mathsf{uop}(op,v)} {s'} {v'}}\ \ \ \ \ \
    \inferrule*[left=\kw{BOP1}]{\eval s {e_1} {s'} {e_1'}}{\eval s {\mathsf{bop}(op,e_1,e_2)} {s'} {\mathsf{bop}(op,e_1',e_2)}}\\[2ex]
    \inferrule*[left=\kw{BOP2}]{\eval s {e_2} {s'} {e_2'}}{\eval s {\mathsf{bop}(op,v_1,e_2)} {s'} {\mathsf{bop}(op,v_1,e_2')}}\ \ \ \ \ \
    \inferrule*[left=\kw{BOPV}]{v = \kw{if} \ \mathsf{unsafe}(op,v_1,v_2,s) \ \kw{then} \ 0 \ \kw{else} \ \mathsf{bop_{sem}}(op,v_1, v_2)}{\eval s {\mathsf{bop}(op, v_1, v_2)} {s'} {v}}\ \ \ \ \ \\[2ex]
    \inferrule*[left=\kw{COND}]{\eval s {e_1} {s'} {e_1'}}{\eval s {\cond {e_1} {e_2} {e_3}} {s'} {\cond {e_1'} {e_2} {e_3}}}\ \ \ \ \
    \inferrule*[left=\kw{CONDT}]{~}{\eval s {\cond {\kw{true}} {e_2} {e_3}} s e_2}\\[2ex]
    \inferrule*[left=\kw{CONDF}]{~}{\eval s {\cond {\kw{false}} {e_2} {e_3}} s e_3}\ \ \ \ \
    \inferrule*[left=\kw{STRUCT}]{\forall \ e \ \in \ \overline{e}: \eval s {e} {s'} {e'}}{\eval s {\structi x {\tau} {f} {e}} {s'} {\structi x {\tau} {f} {e'}}}\\[2ex]
      \inferrule*[left=\kw{STRUCTV}]{\tau = \tpstruct {sid} \wedge \Delta[sid] = co \wedge \forall \ f \ \in \overline{f}, \mathsf{field_{ofs}}(\Delta, f, co) = [o_1, \ldots,  o_n] \wedge \Omega[|x \rightarrow (l, o)|] = \Omega'\\   
                                     \wedge \ \Theta[f_1, (l + o + o_1) \rightarrow v_1, \ldots, f_n, (l + o + o_n) \rightarrow v_n] = \Theta'}{\eval s {\structi x {f} {v}} {\prec s_\Delta, \Omega', \Theta'\succ} {\loc l o}}\\[2ex]
    \inferrule*[left=\kw{FACCESS}]{\eval s {e} {s'} {e'}}{\eval s {\structf e f} {s'} {\structf e' f}}\ \ \ \ \ \
    \inferrule*[left=\kw{FACCESSV}]{\tau = \mathsf{struct \ id} \wedge \Delta[id] = co \wedge \mathsf{field_{ofs}}(\Delta, f, co) = o'}{\eval s {\structf {(l,o)} f} s (l, o + o')}\\[2ex]
    \inferrule*[left=\kw{FOR1}]{\eval s {e_1} {s'} {e_1'}}{\eval s {\fore {e_1} {e_2} d {e}} {s'} {\fore {e_1'} {e_2} d {e}}}\ \ \ \ \ \
    \inferrule*[left=\kw{FOR2}]{\eval s {e_2} {s'} {e_2'}}{\eval s {\fore {v_1} {e_2} d {e}} {s'} {\fore {v_1} {e_2'} d {e}}}\\[2ex]
    \inferrule*[left=\kw{FORV}]{\mathsf{range}(v1,v2,d) = n \wedge s, (e, \ldots, e)_n = {s'}, v}{\eval s {\fore {v_1} {v_2} d {e}} {s'} v}\ \ \ \ \ \
    \inferrule*[left=\kw{APP1}]{\eval s {e_1} {s'} {e_1'}}{\eval s {\app {e_1} {e_2}} {s'} { \app {e_1'} {e_2}}}\ \ \ \ \ \
    \inferrule*[left=\kw{APP2}]{\forall e \in \overline{e}, \eval s e {s'} e'}{\eval s {\app {fn} {e}} {s'} { \app {fn} {e'}}}\\[2ex] \inferrule*[left=\kw{APP3}]{\Delta[fn] = fd \wedge fd = \{\kw{sec} := sa; \kw{rt} := \tau; \kw{ef} := \eta; \kw{args} := as; \kw{vars} := vs; \kw{body} := e; \kw{flag} := b\} \\ \wedge \ as \cap vs = \phi \wedge \Omega[|ls \leftarrow as, l{s'} \leftarrow vs|] = \Omega' \wedge \Theta[ls \rightarrow \overline{v}]}{\eval s {\app {fn} {v}} {\prec s_\Delta, \Omega', \Theta' \succ} {e}}\\[2ex]
    \inferrule*[left=\kw{EAPP}]{\Delta[fn] = fd \wedge fd = \{\kw{arg_{type}} := \overline{\tau}; \kw{arg_{type}} := \tau_r; \kw{ef} := \eta; \kw{cc} := cc\} \wedge fn(\overline{v}) = v', {s'}}{\eval s {\app {fn} {v}} {s'} {v'}}\\[2ex]
    
  \end{array}
 \]   
 \caption{Semantics for expressions.}
\label{fig:exp_sem}
\end{figure}

%% file: figures/expr_semantics1.tex
\begin{figure}[h]
 \footnotesize
  \[
\begin{array}{@{}c@{}}
    \inferrule*[left=\kw{NONE}]{~}{\eval s {\kw{None}} s {\kw{None}}}\ \ \ \ \ \
    \inferrule*[left=\kw{SOME}]{\eval s e s' e'}{\eval s {\kw{Some} \ e} s' {\kw{Some} \ e'}}\ \ \ \ \ \
    \inferrule*[left=\kw{SOMEV}]{~}{\eval s {\kw{Some} \ v} s {\kw{Some} \ v}}\\[2ex]
    \inferrule*[left=\kw{MATCH1}]{\eval s e s' e'}{\eval s {\match e p e} s' {\match {e'} p e}}\\[2ex]
    \inferrule*[left=\kw{MSOME}]{\mathsf{isSome} \ v \wedge \overline{p,e} = \{(\kw{Pnone}, e_1), (\kw{Psome \ x}, e_2), \ldots\}}
    {\eval s {\match v p e} s {{e_2}[x \leftarrow v]}}\ \ \ \ \ \
    \inferrule*[left=\kw{MNONE}]{\overline{p,e} = \{(\kw{Pnone}, e_1), (\kw{Psome} \ x, e_2), \ldots\}}
    {\eval s {\match {\kw{None}} p e} s {e_1}}\\[2ex]
    \inferrule*[left=\kw{MBYTES}]{
  \mathsf{typeof}(v) = \kw{bytes} \ \wedge \
  \overline{p,e} = \{(\kpbytes {x} {\tau_x} {y} {\tau_y}, e_1), (p_2, e_2), \ldots\} \ \wedge \
  \kw{length}(v) \geq \kw{sizeof}(\tau_x) \ \wedge \\
  \kw{extract}(v, \tau_x) = (v_x, \{ y_1 \mapsto v_1, \ldots, y_n \mapsto v_n \})
}{
  \eval{s}{\match{v}{p}{e}}{s'}{e_1[\overline{y_i \mapsto v_i},\ x \mapsto v_x]}
}\\[2ex]
    \inferrule*[left=\kw{MBYTESF}]{\mathsf{typeof(v)} = \kw{bytes} \wedge \overline{p,e} = \{(\kpbytes {x} {\tau} {y} {\tau}, e_1), (p_2, e_2), \ldots\} \
    \wedge \ \kw{length}(v) < sizeof(\tau_x)}
    {\eval s {\match v p e} s {e_2}}\\[2ex]
    \inferrule*[left=\kw{EMPTY}]{~}{\evals s {[::]} s {[::]}}\ \ \ \ \ \
    \inferrule*[left=\kw{SEQH}]{\eval s e s' e'}{\evals s {(e :: es)} s' {(e' :: es)}}\ \ \ \ \ \
        \inferrule*[left=\kw{SEQT}]{\evals s {es} s' {es'}}{\evals s {(v :: es)} s' {(v :: es)}}\ \ \ \ \ \

  \end{array}
 \]   
 \caption{Semantics for additional expressions.}
 \label{fig:exp_sem1}
\end{figure}

%% file: sections/formal_proofs.tex
\section{Formally Verified Type System and the Compiler}\label{sec:formalproofs}
This section presents the core formal foundations of BeePL, focusing on its type system and verified compilation strategy. We first describe the type system that enforces key safety properties expected by the eBPF runtime, such as memory safety, control-flow integrity, and termination. We then provide formal proofs that well-typed BeePL programs satisfy a set of eBPF properties at the source level. Finally, we describe the compilation process that translates BeePL programs to BPF bytecode while preserving the functional correctness.

\subsection{Formally Verified Type System}
We start by proving some key properties related to the soundness of the type system. The BeePL type system is sound: any program that typechecks is guaranteed to never reach a stuck state. That is, every well-typed expression either evaluates to a value or can take a valid step in the operational semantics.

\paragraph{A well-typed BeePL program never reaches a stuck state:}
This property follows from the progress lemma~\ref{progress}, which ensures that every well-typed expression is either a value or can take a reduction step. Combined with preservation, this guarantees that the program maintains type correctness throughout execution, ruling out runtime type errors.

Since our operational semantics and typing rules involve various entities like local store, global store, and memory, we need a definition about well-formedness of the state. A well-formed state is defined as follow:
\begin{definition}[Well Formed State]\label{well-formed-state}
A state $s$ is well-formed, written as $\kw{well\_formed\_state}$, if:
\footnotesize
\[
\forall \Gamma, \Sigma, s
\begin{cases}
\forall x, \tau. \ \Gamma \ x = \mathsf{Some} \ \tau &
\begin{cases}
x \in dom(s_{\Omega}) & 
\exists l, v, o.\ s_{\Omega}[x] = (l, \tau) \wedge \Sigma \ l = \tau \wedge s_\Theta[l, o] = v \\[0.8ex]
x \notin dom(s_{\Omega}) &
\exists l, o, v.\ s_\Delta[x] = l \wedge \Sigma \ l = \tau \wedge s_{\Theta}[l,o] = v
\end{cases} \\[1.2ex]

\forall x, \tau. \ \Sigma \ x = \tau &
\kw{isValidAccess}(\Theta, x, \kw{Freeable}) \\[0.8ex]

\forall x.\ \kw{isValidAccess}(\Theta, x, \kw{Freeable})  &
\forall \tau. \ \Sigma \ x = \tau \\[0.8ex]

\forall \Pi, \Psi. &
\forall l,\ \eta,\ \tau_a,\ \eta',\ \tau_r,\ \tau_s.\ 
\typerule{\Gamma}{\Sigma}{\Pi}{\Psi}{l}{\eta}{\tfunction{\tau_a}{\eta}{\tau_r}} 
\wedge\ \typerules{\Gamma}{\Sigma}{\Pi}{\Psi}{es}{\overline{\tau_s}}{\eta'} \\
& \exists f.\ s_{\Delta}[l] = f \wedge f.(\kw{args}) \cap f.(\kw{vars}) = \phi \wedge\ |\kw{args}| = |es| \\
& \wedge \overline{\tau_s} = \kw{typeof}(f.\kw{args}) \wedge \tau_r = f.(\kw{rt}) \\[1.2ex]

\forall \Pi. &
\forall x, id. \
\Gamma \ x = \kw{struct} \ id* \wedge \exists co. \ \Pi \ id = \kw{Some} \ co
\end{cases}
\]
\end{definition}

The predicate \texttt{well\_formed\_state} characterizes when a memory store is consistent with a typing context \(\Gamma\) and store typing environment \(\Sigma\). It ensures that:
\begin{itemize}
  \item local variables in \(\Gamma\) point to valid locations in \(\Omega\), with well defined types in \(\Sigma\);
  \item global variables are absent from \(\Omega\) but mapped in \(\Delta\) with well defined types in \(\Sigma\);
  \item each location typed in \(\Sigma\) is accessible in \(\Theta\) with \texttt{Freeable} permission, and all such accessible locations are well-typed;
  \item functions in \(\Delta\) match their type signatures and avoid variable name clashes;
  \item struct variables must be declared in \(\Pi\).
\end{itemize}
This invariant underpins type soundness: it guarantees progress by ruling out stuck states from invalid memory access, and preservation by maintaining type consistency across evaluation steps.

All lemmas are also established for lists of expressions; however, their statements are omitted for brevity. The complete proofs are provided in Appendix.
\paragraph{Progress} The progress lemma presents the progress property for BeePL's expressions and expression sequences. It asserts that any well-typed expression or sequence of expressions is either a value or can take a computational step, ensuring that well-typed programs do not get stuck during execution.

\begin{lemma}\label{progress}[Progress]
\small
\begin{align*}
\forall \Gamma, \Sigma, \Pi, \Psi, e, \eta, \tau, s.\quad
&\typerule{\Gamma}{\Sigma}{\Pi}{\Psi}{e}{\tau}{\eta} 
\wedge \mathsf{well\_formed\_state}(\Gamma, \Sigma, s) \implies \\
&\mathsf{isVal}\ e 
\vee \exists s', e'.\ \eval{s}{e}{s'}{e'} 
\wedge \mathsf{well\_formed\_state}(\Gamma, \Sigma, s')
\end{align*}
\end{lemma}
The complete proof is provided in Appendix~\ref{progress}.

\paragraph{Preservation}
The preservation lemma ensures that if a well-typed expression takes a step, the resulting expression remains well-typed with the same type and possibly fewer effects. BeePL infers effects conservatively—for example, $\kw{if} \ b \ \kw{then} \ *p \ \kw{else} \ 0$ may carry a read effect even if $p$ is never dereferenced at runtime. Thus, preservation permits the resulting effects to be a subset of the original.
\begin{lemma}\label{preservation}[Preservation]
\small
\begin{align*}
\forall \Gamma, \Sigma, \Pi, \Psi, e, e', \eta, \tau, s, s'.\quad
&\typerule{\Gamma}{\Sigma}{\Pi}{\Psi}{e}{\tau}{\eta} 
\wedge \mathsf{well\_formed\_state}(\Gamma, \Sigma, s) \wedge \eval{s}{e}{s'}{e'} \implies \\
&\exists \eta'. \ \typerule{\Gamma}{\Sigma}{\Pi}{\Psi}{e'}{\tau}{\eta'} \wedge \eta' \subseteq \eta \wedge \mathsf{well\_formed\_state}(\Gamma, \Sigma, s') 
\end{align*}
\end{lemma}
The complete proof is provided in Appendix~\ref{preservation}.

To prove soundness we define a semantic closure for BeePL program defined in~\figref{sem_closure} (present in appendix). The soundness lemma states that if an expression (or sequence) is well-typed and the state is well-formed, then its evaluation to a final expression (or sequence) must result in a value. 
\begin{lemma}[Soundness]\label{soundness}
\small
\begin{align*}
\forall \Gamma, \Sigma, \Pi, \Psi, e, e', s, s', \eta, \tau, n.\quad
&\typerule{\Gamma}{\Sigma}{\Pi}{\Psi}{e}{\tau}{\eta} \wedge 
  \mathsf{well\_formed\_state}(\Gamma, \Sigma, s) \\
&\wedge\ \meval{s}{e}{s'}{e'}{n} \wedge e' \not\rightarrow \implies 
  \kw{isVal}\ e \wedge \exists \eta'.\ 
  \typerule{\Gamma}{\Sigma}{\Pi}{\Psi}{e'}{\tau}{\eta'} \\
&\wedge\ \eta' \subseteq \eta \wedge 
  \mathsf{well\_formed\_state}(\Gamma, \Sigma, s')
\end{align*}
\end{lemma}

The complete proof is provided in Appendix~\ref{soundness}.

\subsection{eBPF properties}\label{sec:ebpf-proof}
In this section, we present formal proofs demonstrating how each of the aforementioned eBPF safety properties defined in~\secref{ebpf-spec} is enforced by our framework.

Lemma~\ref{termination} ensures that any well-typed BeePL expression or list of expressions evaluates to a value in finitely many steps if the initial state is well-formed. It also guarantees that the divergence effect is not inferred, capturing both runtime termination and its sound reflection in the effect system.
\begin{lemma}\label{termination}[Termination]
\small
\begin{align*}
\forall \Gamma, \Sigma, \Pi, \Psi, es, \eta_s, \tau_s, s.\quad
&\typerules{\Gamma}{\Sigma}{\Pi}{\Psi}{es}{\tau_s}{\eta_s} 
\wedge \mathsf{well\_formed\_state}(\Gamma, \Sigma, s) \implies \\ 
& \exists n, s', vs.\ \mevals{s}{es}{s'}{vs}{n}  \wedge \kw{isVal} \ vs \wedge \kw{divergence} \notin \eta_s \\
\wedge\ 
\forall \Gamma, \Sigma, \Pi, \Psi, e, \eta, \tau, s.\quad
&\typerule{\Gamma}{\Sigma}{\Pi}{\Psi}{e}{\tau}{\eta} 
\wedge \mathsf{well\_formed\_state}(\Gamma, \Sigma, s) \\
& \exists n, s', v.\ \meval{s}{e}{s'}{v}{n}  \wedge \kw{isVal} \ v \wedge \kw{divergence} \notin \eta
\end{align*}
\end{lemma}
\begin{proof}
We proceed by mutual induction on the structure of the typing derivation for both expressions and expression sequences. Below, we present the proof for the case of a single expression. The case for a sequence of expressions follows similarly, by applying the induction hypothesis recursively at each step of the sequence. Induction on the structure of $e$ gives us different subgoals for each expressions of BeePL. We discuss interesting cases here: \\
\textit{Var case.}  
From the assumptions \(\typerule{\Gamma}{\Sigma}{\Pi}{\Psi}{x}{\tau}{\phi}\) and \(\mathsf{well\_formed\_state}(\Gamma, \Sigma, s)\), we invoke the progress theorem to conclude that the variable \(x\) evaluates to a value \(v\) in one step, yielding a resulting state \(s'\). That is, there exists \(n = 1\) such that $\meval{s}{x}{s'}{v}{1}.$
By applying the rule \textsc{SMULTI}, we reduce the goal to proving $\eval{s}{x}{s'}{v}$,
which is discharged by the result of the progress theorem. Since the evaluation yields a value and no further computation occurs, it terminates. Furthermore, by applying the preservation theorem to the typing assumption \(\typerule{\Gamma}{\Sigma}{\Pi}{\Psi}{x}{\tau}{\phi}\), the well-formedness of the state \(s\), and the evaluation step \(\eval{s}{x}{s'}{v}\), we obtain that the resulting value \(v\) is well-typed with type \(\tau\) and effect \(\eta'\), where \(\eta' \subseteq \phi\). This suffices to say that $\eta'$ = $\phi$, which trivially excludes \(\kw{divergence}\).

\textit{Constant and Values.} The proof is trivial for these cases. 

\textit{Ref.} We proceed by case analysis on the evaluation status of \(e\) in \(\kw{ref}\ e\). 
\begin{itemize}
\item \textit{Case 1:} For the case where $e$ is not a value, induction hypothesis guarantees that the evaluation of $e$ terminates. There exists a natural number $n$, a state $s'$, and a value $v'$ such that:
$\meval{s}{e}{s'}{v'}{n} \ \text{and} \ \kw{divergence} \notin \eta'$. Typing rule \textsc{TREF}, justifies that the type of $\kw{ref} \ e$ is $\tau\rptr$ with effect $\kw{alloc} :: \eta'$. Since $\kw{divergence} \notin \eta'$, implies $\kw{divergence} \notin \kw{alloc} :: \eta'$. After $n$ steps, the expression $\kw{ref}\ e$ reduces to $\kw{ref}\ v'$ in state $s'$. We then apply the operational rule \textsc{REFV}, which allocates a fresh location $l$ and returns it as a value: $\eval{s'}{\kw{ref}\ v'}{s''}{l}$
Applying the multi-step evaluation rule \textsc{SMULTI}, we conclude: $\meval{s}{\kw{ref}\ e}{s''}{l}{n+1}$, where $l$ is a value and $s''$ is the resulting state after allocation. Thus, the full evaluation of $\kw{ref} \ e$ terminates in $n+1$ steps with a value result, and since memory allocation preserves the well-formedness of state (by Lemma~\ref{ref_allocation_succeeds}), we conclude that termination holds.

\item \textit{Case 2:} Suppose the expression \( e \) has already reduced to a value \( v \). Then, the full expression becomes \( \kw{ref} \ v \). From the typing rule \textsc{TREF}, we know:
$\typerule{\Gamma}{\Sigma}{\Pi}{\Psi}{\kw{ref}~v}{\tau\rptr}{\kw{alloc} :: \eta} \ \text{and} \
\typerule{\Gamma}{\Sigma}{\Pi}{\Psi}{v}{\tau}{\eta}$. From the well-formedness of state \( s \), and by applying the operational rule \textsc{REFV}, the expression \( \kw{ref}~v \) allocates a new memory location \( l \), stores \( v \) at that location, and returns the location \( l \). That is, we have:$\eval{s}{\kw{ref}~v}{s'}{l}$ using n = 1 for multi step. To prove absence of divergence, we now apply the preservation theorem~\ref{preservation} to this step:
$\typerule{\Gamma}{\Sigma}{\Pi}{\Psi}{\kw{ref}~v}{\tau\rptr}{\kw{alloc} :: \eta} \ \wedge \ \eval{s}{\kw{ref}~v}{s'}{l} 
 \rightarrow \
\exists \eta'.\ \typerule{\Gamma}{\Sigma}{\Pi}{\Psi}{l}{\tau\rptr}{\eta'}\ \wedge\ \eta' \subseteq \kw{alloc} :: \eta$. Because the result is a value \( l \), and values carry no side-effect, we conclude \( \eta' = \phi \), and hence: $\kw{divergence} \notin \eta'$. Thus, \( \kw{ref}~v \) terminates in a single step to a value \( l \), and the total effect is bounded by the original effect \( \kw{alloc} :: \eta \), which excludes divergence by assumption.
Finally, we wrap this step using \textsc{SMULTI} to obtain: $\meval{s}{\kw{ref}~v}{s'}{l}{1}$. Hence, termination holds for this case.
\end{itemize}

\textit{Deref and Memory assignment.} The proof for dereference and memory assignment follows in the similar style as above. 

\textit{Let-binding.} We proceed by case analysis on whether the expression \( e_1 \) in \( \elet{x}{\tau}{e_1}{e_2} \) is a value or not.

\begin{itemize}
\item \textit{Case 1:} \( e_1 \) is not a value. From the typing judgment $\typerule{\Gamma}{\Sigma}{\Pi}{\Psi}{\elet{x}{\tau}{e_1}{e_2}}{\tau'}{\eta}$ and the typing rule \textsc{TLET}, we know: $
\typerule{\Gamma}{\Sigma}{\Pi}{\Psi}{e_1}{\tau}{\eta_1} \quad \text{and} \ 
\typerule{x \mapsto \tau, \Gamma}{\Sigma}{\Pi}{\Psi}{e_2}{\tau'}{\eta_2}, \ \text{with} \ \eta = \eta_1 \dotplus \eta_2.$
By the induction hypothesis on \( e_1 \), there exists \( n, s', v_1 \) such that $
\meval{s}{e_1}{s'}{v_1}{n} \ \text{and} \ \kw{divergence} \notin \eta_1.$
After \( n \) steps, the expression reduces to \( \elet{x}{\tau}{v_1}{e_2} \), and applying rule \textsc{LETV} yields \( e_2[x \leftarrow v_1] \). Thus, $\meval{s}{\elet{x}{\tau}{e_1}{e_2}}{s'}{e_2[x \leftarrow v_1]}{n+1}.$
To conclude, we must show that \( e_2[x \leftarrow v_1] \) terminates. By Lemma~\ref{subst-preserve-typing} and preservation over \( n+1 \) steps, we get: $\typerule{\Gamma}{\Sigma}{\Pi}{\Psi}{e_2[x \leftarrow v_1]}{\tau'}{\eta_2'} \quad \text{with} \quad \eta_2' \subseteq \eta_2,$
and \( s' \) is well-formed. By applying the progress theorem to \( e_2[x \leftarrow v_1] \), we conclude it either is a value or takes a step. In the latter case, we apply the IH for \( e_2 \) (on the substituted expression) to conclude termination, with divergence still excluded since \( \kw{divergence} \notin \eta_1 \) and \( \kw{divergence} \notin \eta_2 \).
\item \textit{Case 2:} \( e_1 \) is already a value \( v_1 \). Then we apply rule \textsc{LET} directly: $
\elet{x}{\tau}{v_1}{e_2} \rightarrow e_2[x \leftarrow v_1].$
The rest proceeds as in Case 1: we apply the substitution lemma, preservation, and the induction hypothesis to show that \( e_2[x \leftarrow v_1] \) terminates.
\end{itemize}

\textit{For loop} The only form of looping permitted in BeePL is through the construct $\fore{e_1}{e_2}{d}{e}$. The termination proof proceeds by analyzing the different cases of the operational semantics associated with the for-loop.

\begin{itemize}
\item \textit{Case 1:} For the case where \(e_1\) steps to \(e_1'\), the induction hypothesis ensures that \(e_1\) evaluates to a value \(v_1\) in \(n_1\) steps with \(\kw{divergence} \notin \eta_1\), where \(\eta_1\) is the effect from the typing of \(e_1\) in the \texttt{for}-loop. Thus, \(\meval{s}{\fore{e_1}{e_2}{d}{e}}{\fore{v_1}{e_2}{d}{e}}{s_1}{n_1}\), and by preservation, \(v_1\) is well-typed, \(s_1\) is well-formed, and \(\eta_1' \subseteq \eta_1\). Similarly, by the induction hypothesis for \(e_2\), there exists \(n_2\) such that \(e_2 \to^{n_2} v_2\) with \(\kw{divergence} \notin \eta_2\). Hence, \(\meval{s_1}{\fore{v_1}{e_2}{d}{e}}{\fore{v_1}{v_2}{d}{e}}{s_2}{n_2}\), where \(v_2\) is well-typed, \(s_2\) is well-formed, and \(\eta_2' \subseteq \eta_2\) by preservation. Since \(v_1\) and \(v_2\) are numeric constants (either \(\kw{int}\) or \(\kw{long}\) per the \texttt{TFOR} rule), the loop range is computed, yielding a fixed number \(k \in \mathbb{N}\) of iterations, including the case \(k = 0\) when the body is skipped.

We proceed by induction on $k$, the number of loop iterations.

\begin{itemize}
  \item \textit{Base Case ($k = 0$):} In this case, the range is empty, so the loop body does not execute. The loop evaluates in one final step to a unit value: $
  \fore{v_1}{v_2}{d}{e} \to \kw{Unit}$. 
  This takes a bounded number of steps ($n_1$ + $n_2$ + 1), and terminates in a value without divergence.

  \item \textit{Inductive Case ($k > 0$):} Since the loop body \(e\) is independent of the loop index and executes identically each time, the loop can be treated as \(k\) repetitions of \(e\), where \(k\) is determined by \(v_1\), \(v_2\), and direction \(d\). By induction, \(e\) terminates in value \(v_e\) in at most \(n_e\) steps, with \(\kw{divergence} \notin \eta_e\). Thus, the full loop runs \(k\) independent executions of \(e\), each taking up to \(n_e\) steps, and the total steps are bounded by \(n_1 + n_2 + k \cdot n_e\). As none of the components diverge, the loop terminates without divergence. To show \(\kw{divergence} \notin \eta\), we note that it is absent from \(\eta_1\), \(\eta_2\), and \(\eta_e\). Since the typing rule gives \(\eta = \eta_1 \dotplus \eta_2 \dotplus \eta_e\), we conclude \(\kw{divergence} \notin \eta\), completing the proof.

\item \textit{Case 2:} For all other cases where the expressions $e_1$, $e_2$, or $e$ are already values, we proceed in a similar manner as in Case~1. The proof follows the same structure but requires fewer evaluation steps, as some sub-expressions do not require further reduction.
\end{itemize}
Hence, the full \texttt{for}-loop expression evaluates to a value in finite steps and terminates without divergence.
\end{itemize}

\textit{Function application.} The termination proof for function application uses the inductive hypothesis on argument expressions to ensure their termination. We apply preservation inductively over the argument list to establish state well-formedness and then instantiate the typing judgment for the function body. Since BeePL only supports statically defined top-level functions (no closures or dynamic creation), all invoked functions are guaranteed to be well-typed, and their termination is analyzed structurally.

Termination for other expressions—primitive operators, struct initialization, field access, and pattern matching—is omitted as it follows the same structural induction strategy. For external function calls, we assume they terminate and do not prove this within BeePL’s semantics.
\qedhere
\end{proof}
\vspace{-0.2em}

Lemma~\ref{division-by-zero} states that if a unary or binary operator is well-typed and evaluated in a well-formed state, the result is never $\kw{undef}$. This is not guaranteed by the type system alone, but by a program transformation performed by the BeePL compiler. While operators like division, modulo, and shifts are partial in CompCert, BeePL compiler inserts default values in cases of undefined behavior to ensure well-defined results.

\begin{lemma}\label{division-by-zero}[Absence of undefined behavior]
\small
\begin{align*}
\forall \Gamma, \Sigma, \Pi, \Psi, s, e, e_1, e_2, n, v, s'.\ 
\mathsf{well\_formed\_state}(\Gamma, \Sigma, s) \implies \\
\typerule \Gamma \Sigma \Pi \Psi {\kw{bop}(op,e_1,e_2)} \tau \eta \wedge 
\meval{s}{\kw{bop}(op, e_1, e_2)}{s'}{v}{n} \implies
v \neq \kw{undef} \ \wedge \\ 
\typerule \Gamma \Sigma \Pi \Psi {\kw{uop}(op,e)} \tau \eta \wedge 
\meval{s}{\kw{uop}(op, e)}{s'}{v}{n} \implies
v \neq \kw{undef}
\end{align*}
\end{lemma}
\begin{proof}
The proof proceeds by destructing the goal, which gives us two sub-goals to prove-one related to binary operator and the other related to unary operator. 
\begin{itemize}
\item \textit{\kw{bop} case:} 
$\meval{s}{\kw{bop}(op, e_1, e_2)}{s'}{v}{n}$ gives two steps: 
$\meval{s}{\kw{bop}(op, e_1, e_2)}{s_1}{\kw{bop}(op, v_1, v_2)}{n{-}1}$ and 
$\eval{s_1}{\kw{bop}(op, v_1, v_2)}{s'}{v}$. From \textsc{BOPV}, we get: 
\[
v = \kw{if} \ \mathsf{unsafe}(op,v_1,v_2,s) \ \kw{then}\ 0\ \kw{else}\ \mathsf{bop_{sem}}(op,v_1,v_2).
\]
Case-split on the conditional: (i) if $v = 0$, then $v \neq \kw{undef}$ trivially; 
(ii) if $v = \mathsf{bop_{sem}}(op,v_1,v_2)$, apply Lemma~\ref{progress-binaryoperators} to typing 
$\typerule \Gamma \Sigma \Pi \Psi {\kw{bop}(op,e_1,e_2)} \tau \eta$ and 
$\mathsf{well\_formed\_state}(\Gamma, \Sigma, s)$ to show that 
$v' = \mathsf{bop_{sem}}(op,v_1,v_2)$ with $v' \neq \kw{undef}$.

\item \textit{\kw{uop} case:} 
$\meval{s}{\kw{uop}(op, e)}{s'}{v}{n}$ gives 
$\meval{s}{\kw{uop}(op, e)}{s_1}{\kw{uop}(op, v_1)}{n{-}1}$ and 
$\eval{s_1}{\kw{uop}(op, v_1)}{s'}{v}$. From \textsc{UOPV}, we get 
$\mathsf{uop_{sem}}(op,v_1) = v$. Applying Lemma~\ref{progrss-unaryoperators} to typing 
$\typerule \Gamma \Sigma \Pi \Psi {\kw{uop}(op,e)} \tau \eta$ and 
$\mathsf{well\_formed\_state}(\Gamma, \Sigma, s)$ yields 
$v' = \mathsf{uop_{sem}}(op,v_1)$ with $v' \neq \kw{undef}$. Hence proved.
\end{itemize}
\qedhere
\end{proof}

\begin{lemma}\label{no-uinitialized-memory-access}[No uninitialized memory access]
\small
\begin{align*}
\forall \Gamma, \Sigma, \Pi, \Psi, s, \tau, \eta, v, s''.\quad 
\typerule{\Gamma}{\Sigma}{\Pi}{\Psi}{!e}{\tau}{\eta} \wedge \mathsf{well\_formed\_state}(\Gamma, \Sigma, s) \wedge 
\meval{s}{!e}{s''}{v}{n} \implies \\
\exists \eta', l, o, v, s'. \
\typerule{\Gamma}{\Sigma}{\Pi}{\Psi}{e}{\tau\rptr}{\eta'} \wedge \eta' \subseteq \eta \wedge 
\eval{s}{e}{s'}l \wedge s'_\Theta[l] = v
\end{align*}
\end{lemma}
\begin{proof}
    In BeePL, memory locations may arise either through internal allocation via the \texttt{ref} construct or from the external environment. Given the typing hypothesis \(\typerule{\Gamma}{\Sigma}{\Pi}{\Psi}{!e}{\tau}{\eta}\), we know that the expression \(e\) must have type \(\tau*\), i.e., a pointer to type \(\tau\), and not an \texttt{option} type. Any pointers originating from outside the program must be resolved using a \texttt{match} construct before being safely dereferenced. For pointers generated internally within a BeePL program, we know they are never null, as the semantics of the \texttt{ref} construct guarantees that a value is allocated in memory and its address is returned.

From the typing rule for dereferencing, the hypothesis \(\typerule{\Gamma}{\Sigma}{\Pi}{\Psi}{!e}{\tau}{\eta}\) implies the existence of a subderivation \(\typerule{\Gamma}{\Sigma}{\Pi}{\Psi}{e}{\tau*}{\eta'}\) such that \(\eta = \kw{read} :: \eta'\). By the effect subsumption property, we conclude that \(\eta' \subseteq \kw{read} :: \eta'\), and the typing goal is discharged accordingly. Now, from the multi-step evaluation rule \textsc{SMULTI}, applied to the hypothesis \(\meval{s}{!e}{s''}{v}{n}\), we obtain two intermediate steps: \(\meval{s}{!e}{s_1}{!v'}{n - 1}\) and \(\eval{s_1}{!v'}{s'}{v}\), where \(v' = l\) is a location. By the preservation lemma, we know that this location \(l\) remains well-typed in the state \(s_1\), and the state itself is well-formed. This ensures that \(\Sigma(l) = \tau\). Moreover, since a value is always guaranteed to be present at location \(l\), we know that the memory access cannot be uninitialized. Finally, by applying Lemma~\ref{safe-dereference}, which guarantees the safety of dereferencing any well-typed location in a well-formed state, we conclude that there exists a value \(v\) stored at location \(l\) in memory \(s'\). Hence, dereferencing in this case is always safe, and the proof is complete.
\end{proof}

\begin{lemma}[No Null Pointer Dereference]\label{no-null-pointer-dereference}
Let \(\Gamma, \Sigma, \Pi, \Psi\) be a well-formed typing context, and let \(s\) be a well-formed state such that \(\mathsf{well\_formed\_state}(\Gamma, \Sigma, s)\). Then, for any expression \(l\) and type \(\tau_l\), if
\[
\typerule{\Gamma}{\Sigma}{\Pi}{\Psi}{l}{\tau_l}{\eta}
\]
and \(l\) is dereferenced (i.e., appears in a subterm of the form \(!l\)), then one of the following holds:
\begin{itemize}
  \item \( \tau_l = \tau~\rptr \), in which case \(l\) evaluates to a non-null location in memory
  
  \item $\tau_l = \kw{option}~\tau*$, and the dereference of \ \(l\) occurs only inside a \texttt{match} construct ($\kw{match}\ l\ \texttt{with}\ \kw{Pnone} \Rightarrow e\ \mid\ \kw{Psome}\ x \Rightarrow !x$),
  which ensures that the dereference is only performed in the \texttt{Some} branch, after checking for null.
\end{itemize}
Hence, in all well-typed programs, pointer dereferences are guaranteed to be safe: either the pointer is guaranteed non-null, or a runtime \texttt{match} check eliminates the null case before dereferencing.
\end{lemma}
\begin{proof}
\textit{Case 1:} \( \tau_l = \tau~\rptr \). From the typing rule for dereferencing, we know that the judgment \(\typerule{\Gamma}{\Sigma}{\Pi}{\Psi}{!l}{\tau}{\kw{read} :: \eta'}\) is derivable only when \(l\) has type \(\tau~\rptr\). The typing rules for BeePL do not permit dereferencing expressions of other types directly. 
Since \(\typerule{\Gamma}{\Sigma}{\Pi}{\Psi}{l}{\tau~\rptr}{\eta'}\) and the state \(s\) is well-formed. From the typing rule of location $\kw{TLOC}$, we know that $\Sigma(\ell) = \tau \rptr$. Moreover, in BeePL, locations are only introduced via the \texttt{ref} construct, which allocates memory and ensures that the location stores a value immediately. Therefore, dereferencing \(\ell\) is safe and never accesses null pointer.

\textit{Case 2:} \( \tau_l = \kw{option}~\tau* \). By inversion on the typing rule, it is not possible to type \(!l\) directly if \(l\) has type \(\kw{option}~(\tau*)\). The only way a dereference can occur in this case is when \(l\) is pattern matched using a \texttt{match} construct: $\kw{match}\ l\ \texttt{with}\ \kw{Pnone} \Rightarrow e_1\ \mid\ \kw{Psome}\ x \Rightarrow !x$. By the typing rule for \texttt{match}, both branches \(e_1\) and \(!x\) should be type checked with the same type. Specifically, in the \(\kw{Psome}\ x\) branch, \(x : \tau*\) is added to the context. Therefore, inside the \(\kw{Psome}\) branch, dereferencing \(!x\) is allowed as presence of $x$ ensures that $l$ is not null. Operationally, the semantics of BeePL guarantee that only the branch corresponding to the runtime constructor is evaluated during pattern matching. Hence, \(!x\) is only executed when \(l\) evaluates to \(\kw{Psome}\ \ell\) for some valid memory location \(\ell\). 

In both cases, dereferencing is only performed on non-null pointers. Either the pointer has type \(\tau~\rptr\) and is directly dereferenced, or it is an \texttt{option} type and pattern matched to safely isolate the \(\kw{Some}\) case before dereferencing. Therefore, null pointer dereferences cannot occur in any well-typed BeePL program.
\end{proof}

\begin{lemma}[Memory Bounds Safety]\label{memory-bounds-safety}
Let $\Gamma, \Sigma, \Pi, \Psi \vdash \match{v}{p}{e} : \tau, \eta$ and the evaluation succeeds via the rule $\kw{MBYTES}$ (the success case), that is, $\eval{s}{\match{v}{p}{e}}{s}{e_1[\overline{y_i \mapsto v_i},\ x \mapsto v_x]}.$
Then, all accesses to the byte value $v$ during the pattern match are within bounds, i.e., the extraction $\kw{extract}(v, \tau_x) = (v_x, \{ y_1 \mapsto v_1, \ldots, y_n \mapsto v_n \})$ is defined only when $\kw{length}(v) \geq \kw{sizeof}(\tau_x)$, ensuring no out-of-bounds memory access occurs.
\end{lemma}

\begin{proof}
From the typing derivation $\Gamma, \Sigma, \Pi, \Psi \vdash \match{v}{p}{e} : \tau, \eta$, we know:
\begin{itemize}
  \item $v$ is a value of type \kw{bytes},
  \item the pattern includes a case of the form $\kpbytes{x}{\tau_x}{y}{\tau_y} \Rightarrow e_1$,
  \item $\kw{sizeof}(\tau_x)$ is defined and finite as it extends the context to type check $\overline{e}$.
\end{itemize}
The result is a direct consequence of the operational semantics, as for ${\match{v}{p}{e}}$ to proceed with extraction and substitution, the check on the length of the buffer ($\kw{length}(v) \geq \kw{sizeof}(\tau_x)$ in this case) is necessary to hold. Hence this proves that extraction of a byte buffer $v$ in some target of type $\tau_x$ is safe by construction according to the semantics of $\kw{match}$ ($\kw{MBTYTES}$). If the condition fails, evaluation proceeds via the fallback rule \textsc{MBYTESF}, which does not perform any memory access. Hence, no out-of-bounds access can occur during evaluation of a match on bytes. Hence, no rule in the semantics permits extraction from out-of-bounds memory, and all memory accesses during match evaluation on byte values are within valid bounds.
\end{proof}

\subsection{Formally verified compiler from BeePL to BPF bytecode}
We implemented the BeePL compiler in the Rocq proof assistant. This compiler translates BeePL expressions into \texttt{Csyntax}, the first intermediate representation used by the CompCert verified C compiler. The translation preserves the semantics and structural discipline of BeePL while mapping its high-level constructs—such as function calls, control flow, structured data, and memory operations—into the imperative syntax expected by CompCert. The resulting \texttt{Csyntax} program is then compiled to BPF bytecode using the CompCert compilation pipeline. This end-to-end path—from BeePL to BPF via C—enables the generation of safe and semantically sound BPF bytecode, suitable for secure deployment in the kernel.

We present the design, structure, and implementation of the BeePL compiler in the Appendix. There, we detail the translation of key constructs such as \kw{ref}, \kw{let}-bindings, \kw{for}-loops, and \kw{match} expressions. Cases with direct one-to-one mappings from BeePL expressions to C statements are omitted for brevity. At a high level, we highlight how the translation of the \kw{for}-loop preserves termination, reflecting a deliberate design choice in the compiler. As detailed in \secref{appendix}, the BeePL-to-C translation evaluates the loop bounds $e_1$ and $e_2$ once at the start and assigns them to fresh variables $l$ and $h$. Based on the specified direction (\texttt{Up} or \texttt{Down}), the translation introduces a fresh loop index $i$, initialized to $l$, and updates it using either $i{++}$ or $i{-}{-}$ in each iteration. The loop executes only if the initial condition ($l \leq h$ for upward traversal, $l \geq h$ for downward) is satisfied, and this condition is checked before entering the loop. Since the bounds are fixed and the loop index advances by exactly one in each iteration, it cannot overshoot or skip the endpoint. As a result, the loop always makes progress and is guaranteed to terminate. The formal proof of the compiler’s correctness is given in Lemma~\ref{compiler-correctness}, which appears in the Appendix.

\paragraph{Extension of CompCert for Verified Compilation to BPF Bytecode}

CompCert already features backends for classic Instruction Set Architectures (ISA), namely PPC, x86, ARM and RISCV. 
Compared to those ISAs, eBPF has a very restricted instruction set. For instance, there is no support for signed arithmetic or floating-point numbers.
Our eBPF backend emulates signed arithmetic (notably signed division) at instruction selection. 
Other signed operations e.g., a cast to a signed byte is encoded using a left shift followed by an arithmetic right shift. 
Programs using floating-point numbers are simply rejected at assembly generation.
All the necessary code for performing and proving the CompCert optimisation passes
is borrowed from existing backends with little modifications.

The handling of functions calls and return require more work because the eBPF semantics deviate from what is usually implemented. In standard ISA (except x86), a call is a branch-and-link; the function prologue allocates a new stack frame, updates the stack pointer $sp$ and 
saves the stack pointer of the caller and the return address (if needed).
At function exit, the stack pointer $sp$ is explicitly restored and control is given back to the caller. As soon as the Application Binary Interface (ABI) is respected, the precise handling of the stack pointer is responsibility of the compiler.
Therefore, at function entry, CompCert  sets the stack pointer at the bottom of the newly allocated stack frame. 
As a result, within the function, stack accesses are of the form $\mathit{sp+ofs}$ where $\mathit{ofs}$ is a positive offset.

In eBPF, things are done differently. A deviation is that the stack register (R10) is documented as read-only\footnote{\url{https://www.kernel.org/doc/html/v5.17/bpf/instruction-set.html}}.
Actually, this is the responsibility of the JIT compiler to update the stack register; making sure that stack frames are allocated and restored
but also to enforce calling-conventions (callee-save registers are saved at function entry and restored at function exit).
As a result, though the eBPF program never explicitly updates R10, after a call, R10 points to the top of a newly allocated stack frame.
Unlike CompCert, stack accesses are of the form $R10-\mathit{ofs}$. Moreover, the stack pointer of the caller and the return address are not available.
To give a semantics to this behaviour while maintaining maximum compatibility with CompCert, we arrange the stack-frame layout as shown in Figure~\ref{fig:stacklayout}. 
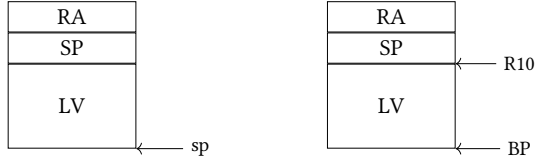
\begin{figure}[h]
\centering
\scalebox{0.85}{
\begin{tikzpicture}

  \node[rectangle, draw=black, minimum width=2cm, minimum height=0.4cm] (mra) at (0,0) {RA};
  \node[rectangle, draw=black, minimum width=2cm, minimum height=0.4cm, anchor=north] (msp) at (mra.south) {SP};
  \node[rectangle, draw=black, minimum width=2cm, minimum height=1.3cm, anchor=north] (mlv) at (msp.south) {LV};
  \path (mlv.south east)++(1,0) node (nsp) {\small sp};
  \draw[->] (nsp) -- (mlv.south east);

  \node[rectangle, draw=black, minimum width=2cm, minimum height=0.4cm] (ra) at (5,0) {RA};
  \node[rectangle, draw=black, minimum width=2cm, minimum height=0.4cm, anchor=north] (sp) at (ra.south) {SP};
  \node[rectangle, draw=black, minimum width=2cm, minimum height=1.3cm, anchor=north] (lv) at (sp.south) {LV};
  \path (sp.south east)++(1,0) node (r10) {\small R10};
  \draw[->] (r10) -- (sp.south east);
  \path (lv.south east)++(1,0) node (bp) {\small BP};
  \draw[->] (bp) -- (lv.south east);

\end{tikzpicture}
}
\caption{Stack layout: Mach language (left), eBPF assembly (right)}
\label{fig:stacklayout}
\end{figure}

The left hand side of Figure~\ref{fig:stacklayout} shows the layout of a stack frame of the Mach language which is the last CompCert intermediate language before assembly. We inherit from the CompCert infrastructure that the stack register $sp$ points to the bottom of the stack frame.
However, the precise layout of the stack frame is architecture dependent.
For the eBPF backend, we specify that the return address RA and the stack pointer of the caller SP are saved at the top of the stack frame.
At assembly level (right hand side of Figure~\ref{fig:stacklayout}), the layout of the stack frame is the same with the twist 
that the stack register $R10$ points at top of the stack frame just below SP. 

To facilitate the simulation proof, we introduce a (non-architectural) register BP (Base Pointer) which simulates the Mach $sp$ register.
Within a function $f$, the stack offset $\delta$ has value $f.fn\_stacksize - 2 * 8$ because SP and RA are 8 bytes i.e., 64 bits quantities 
and we have the invariant that $R10 = BP + \delta$.

In terms of assembly generation, there are several consequences that are worth mentioning.
First of all, every Mach stack access of the form $sp+\mathit{ofs}$ is compiled into an eBPF access of the form $R10 + (ofs - \delta)$.
In our current backend, the Mach instuctions $Mtailcall$ which performs a tail call and $Mgetparam$ which fetches arguments from the caller's 
stack are not supported. 
For tailcalls, the reason is that the classic assembly code sequence to model them consists in a explicit free of the current stack frame 
followed allocation by a direct jump to the function entry. As eBPF stack frame are implicitly handled by the runtime and because 
it is unclear how to semantically distinguish between a recursive tailcall to the local function from a genuine local jump, 
we simply forbid this construction.
Another consequence of the fact that the stack frames are managed by the runtime is that the eBPF calling conventions 
only allow parameter passing through registers. 
The rationale being that the parent stack frame (the value of SP) is hidden to the callee. 
This explains why our backend does not support the Mach instruction $Mgetparam$ which fetches arguments from the caller's stack.

Under these restrictions, our current eBPF backend is proved correct within the CompCert infrastructure and therefore ensure semantics 
preservation down to eBPF assembly. The assembly is thereafter compiled down to a loaded elf file using the existing eBPF assembler 
of the llvm toolchain.

%% file: sections/relatedwork.tex
\section{Related Work}\label{sec:relatedwork}

\paragraph{Verifier correctness}
There have been promising efforts to enhance verifier correctness, for instance,~\cite{10.1007/978-3-031-37709-9_12, Bhat2022FormalVO} improve range analysis but do not cover memory safety or termination checks. Although these works provide formal justifications, they lack machine-checked proofs and remain susceptible to state-explosion. Similarly, \cite{10.1145/3341301.3359641} employs scalable symbolic execution to verify eBPF safety, but despite path-reduction optimizations, it still faces state-explosion and complex constraint-solving issues. Another approach~\cite{258848} ensures BPF Just-In-Time (JIT) compilation correctness through machine-checked proofs. However, it still relies on the verifier or advanced symbolic execution tools like Serval for bytecode safety, focusing on preserving correctness in the translation from BPF bytecode to native machine code rather than addressing verifier limitations.

State Embedding~\cite{298730} validates the Linux eBPF verifier by embedding runtime checks to detect logic bugs, improving trust in existing infrastructure. In contrast, BeePL ensures safety by construction through a high-level language and a formally verified compiler, eliminating reliance on the kernel verifier. VEP~\cite{306007} preserves annotations from C to bytecode to aid verification and improve error messages, but its compiler's preservation is not formally verified. In contrast, CompCert also preserves annotations with formal guarantees in Rocq, while BeePL avoids reliance on annotations altogether. Compilers like Jasmin~\cite{10.1145/3133956.3134078} and CompCert~\cite{10.1007/11813040_31} are formally verified to preserve a range of properties. Unfortunately, these cannot be directly utilized to compile eBPF programs due to their limitations at both the source and target levels. Neither compiler can reason about eBPF properties at the source level or produce BPF bytecode as a target.

\paragraph{Language features}
Most BeePL language features on themselves are not unique.
Provable terminating loops are part of theorem provers like Rocq~\cite{barras1999coq}, Lean~\cite{moura2021lean}, Agda~\cite{bove2009brief}, and Idris~\cite{brady2013idris}.
All use using structural recursion on algebraic data types to prove termination \cite{abel2002predicative,abel2004termination},we use restricted loop constructs and type-and-effect system to guarantee termination; no arbitrary recursion.
Bitstrings originate from Erlang~\cite{armstrong2007history}, and are also part of other languages compiling to the BEAM, like Elixir~\cite{elixir2025ref} and Gleam~\cite{gleam2025tour}.
Preventing null pointer dereferencing using Option-types originates from functional languages \cite{pierce2002types},
and are part of multiple modern system languages, among others Rust~\cite{klabnik2023rust} and Swift~\cite{apple2025swift}.
These languages also prevent accessing uninitialised memory access with flow analysis.
To the best of our knowledge, we do not know of any language preventing use-after-free bugs by prohibiting procedures to return pointer types. In a general purpose language, this is a severe restriction. However, in a DSL like BeePL we think it is a reasonable trade-off. The combination of the above language features and the certified pipeline make BeePL an appealing source language to write verified eBPF kernel code.

\paragraph{Safe programming languages}
We would also like to highlight the scope of extending or reusing other languages such as Koka~\cite{article}, CakeML~\cite{10.1145/2535838.2535841}, and Cogent~\cite{O’CONNOR_CHEN_RIZKALLAH_JACKSON_AMANI_KLEIN_MURRAY_SEWELL_KELLER_2021}
for our purpose. While CakeML features a fully verified compiler from a high-level functional language to machine code, it is not well suited for writing eBPF programs due to its general-purpose design. In particular, it does not target BPF bytecode and lacks support for the specific runtime conventions required by the eBPF execution environment. The same limitations apply to Koka. Both languages assume a general-purpose memory model that allows unrestricted use of dynamic memory allocation (e.g., via malloc), and their standard libraries heavily rely on such capabilities. Adapting these languages to target eBPF would require imposing the same constraints we enforce in BeePL, such as a restricted memory model, precise effect tracking, and compatibility with the eBPF verifier, which would entail substantial modifications and effort. Instead, we focused on adding a new front-end to CompCert, allowing us to reuse its verified optimization and code generation passes without having to reimplement/modify them. 

Cogent employs uniqueness typing to eliminate errors such as uninitialized memory access, use-after-free bugs, and undefined behavior. It generates refinement obligations in Isabelle/HOL that users must manually discharge to establish correctness relative to their specifications, though its compiler remains part of the trusted base. In contrast, BeePL provides a formally verified compilation toolchain, guaranteeing preservation of safety properties stated at the source level without additional user effort.

\paragraph{Fuzzing and testing}
Previous works~\cite{10278676, syzkaller, Buzzer, 10.1145/3643778, 10.1145/3627703.3629562} have applied fuzzing techniques to detect verifier bugs. While fuzzing is more effective than traditional testing, it lacks completeness, produces false positives, and struggles with scalability in detecting logical or semantic issues. The work~\cite{10278676} introduces a fuzzing method using a Bound-Violation Indicator (BVI) to find logical bugs in the eBPF verifier. It detects inconsistencies in range analysis between the verifier and actual execution. This exposes bugs that can lead to unsafe memory access. Syzkaller~\cite{syzkaller} is an automated fuzzer designed to test the Linux kernel by generating syscall sequences that explore deep kernel paths. It uses coverage-guided fuzzing to find bugs like crashes, memory corruption, and security vulnerabilities. BUZZER~\cite{Buzzer} is a fuzzer for the Linux kernel’s eBPF verifier that targets semantic bugs by solving symbolic constraints to generate programs that bypass the verifier. It combines symbolic execution with fuzzing to explore hard-to-reach verifier logic paths. The work~\cite{10.1145/3643778} highlights that fuzzers often fail to generate inputs that adhere to semantic specifications and necessary dependencies, making them ineffective for detecting correctness bugs. Since fuzzer-generated programs do not sufficiently trigger the verifier’s complex analyses~\cite{10.1145/3627703.3629562}, they mainly uncover simple memory issues rather than deeper correctness flaws. Addressing such bugs requires stronger assertions at key verification checkpoints.

\section{Limitations and future work}

Our framework does not currently analyze stack usage or enforce stack frame size limits, nor does it perform detailed scope or liveness analysis. To ensure memory safety, we conservatively disallow functions from returning the address of local variables, avoiding access to deallocated stack memory. All eBPF helpers are currently treated as external calls with assumed correctness, placing them in the trusted computing base (TCB). To reduce the TCB, we plan to formally specify and verify commonly used helpers, either by modeling their behavior in our semantics or proving their contracts. Future work includes integrating formal scope and liveness analysis to safely support address-returning functions, extending the type system to track memory-region-based effects for finer control over aliasing and lifetimes.

%% file: sections/appendix.tex
\section{Appendix}\label{sec:appendix}

\subsection{Auxillary Lemmas}
Lemma~\ref{progrss-unaryoperators} states that if a unary operation is applied to a well-typed value in a well-formed state, then its evaluation will succeed and produce a defined result. It guarantees that no unary operation will cause the program to get stuck.
\begin{lemma}[Progress for unary operators]\label{progrss-unaryoperators}
\small
\begin{align*}
\forall \Sigma, \Gamma, \Pi, \Psi, op, v. \ 
\typerule \Gamma \Sigma \Pi \Psi {\kw{uop}(op,v)} \tau \eta \wedge \kw{well\_formed\_state}(\Gamma, \Sigma, s) \implies \\
\exists v', \mathsf{uop_{sem}}(op,v) = v' \wedge v' \neq \kw{undef}
\end{align*}
\end{lemma}
\begin{proof}
The proof proceeds by performing a case analysis on the operator $op$, which yields separate goals for each unary operator supported in BeePL. BeePL supports unary operators such as boolean negation, integer complement, and unary negation. The semantics of unary operators are closely tied to the type of their arguments and only fail when applied to values of incorrect type. However, the typing rules in BeePL ensure that such ill-typed cases are ruled out statically, preventing any undefined behavior at runtime. Also, BeePL type system ensures that unary operators are only allowed on primitive types.
\textit{Boolean negation case.} From the hypothesis (typing rule $\kw{TUOP}$), we know that the type of the value $v$ is $\kw{bool}$. Negating a boolean value (i.e., $\kw{true}$ or $\kw{false}$) always succeeds and returns its logical complement ($\kw{false}$ or $\kw{true}$), which is a well-defined value. Thus, this case is proven.
\textit{Integer complement and negation case.} From the hypothesis (typing rule $\kw{TUOP}$), we know that the type of the value $v$ is either $\kw{int}$ or $\kw{long}$. We consider each type separately. Based on the semantics of integer complement and unary negation, both operations succeed on 32-bit and 64-bit integers and always produce defined results. Hence, the goal follows. 
\end{proof}

Lemma~\ref{progress-binaryoperators} states that if a binary operation is applied to a well-typed values in a well-formed state, then its evaluation will succeed and produce a defined result. It guarantees that no binary operation will cause the program to get stuck.
\begin{lemma}[Progress for binary operators]\label{progress-binaryoperators}
\small
\begin{align*}
\forall \Sigma, \Gamma, \Pi, \Psi, op, v_1, v_2. \ 
\typerule \Gamma \Sigma \Pi \Psi {\kw{bop}(op,v_1, v_2)} \tau \eta \wedge \kw{well\_formed\_state}(\Gamma, \Sigma, s) \implies \\
\exists v', \mathsf{bop_{sem}}(op,v_1, v_2) = v' \wedge v' \neq \kw{undef}
\end{align*}
\end{lemma}
\begin{proof}
The proof proceeds by performing a case analysis on the binary operator $op$, which yields separate subgoals for each binary operator supported by BeePL. To ensure well-defined semantics, we must guarantee that operators are applied to operands of appropriate types—a property enforced by the typing rule ($\kw{TBOP}$) for binary operators, as assumed in our hypothesis. The most interesting cases arise with operators that may result in undefined behavior, such as division, modulo, and shift. These operators are susceptible to runtime errors like division by zero or arithmetic overflow. We illustrate the reasoning using the case of the division operator.

The division operator in BeePL is categorized into signed and unsigned variants. For unsigned division, undefined behavior may occur only when the divisor is zero. According to BeePL's semantics, an unsafe unsigned division is handled gracefully by producing a default value—specifically, zero. In such cases, we apply the $\kw{BOPV}$ rule by constructing a value equal to 0 and discharging the goal that requires proving $\kw{unsafe}(op, v_1, v_2, s) = \kw{true}$. By definition, the $\kw{unsafe}$ predicate returns $\kw{true}$ precisely when $op$ is an unsigned division and $v_2 = 0$. For the safe case, i.e., when $\kw{unsafe}$ returns $\kw{false}$, we appeal to the semantics of binary operations, which ensures that the function $\kw{bop_{sem}}$ produces a valid result.

In the case of signed division, the semantics of $\kw{unsafe}$ are more involved, covering not only division by zero but also arithmetic overflow conditions (e.g., dividing the most negative value by $-1$ in two's complement). Nevertheless, the proof strategy remains the same: distinguish between the safe and unsafe cases using the $\kw{unsafe}$ predicate and construct the corresponding semantic outcome. This approach generalizes naturally to other potentially unsafe operators such as modulo and shift. In each case, the semantics ensure that well-typed expressions either evaluate to a defined value or are handled explicitly through default values when unsafe conditions are detected. Hence, we conclude that all well-typed binary operations in BeePL have a well-defined semantics, ensuring progress and safety even in the presence of operations that are typically prone to undefined behavior. 
\end{proof}

\begin{lemma}[Safe dereference for valid pointers]\label{safe-dereference}
\small
\begin{align*}
\forall \Sigma, l, s, \tau. \ 
&\Sigma \ l = \tau* \wedge \kw{isValidAccess}(s_\Theta,l, \kw{Freeable}) \implies \exists \ v, s_\Theta[l] = v \wedge \kw{typeof}(v) = \tau
\end{align*}
\end{lemma}
\begin{proof}
From the hypothesis, we know that the location $l$ is well-typed in the store typing environment $\Sigma$. Such a location appears in the store context only if it was freshly allocated using the $\kw{ref}$ construct or produced during the evaluation of an expression. If $l$ was created via $\kw{ref}$, then a value has already been stored at that location, making dereferencing safe. If $l$ originates from the external environment, its use is always guarded by a $\kw{match}$ construct that rules out null pointers, ensuring safety before dereferencing. Additionally, the hypothesis and the definition of valid access confirm that the address $l$ is allocated in memory $\Theta$ and has $\kw{Freeable}$ permission. According to the CompCert memory model, this level of permission allows all memory operations, including reads. Thus, a load from $l$ is guaranteed to succeed and return a value $v$. It remains to show that the value $v$ has the expected type. Since $\Sigma \ l = \tau*$, the location $l$ was allocated with type $\tau$, and only values of type $\tau$ can be stored at $l$ (due to the type preservation invariant of the CompCert memory model). Hence, the type of the value $v$ loaded from $l$ is indeed $\tau$, i.e., $\kw{typeof}(v) = \tau$.
\end{proof}

\begin{lemma}[Safe assignment for valid pointers]\label{safe-assignment}
\small
\begin{align*}
\forall \Gamma, \Sigma, l, s, v, \tau. \ 
\Sigma \ l = \tau* \wedge \kw{typeof}(v) = \tau \wedge \kw{isValidAccess}(s_\Theta,l, \kw{Freeable}) \implies\\ \exists \ \Theta', s_\Theta[l \rightarrow v] = \Theta' \wedge \kw{well\_formed\_store}(\Gamma, \Sigma, \prec s_\Delta, s_\Omega, \Theta'\succ)
\end{align*}
\end{lemma}
\begin{proof}
From the hypothesis, we know that the location $l$ is well-typed in the store typing environment $\Sigma$. This location may either result from a fresh allocation using the $\kw{ref}$ construct or originate from the external environment. In the case of $\kw{ref}$, the location is guaranteed to be valid, and the hypothesis ensures that the value being assigned is compatible with the pointer type associated with $l$, thereby making the assignment safe. For locations obtained from the external environment, the language enforces that any assignment must be wrapped in a $\kw{match}$ construct. This discipline ensures that assignments are only performed on valid, non-null locations, preventing unsafe memory updates. Finally, we show that the updated state $\prec s_\Delta, s_\Omega, \Theta' \succ$ remains well-formed. Here, $\Theta'$ is the updated memory obtained by assigning the value $v$ to the location $l$. Since the assignment respects the expected type of $l$, the location remains well-typed in the context $\Sigma$. Moreover, the memory invariant $\kw{isValidAccess}(\Theta', l, \kw{Freeable})$ continues to hold: the original memory $\Theta$ guaranteed valid access to $l$, and the assignment does not alter memory permissions. Hence, the updated state preserves all well-formedness invariants.
\end{proof}

Lemma~\ref{well-formed-allocation} ensures that memory allocation always succeeds and yields an updated memory that remains well-formed, provided the initial state was well-formed.
\begin{lemma}[well\_formed\_allocation]\label{well-formed-allocation}
\small
\begin{align*}
\forall \Gamma, \Sigma, s, lo, hi. 
&\kw{well\_formed\_state}(\Gamma, \Sigma, s) \implies \\
&\exists \Theta', l, v, \kw{alloc}(s_\Theta, lo, hi) = (\Theta', l) \wedge 
\kw{well\_formed\_store}(\Gamma, \Sigma, \prec s_\Delta, s_\Omega, \Theta'\succ)
\end{align*}
\end{lemma}
\begin{proof}
The proof follows from the properties of the $\kw{alloc}$ function. In CompCert, memory allocation is defined as a total function that always succeeds, assuming an infinite memory model at the source level. BeePL adopts this allocation model and adapts it to its own type system by wrapping CompCert’s memory primitives around BeePL-specific types. The allocation is parameterized by the bounds $lo$ and $hi$, which determine the size of the memory block to be allocated. From the hypothesis, we know that the initial state $s$ is well-formed. The allocation operation modifies only the memory component $s_\Theta$ of the state, producing an updated memory $\Theta'$ and a fresh location $l$. The updated memory $\Theta'$ is also well-formed. This is because allocation does not interfere with existing locations—it guarantees that the new location $l$ is fresh and does not overlap with any location already present in $\Theta$. According to the CompCert memory model, $l$ becomes valid only in $\Theta'$ and not in the original memory $\Theta$.
Furthermore, all locations that were valid in $\Theta$ remain valid in $\Theta'$ due to the monotonicity properties of memory allocation in CompCert. As a result, the invariant of well-formedness is preserved across the allocation, completing the proof.
\end{proof}

\begin{lemma}[ext\_well\_formed]\label{ext-well-formed}
\small
\begin{align*}
\forall\ \Gamma,\ \Sigma,\ s,\ x,\ l,\ \tau.\quad
&\kw{well\_formed\_state}(\Gamma, \Sigma, s)\ \wedge 
\big(\Gamma(x) = \kw{None} \lor (\exists\, \tau'.\, \Gamma(x) = \kw{Some}\ \tau' \wedge \tau' = \tau)\big)\ \Rightarrow \\
&\kw{well\_formed\_store}(\Gamma, \Sigma, \prec s_\Delta, s_\Omega[x \mapsto (l, \tau)], \Theta'\succ)
\end{align*}
\end{lemma}
\begin{proof}
We proceed by considering an arbitrary variable \( x' \) such that \( \Gamma(x') = \tau_x \) for some type \( \tau_x \). Our goal is to show that in the extended variable map \( s_\Omega[x \mapsto (l, \tau)] \), the value assigned to \( x' \) is a pair \( (l', \tau') \) such that \( \tau' = \tau \) and the location \( l' \) is well-typed in memory. The proof proceeds by case analysis on whether \( x' \) is the newly added variable:
\textit{Case 1: \( x' = x \).} Since \( \Gamma(x') = \tau_x \), it follows that \( x \) is already present in the typing environment. By our assumption, we have \( \tau = \tau_x \) and the memory location \( l \) is well-typed for type \( \tau \). Therefore, in the extended store \( s_\Omega[x \mapsto (l, \tau)] \), the binding for \( x' \) satisfies the well-formedness condition.
\textit{Case 2: \( x' \ne x \).} In this case, the extended variable map does not modify the binding for \( x' \), i.e.,  
\[
s_\Omega[x \mapsto (l, \tau)](x') = s_\Omega(x').
\]
Since the original store \( s_\Omega \) was well-formed, and the binding for \( x' \) is unchanged, the well-formedness condition continues to hold. Thus, in both cases, the extended store preserves the well-typedness of the state.
\end{proof}

Lemma ~\ref{ref_allocation_succeeds} states that if the state $s$ is well-typed and a value $v$ has type $t$, then we can always allocate fresh memory for $t$, store $v$ in it, and the resulting state will still be well-typed.
\begin{lemma}[ref\_allocation\_succeeds]\label{ref_allocation_succeeds}
\small
\begin{align*}
\forall \Gamma, \Sigma, s, v, t. 
&\kw{well\_formed\_state}(\Gamma, \Sigma, s) \wedge \kw{typeof}(v) = \tau \implies \\
&\exists \Theta', l, v, \kw{alloc}(s_\Theta, \kw{size}(\tau)) = (\Theta', l) \wedge \Theta'[l \rightarrow v] = \Theta'' \wedge 
\kw{well\_formed\_store}(\Gamma, \Sigma, \prec s_\Delta, s_\Omega, \Theta''\succ)
\end{align*}
\end{lemma}
\begin{proof}
From the hypothesis, we know that the store $s$ is well-formed and the value $v$ has type $\tau$. Let $sz$ denote the size computed from $\tau$ using the function $\kw{size}$, representing the number of bytes required to allocate a value of type $\tau$ in memory $s_\Theta$. The allocation range is defined with a lower bound of $0$ and an upper bound of $sz$. In CompCert, the function $\kw{alloc}$ is total: it always returns a newly allocated memory block along with an updated memory state. The underlying memory model is assumed to be unbounded at a high level, and the function succeeds as long as the constraint $lo \leq hi$ holds—in our case, $0 \leq sz$. This is also ensured by the Lemma~\ref{well-formed-allocation}. To guarantee that storing the value $v$ in memory succeeds, we must ensure that the address being written to is valid, the memory block has sufficient space to hold the value, and the location has appropriate write permissions. The allocation performed by $\kw{alloc}$ ensures these conditions by returning a fresh location $l$ with a valid offset range $[0, sz]$ and assigning it the $\kw{Freeable}$ permission. This permission level allows the memory to be read, written, and later deallocated safely.

The well-typedness of the value, captured by the hypothesis $\kw{typeof}(v) = \tau$, guarantees that the value is semantically consistent with the allocated memory layout and can be safely stored at the target location. Lastly, we prove that the update state$\prec s_\Delta, s_\Omega, \Theta''\succ$ is also well formed. The memory $\Theta'$ differs from $\Theta$ only at location $l$, which is freshly allocated and therefore not present in $\Gamma$ or $s_\Omega$. Since all three constructors of $\kw{well\_formed\_store}$ quantify only over variables and locations already present in $\Gamma$ and $\Sigma$, and all such blocks are unchanged in $\Theta'$ and $\Theta''$, the original well-formedness proof for $s$ still applies to $\prec s_\Delta, s_\Omega, \Theta'' \succ$.
\end{proof}

Lemma~\ref{var-allocation-succeeds} states that if we begin with a well-formed state $s$ and allocate a list of variables denoted by $\overline{x}$, then there exists an updated state in which all variables have been successfully allocated, and this new state remains well-formed.
\begin{lemma}[Var\_allocation\_succeeds]\label{var-allocation-succeeds}
\small
\begin{align*}
\forall \Gamma, \Sigma, s, x.\ 
&\kw{well\_formed\_state}(\Gamma, \Sigma, s) \wedge \kw{unique}\ {\overline{x}} \implies \\
&\exists \Omega', \Theta', \kw{alloc}(s_\Theta, \kw{size}(\overline{\tau_i})) = (\Theta', \overline{l}) \wedge 
s_\Omega[\overline{x} \rightarrow \overline{l_i,\tau_i}] = \Omega' \wedge
\kw{well\_formed\_store}(\Gamma, \Sigma, \prec s_\Delta, \Omega', \Theta'\succ)
\end{align*}
\end{lemma}
\begin{proof}
    The proof proceeds by induction on the list of variables \( \overline{x} \). \textit{Base case.} When the list is empty, the result is immediate. Since there are no variables to allocate, the original state is returned unchanged, and well-formedness follows directly from the hypothesis. \textit{Inductive case.} Suppose the list is of the form \( (x_1, \tau_1) :: \overline{(x_i, \tau_i)} \). By the induction hypothesis, allocating the remaining variables \( \overline{(x_i, \tau_i)} \) yields a new state that is well-formed, assuming the list contains no duplicates. For the head element \( (x_1, \tau_1) \), we apply Lemma~\ref{well-formed-allocation}, which guarantees that allocation always succeeds and produces a well-formed memory state. We use this result to instantiate the induction hypothesis on the tail \( \overline{(x_i, \tau_i)} \), using the updated memory obtained by allocating space for \( x_1 \) starting at offset \( 0 \) with size of $\tau_1$, resulting in the final memory state \( \Theta' \). Finally, Lemma~\ref{ext-well-formed} ensures that extending the variable environment with the new bindings \( s_\Omega[\overline{x} \rightarrow \overline{l_i, \tau_i}] \) preserves well-formedness. Thus, the extended state remains well-formed after allocating all variables in the list.
\end{proof}

Lemma~\ref{bind-variables-succeeds} states that if we begin with a well-formed state $s$ and bind the variables $\overline{x}$ with values $\overline{v}$ where the length of $\overline{x}$ is same as length of $\overline{v}$, then there exists an updated state in which all variables have been successfully assigned the values, and this new state remains well-formed.
\begin{lemma}[bind\_variables\_succeeds]\label{bind-variables-succeeds}
\small
\begin{align*}
\forall \Gamma, \Sigma, s.\ 
&\kw{well\_formed\_state}(\Gamma, \Sigma, s) \wedge |\overline{x}| = |\overline{v}| \implies \\
&\exists \Theta', s_\Theta[\overline{x} \rightarrow \overline{v}] = \Theta' \wedge
\kw{well\_formed\_store}(\Gamma, \Sigma, \prec s_\Delta, s_\Omega, \Theta'\succ)
\end{align*}
\end{lemma}
\begin{proof}
The proof proceeds by induction on the lists \( \overline{x} \) and \( \overline{v} \).
\textit{Base case.} If both lists are empty, then there are no variable-value bindings to perform. The state \( s \) remains unchanged, and the result follows immediately from the assumption that \( s \) is well-formed.
\textit{Inductive case.} Suppose the lists are of the form \( \overline{x} = x_1 :: x_n \) and \( \overline{v} = v_1 :: v_n \). By the well-formedness of the state \( s = \prec s_\Delta, s_\Omega, \Theta \succ \), we know that there exists a location \( l \) and a type \( \tau_1 \) such that \( s_\Omega[x_1] = (l, \tau_1) \), \( \Sigma \ l = \tau_1 \), and the memory \( \Theta \) allows a valid write access at location \( l \). Therefore, assigning \( v_1 \) to location \( l \) yields an updated memory \( \Theta_1 \), which remains well-formed due to the assumptions on the typing and access validity of \( l \). We now apply the induction hypothesis to the remaining lists \( x_n \) and \( v_n \), using the intermediate well-formed state \( \prec s_\Delta, s_\Omega, \Theta_1 \succ \). This yields a final store \( \Theta' \) obtained by assigning the remaining values \( v_n \) to the corresponding variables \( x_n \), such that the resulting state \( \prec s_\Delta, s_\Omega, \Theta' \succ \) is well-formed. Therefore, after successively updating the memory for all variable-value pairs, we obtain a final memory \( \Theta' \) such that the entire state remains well-formed.
\end{proof}

Lemma~\ref{subst-preserve-typing} states that substitution of a value $v$ for a variable $x$ in the expression $e$ preserves the type and subsumption rule for the effects. It is important to observe that substitution semantics accounts for the effect associated with the expression $e$. However, when e evaluates to a value prior to substitution, its effect is no longer present in the substituted expression. As a result, the conclusion must allow for a potential drop in effects, and we apply a subsumption rule on effects rather than requiring exact equality.
\begin{lemma}[subst\_preserve\_typing]\label{subst-preserve-typing}
\small
\begin{align*}
\forall \Gamma, \Sigma, x, e, e', \tau, \eta, \tau', \eta'.\ 
&{\typerule {(x \rightarrow \tau, \Gamma)} \Sigma \Pi \Psi {e'} {\tau'} {\eta'}} 
\wedge \typerule \Gamma \Sigma \Pi \Psi {e} {\tau} {\eta} \implies \\
& \exists \eta'', \typerule \Gamma \Sigma \Pi \Psi {e'[x \leftarrow e]} {\tau'} {\eta''} \wedge \eta'' \subseteq (\eta \dotplus \ \eta')
\end{align*}
\end{lemma}
\begin{proof}
The proof follows by doing induction on the structure of e. We discuss interesting cases in detail. 

\textit{Var case.} (i) \( x = y \). In this case, the variable \( y \) is substituted with the value \( v \). From the assumptions, we have \( \Gamma \vdash v : \tau \), and the typing environment includes \( x : \tau \). The expression \( y \) typechecks in the extended environment \( (x \rightarrow \tau, \Gamma) \), which implies that \( y \) has type \( \tau' = \tau \), since \( x = y \). After substitution, the resulting expression is \( v \), which has type \( \tau' \) as required. For the effects, since \( v \) is a value, we instantiate \( \eta'' \) with the empty effect \( \phi \), and since the empty effect is a subeffect of all other effects, we conclude that \( \eta'' \subseteq (\eta \dotplus \eta') \). Thus, the substitution preserves both typing and effects.
(ii) \( x \neq y \). In this case, no substitution occurs, and the resulting expression remains \( y \). From the assumption, we know that \( \Gamma \vdash y : \tau' \). Therefore, the type of the substituted expression is still \( \tau' \). The effect \( \eta'' \) is instantiated as \( \eta' \), and by the definition of effect subsetting, \( \eta' \subseteq (\eta \dotplus \eta') \). Hence, typing and effect preservation hold in this case as well. 

\textit{Function application} In the case of function application \( fn(es) \), we first apply the induction hypothesis to show that substitution preserves the typing of each argument in the list \( es \). Using these results, along with the typing rule for function application, we conclude that the entire expression \( fn(es) \) remains well-typed after substitution. 

\textit{Ref, Deref and Memory assignment case.} The proof proceeds by applying the induction hypothesis and appropriately instantiating the effect \( \eta'' \): with \( \kw{alloc :: \eta'} \) in the case of \texttt{ref}, \( \kw{read :: \eta'} \) in the case of \texttt{deref}, and \( \eta' :: \kw{alloc} \) in memory assignment, ensuring that the resulting effect annotations are correctly preserved. \textit{Uop, Bop, and Cond case.} The proof is trivial and utilizes induction hypothesis. 

\textit{Let-binding case.} For the let-binding expression \( \elet{y}{\tau}{e_1}{e_2} \), corresponding to the typing rule \textsc{TLET}, the proof proceeds in a manner similar to the variable case. We perform a case analysis on whether \( y = x \) or \( y \neq x \). In both cases, we apply the induction hypothesis to the subexpressions \( e_1 \) and \( e_2 \) as needed. The effect \( \eta'' \) is instantiated with the appropriate effect composition for the \textsc{TLET} rule. This ensures that the substituted expression remains well-typed.

\textit{Field access.} In the case of field access \( \structf{e}{f} \), substitution only applies to the subexpression \( e \). By applying the induction hypothesis to \( e \), and reapplying the typing rule \textsc{TFIELD}, we conclude that the substituted expression \( \structf{e[x \leftarrow se]}{f} \) remains well-typed, with the effect preserved as a subeffect of the original.

\textit{Struct initialization.} In the case of struct initialization \( \structi{y}{f}{e} \), we proceed by case analysis on whether \( y = x \) or \( y \neq x \).\textit{Case \( y = x \):} The variable \( y \) bound in the struct initialization shadows the variable \( x \) targeted by the substitution. Therefore, no substitution takes place within the list of expressions \( \overline{e} \). Since the original typing derivation holds and substitution does not alter the expression, the type is preserved. The effect is instantiated with the effect of the field initializers \( \overline{e} \), and the typing follows by applying the rule \textsc{TSINIT}.
\textit{Case \( y \neq x \):} In this case, substitution proceeds recursively into each initializer \( e_i \in \overline{e} \). By applying the induction hypothesis to each \( e_i \), we obtain the corresponding typing and effect preservation. Reapplying the rule \textsc{TSINIT} to the substituted list establishes the preservation of the overall expression.

\textit{For-loop} Substitution is applied to all three subexpressions \( e_1, e_2, e \) of a for-loop $\fore {e_1} {e_2} d {e}$, and the direction $d$ variable is assumed not to interfere with the substitution variable \( x \). By applying the induction hypothesis to \( e_1 \), \( e_2 \), and \( e \), we obtain that each substituted subexpression is well-typed with preserved or reduced effects. Reapplying the typing rule \textsc{TFOR}, we conclude that the substituted loop remains well-typed, and the total effect is preserved under subsumption rule for effects. \textit{None and Some case} For the None case, the substitution has no effect; hence type and effects are preserved. For the Some case, the induction hypothesis and the rule $\kw{TSOME}$ is used to discharge the goal.

\texttt{Match case} We consider the expression \( \kw{match}~e~\kw{with}~p_1 \Rightarrow e_1~|~\dots~|~p_n \Rightarrow e_n \), with typing judgment $
\Gamma, x : \tau \vdash \kw{match}~e~\kw{with}~p_1 \Rightarrow e_1~|~\dots~|~p_n \Rightarrow e_n : \tau_r, \eta$.
Let the effect be \( \eta = \eta_e \dotplus \eta_1 \dotplus \dots \dotplus \eta_n \), where: \( \Gamma, x : \tau \vdash e : \mathsf{option}~\tau', \eta_e \). For each branch \( i \), \( \Gamma_i = \Gamma, x : \tau \cup \textsf{binds}(p_i) \vdash e_i : \tau_r, \eta_i \) We proceed by case analysis on whether the pattern \( p_i \) binds \( x \).(i) Apply the induction hypothesis to the matched expression \( e \).  
We obtain: $\Gamma \vdash e[x \leftarrow se] : \mathsf{option}~\tau', \eta_e'
\quad \text{with} \quad \eta_e' \subseteq \eta_e \dotplus \eta_{se}$ (ii) For each branch \( i \), we consider the pattern \( p_i \) and the branch body \( e_i \). If \( x \notin \textsf{binds}(p_i) \), we apply the induction hypothesis:
  $\Gamma \cup \textsf{binds}(p_i) \vdash e_i[x \leftarrow se] : \tau_r, \eta_i'
  \ \text{with} \ \eta_i' \subseteq \eta_i \dotplus \eta_{se}$ If \( x \in \textsf{binds}(p_i) \), then \( x \) is shadowed, and substitution does not apply in \( e_i \), so: $\Gamma \cup \textsf{binds}(p_i) \vdash e_i : \tau_r, \eta_i
  \quad \Rightarrow \quad \eta_i' = \eta_i$
(iii) By reapplying the \textsc{TMATCH} rule, we derive:$
\Gamma \vdash \kw{match}~e[x \leftarrow se]~\kw{with}~p_1 \Rightarrow e_1'~|~\dots~|~p_n \Rightarrow e_n' : \tau_r, \eta'$
where $\eta' = \eta_e' \dotplus \eta_1' \dotplus \dots \dotplus \eta_n'$, and by construction:$
\eta' \subseteq (\eta_e \dotplus \eta_1 \dotplus \dots \dotplus \eta_n) \dotplus \eta_{se} = \eta \dotplus \eta_{se}$
Thus, substitution preserves both typing and effects.
\end{proof}

\begin{lemma}\label{progress}[Progress]
\small
\begin{align*}
\forall \Gamma, \Sigma, \Pi, \Psi, es, \eta_s, \tau_s, s.\quad
&\typerules{\Gamma}{\Sigma}{\Pi}{\Psi}{es}{\tau_s}{\eta_s} 
\wedge \mathsf{well\_formed\_state}(\Gamma, \Sigma, s) \implies \\
&\mathsf{isVal}\ es 
\vee \exists s', es'.\ \evals{s}{es}{s'}{es'} 
\wedge \mathsf{well\_formed\_state}(\Gamma, \Sigma, s') \\
\wedge\ 
\forall \Gamma, \Sigma, \Pi, \Psi, e, \eta, \tau, s.\quad
&\typerule{\Gamma}{\Sigma}{\Pi}{\Psi}{e}{\tau}{\eta} 
\wedge \mathsf{well\_formed\_state}(\Gamma, \Sigma, s) \implies \\
&\mathsf{isVal}\ e 
\vee \exists s', e'.\ \eval{s}{e}{s'}{e'} 
\wedge \mathsf{well\_formed\_state}(\Gamma, \Sigma, s')
\end{align*}
\end{lemma}
\begin{proof}
The proof proceeds by mutual induction on the structure of expressions and expression sequences, considering each form of expression and demonstrating that they either are values or can progress to the next state. Below, we present the proof for the case of a single expression. The case for a sequence of expressions follows similarly, by applying the induction hypothesis recursively at each step of the sequence. \\
\textit{Variable case}
When evaluating a variable expression $x$, we distinguish between two cases based on the scope of the variable: \emph{local} and \emph{global}. From the typing rule \textsc{(TVAR)}, we know that the environment $\Gamma$ maps $x$ to a type $\tau$, i.e., $\Gamma(x) = \tau$. The well-formedness of the state $s$ ensures that the variable $x$ appears either in the local environment $\Omega$ or in the global environment $\Delta$.
\begin{itemize}
  \item If $x \in \mathrm{dom}(\Omega)$, then there exists a location $l$ such that $\Omega(x) = (l, \tau)$, and by the well-formedness of memory, we have $\Sigma(l) = \tau$ and $s_\Theta[l, o] = v$ for some offset $o$ and value $v$ with $\mathrm{typeof}(v) = \tau$.
  \item If $x \notin \mathrm{dom}(\Omega)$, then $x$ must be a global variable, and we have $\Delta(x) = l$ for some location $l$. By well-formedness, $\Sigma(l) = \tau$ and there exists $o$ and $v$ such that $s_\Theta[l, o] = v$ and $\mathrm{typeof}(v) = \tau$.
\end{itemize} In both cases, the evaluation of $x$ yields the value $v$ without modifying the state $s$, and the resulting state remains well-formed. The proof for constant expressions is straightforward, as constants evaluate to themselves with the same type and do not affect the state.\\
\textit{Ref Case.}
The evaluation of the expression $\kw{ref}\ e$ proceeds by a case analysis:

\begin{itemize}
  \item \textit{Value case:} If the inner expression of \( \kw{ref} \ e \) has reduced to a value \( v \), then by the typing rule \textsc{(TREF)} and the well-formedness of the current state \( s \), Lemma~\ref{ref_allocation_succeeds} ensures that a fresh memory location \( l \) is successfully allocated and it is safe to store the value \( v \) at \( l \). Moreover, Lemma~\ref{ref_allocation_succeeds} guarantees that the resulting state \( s' \), obtained after allocation, remains well-formed. Therefore, the expression \( \kw{ref} \ v \) steps to a location \( l \), as captured by the semantic rule \textsc{(REFV)}. The subgoals arising from this rule---such as memory consistency and typing of the allocated location---are discharged using the guarantees provided by Lemma~\ref{ref_allocation_succeeds}.
   \item \textit{Step case:} By the induction hypothesis, the evaluation step $e \rightarrow e'$ preserves the well-formedness of the state. By the semantic rule $\textsc{(REF)}$, we have: $\kw{ref}\ e \rightarrow \kw{ref}\ e'$
  and the resulting state remains well-formed. 
\end{itemize}
\textit{Dereference case.} The evaluation of the expression \( !e \) proceeds by case analysis. From the typing rule \textsc{(TDEREF)}, we know that the type of \( e \) must be of the form \( \tau* \), i.e., a pointer to type \( \tau \), where the pointer may originate either from the BeePL program or from the external environment. In both cases, \( e \) denotes a memory location \( l \), and from the typing rule \textsc{(TLOC)} we obtain the judgment \( \Sigma \ l = \tau \), ensuring that the location is well-typed in the store typing context \( \Sigma \). Furthermore, from the structure of the typing derivation using \textsc{(TDEREF)}, it follows that the type of \( e \) cannot be an option type; such cases must be eliminated by a pattern match before dereference is allowed. 
\begin{itemize}
\item \textit{Value case:} Inner expression of \(!e \) has reduced to a location \( l \). We invoke Lemma~\ref{safe-dereference}, which, under the assumption that the state \( s \) is well-formed, guarantees the existence of a value \( v \) such that \( s_\Theta[l] = v \) and \( \kw{typeof}(v) = \tau \). Thus, \( s, !e \rightarrow s, v \) via the semantic rule \textsc{(DEREFV)}. Since dereferencing does not modify the state, we have \( s' = s \), and the well-formedness of the state is preserved.
\item \textit{Step case:} By the induction hypothesis, the evaluation step $e \rightarrow e'$ preserves the well-formedness of the state. Applying the semantic rule \textsc{(DEREF)}, we derive that \( s, !e \rightarrow s', !e' \), and since the well-formedness is preserved in the step from \( e \) to \( e' \), it follows that the overall evaluation of \( !e \) maintains the invariant that the state remains well-formed.
\end{itemize}

\textit{Memory assignment case.} The proof for memory assignment proceeds in a similar manner, using Lemma~\ref{safe-assignment} to handle the case where both $e_1$ and $e_2$ in the assignment expression $e_1 := e_2$ have been reduced to values. And for the case where they are not values, we utilize the induction hypothesis. 

\textit{Unary operators case.} The proof for unary operations relies on Lemma~\ref{progrss-unaryoperators}, which handles the case where the operand of the operator has already reduced to a value. According to this lemma, a well-typed unary operation always makes progress. Moreover, since unary operations do not modify the program state, the well-formedness of the state is preserved after evaluation. For the case where the operand has not yet reduced to a value, we apply the induction hypothesis to complete the proof.

\textit{Binary operator case.} The proof for binary operations relies on Lemma~\ref{progress-binaryoperators}, which handles the case where the operands of the operator has already reduced to a value. According to this lemma, a well-typed binary operation always makes progress. Moreover, since binary operations do not modify the program state, the well-formedness of the state is preserved after evaluation. For the case where the operands have not yet reduced to values, we apply the induction hypothesis to complete the proof.

\textit{Let-binding case.} We consider the expression \( \elet{x}{\tau}{e_1}{e_2} \). The evaluation proceeds by case analysis on whether \( e_1 \) is a value.
\begin{itemize}
\item \textit{Case 1:} \( e_1 = v \) for some value \( v \). In this case, the evaluation rule \textsc{(LETV)} applies, and the let-binding expression reduces as follows: $\elet{x}{\tau}{v}{e_2} \rightarrow e_2[x \leftarrow v]$. By unfolding the definition of substitution, we obtain the resulting expression \( e_2' := e_2[x \leftarrow v] \), which serves as the witness for the progress step. Since the substitution does not alter the memory state, we have \( s' = s \), and the well-formedness of the state is preserved.
\item \textit{Case 2:} \( e_1 \rightarrow e_1' \).
By the induction hypothesis applied to \( e_1 \), there exists an expression \( e_1' \) and a state \( s' \) such that:$
(s, e_1) \rightarrow (s', e_1')$ and $\mathsf{well\_formed\_state}(\Gamma, \Sigma, s')$
We then apply the semantic rule \textsc{(LET)} to conclude:$
s, \elet{x}{\tau}{e_1}{e_2} \rightarrow s', \elet{x}{\tau}{e_1'}{e_2}$
Hence, the let-binding expression makes progress, and the resulting state remains well-formed.
\end{itemize}

\textit{Cond case.} The typing rule $\kw{TCOND}$ ensures that the condition expression has type $\kw{bool}$ and that both branches have the same type, preserving type consistency. The proof involves three cases. The first two cases correspond to when the condition has already reduced to a boolean value: If the condition $e_1$ has evaluated to $\kw{true}$, then the conditional expression $\cond{e_1}{e_2}{e_3}$ reduces to $e_2$ by applying the rule $\kw{CONDT}$. If $e_1$ has evaluated to $\kw{false}$, the expression reduces to $e_3$ using the rule $\kw{CONDF}$. In both cases, the program state remains unchanged, and thus the well-formedness of the state is preserved. The third case addresses when the condition $e_1$ has not yet reduced to a value. In this case, the expression takes a step to $\cond{e_1'}{e_2}{e_3}$ by applying the rule $\kw{COND}$, where $e_1'$ is obtained from the induction hypothesis. By the induction hypothesis, we know that the state after evaluating $e_1$ remains well-formed.

\textit{Function application case.} The proof for function application proceeds by considering following cases.
\begin{itemize}
    \item \textit{Case 1:} The function \( fn \) is applied to a list of fully evaluated values \( \overline{v} \). From the typing rule \( \kw{TAPP} \), we know that the function \( fn \) has an arrow type. Moreover, from the definition of well-formedness of the state \( s \) (as specified in Definition~\ref{well-formed-state}), a well-typed function must be declared in the environment \( \Delta \). The number of arguments provided in the application must match the number declared in the function signature, and their types must align with the expected parameter types. In this case, the function application reduces to the function body, as specified in the declaration of \( fn \). This transition is captured by the rule \( \kw{APP3} \), which produces several subgoals related to memory allocation for local variables and binding of parameters to the argument values. These subgoals are discharged using Lemma~\ref{var-allocation-succeeds}, which ensures that the allocation of function-local variables preserves well-formedness, and Lemma~\ref{bind-variables-succeeds}, which guarantees that binding argument values to parameters results in a well-formed store.
   \item \textit{Remaining cases:} If some of the arguments to the function application are not yet fully evaluated, we use the induction hypothesis to show that the evaluation of these arguments progresses, preserving the well-formedness of the state.
\end{itemize}

\textit{Struct initialization and field access case.}We consider two subcases for struct initialization.
\begin{itemize}
\item \textit{Case 1:} All expressions \( \overline{e} \) to be assigned to the struct fields are fully evaluated to values \( \overline{v} \). From the well-formedness condition of the state, we know that every well-typed struct variable must be declared in the environment \( \Pi \). Using this fact, we extract the field layout information—specifically, the memory location and offset associated with each field name—and assign each value \( v_i \in \overline{v} \) to the corresponding memory position. By Lemma~\ref{safe-assignment}, we ensure that each write to memory is safe and that the resulting state \( s' \) remains well-formed.
\item \textit{Case 2:} Some expressions \( e_i \in \overline{e} \) have not yet reduced to values. \\
In this case, we apply the induction hypothesis to show that each such \( e_i \) can take a step \( e_i \rightarrow e_i' \) in a well-formed state. By applying the evaluation rule for struct initialization and appealing to the preservation of well-formedness under expression evaluation, we discharge the subgoal.
\end{itemize}

\textit{Struct Field Access.} To access a field \( f \) of a struct variable \( x \), we use the typing judgment to establish that \( x \) is well-typed and declared in the environment \( \Pi \). By the well-formedness of the state, we extract the memory location and offset of field \( f \). Then, by applying Lemma~\ref{safe-dereference}, we read the value at the specified location and offset, guaranteeing that the operation is safe and that the state remains well-formed.

\textit{For-loop case.} We consider the evaluation of the expression \( \kw{For}(e_1~\ldots~e_2, d)~e \) by performing case analysis on the evaluation status of the bound expressions \( e_1 \) and \( e_2 \).
\begin{itemize}
    \item \textit{Case 1:} \( e_1 \) is not a value.
From the typing rule \textsc{(TFOR)}, we know that \( e_1 \) is well-typed with type \( \tau \in \{ \kw{int}, \kw{long} \} \). By the induction hypothesis, the expression \( e_1 \) satisfies progress; hence, there exists \( e_1' \) and an updated state \( s' \) such that:$s, e_1 \rightarrow s', e_1'$
Applying the evaluation rule \textsc{(FOR1)}, we derive:$s, \kw{For}(e_1~\ldots~e_2, d)~e \rightarrow s', \kw{For}(e_1'~\ldots~e_2, d)~e$
\item \textit{Case 2:} \( e_1 \) is a value, but \( e_2 \) is not.
From the typing derivation, we similarly know that \( e_2 \) is well-typed. By the induction hypothesis, \( e_2 \) satisfies progress, so there exists \( e_2' \) and an updated state \( s' \) such that:$s, e_2 \rightarrow s', e_2'$. Applying the evaluation rule \textsc{(FOR2)}, we derive:$s, \kw{For}(v_1~\ldots~e_2, d)~e \rightarrow s', \kw{For}(v_1~\ldots~e_2', d)~e$
In both cases, the updated state \( s' \) remains well-formed by the induction hypothesis.
\item \textit{Case 3:} Both \( e_1 \) and \( e_2 \) are values.
Let \( v_1 = e_1 \), \( v_2 = e_2 \), and consider the expression:$\kw{For}(v_1~\ldots~v_2, d)~e$
From the typing rule \textsc{(TFOR)}, we have \( \kw{typeof}(v_1) = \kw{typeof}(v_2) = \tau \), where \( \tau \in \{\kw{int}, \kw{long}\} \), and that the loop body \( e \) is well-typed independently of \( e_1 \) and \( e_2 \). Specifically, the rule enforces:
$\mathsf{fvar}(e) \cap \mathsf{fvar}(e_1) = \phi \quad \text{and} \quad \mathsf{fvar}(e) \cap \mathsf{fvar}(e_2) = \phi$
ensuring that the loop bounds do not appear free in the loop body. Thus, the evaluation of \( e \) is unaffected by the values of \( v_1 \) and \( v_2 \), and once the bounds are fully evaluated, the loop can be unrolled. Let \( \kw{range}(v_1, v_2, d) = n \) denote the total number of iterations, and let \( \kw{repeat}(e, n) = e' \) denote the sequential unrolling of \( e \) repeated \( n \) times. The evaluation proceeds via the rule \textsc{(FOR3)}:$\kw{For}(v_1~\ldots~v_2, d)~e \longrightarrow e'$
\end{itemize}

\textit{None and Some case.} The proof for the $\kw{None}$ and $\kw{Some} \ v$  case is straight-forward as it does not alter the state and is a value; hence preserves the well-formedness of the state. For the case where $e$ in $\kw{Some} \ e$ steps to $e'$, the proof follows from the induction hypothesis. 

\textit{Match case.} We proceed by case analysis on the structure of the expression \( e_0 \) in the match expression \( \kw{match}~e_0~\kw{with}~\overline{p \rightarrow e} \). If \( e_0 \) is not a value, then by the induction hypothesis, there exists \( e_0' \) and \( s' \) such that \( s, e_0 \rightarrow s', e_0'\). In this case, we apply rule \textsc{Match1} to conclude that \( s, \kw{match}~e_0~\kw{with}~\overline{p \rightarrow e} \rightarrow s', \kw{match}~e_0'~\kw{with}~\overline{p \rightarrow e} \).

If \( e_0 \) is a value, we proceed by analyzing its form. If \( e_0 = \kw{Some}~v \), then since \( \kw{isSome}(v) \) holds and the pattern list contains a clause of the form \( (\kw{Psome}~x, e_2) \), rule \textsc{MSOME} applies, and we take the step to $s, e_2[x \leftarrow v]$. If \( e_0 = \kw{None} \), and the pattern list contains a clause \( (\kw{Pnone}, e_1) \), we apply rule \textsc{MNONE} to step to \( s, e_1 \).

If \( e_0 = v \) and \( \kw{typeof}(v) = \kw{bytes} \), we further distinguish based on the size of the byte sequence. If \( \kw{length}(v) \geq \kw{sizeof}(\tau_x) \), and the pattern list contains a clause \( ((x, \tau_x), e_1) \), then the byte extraction succeeds, and by rule \textsc{MBYTES}, we step to \( (e_1[\overline{y_i} \mapsto \overline{v_i}, x \mapsto v_x], s) \). Otherwise, if \( \kw{length}(v) < \kw{sizeof}(\tau_x) \), rule \textsc{MBYTESF} applies, and the evaluation steps to the fallback clause \( e_2 \), yielding \( s, e_2 \). This completes the case analysis for all possible forms of the match expression. In each case, we have shown that the expression either takes a step or is already a value (which does not apply for match), thereby establishing the progress property.

This completes the case analysis for all expression forms in the language. In each case, we have shown that a well-typed expression is either a value or can take a step according to the operational semantics, and that the state remains well-formed. Therefore, the progress property holds for all well-typed expressions of BeePL.
\end{proof}

\begin{lemma}\label{preservation}[Preservation]
\small
\begin{align*}
\forall \Gamma, \Sigma, \Pi, \Psi, es, es', \eta_s, \tau_s, s, s'.\quad
&\typerules{\Gamma}{\Sigma}{\Pi}{\Psi}{es}{\tau_s}{\eta_s} 
\wedge \mathsf{well\_formed\_state}(\Gamma, \Sigma, s) \wedge \evals{s}{es}{s'}{es'}  \implies \\
& \exists \eta_s'. \ \typerules{\Gamma}{\Sigma}{\Pi}{\Psi}{es'}{\tau_s}{\eta_s'} \wedge \eta_s' \subseteq \eta_s  
\wedge \mathsf{well\_formed\_state}(\Gamma, \Sigma, s') \\
\wedge\ 
\forall \Gamma, \Sigma, \Pi, \Psi, e, e', \eta, \tau, s, s'.\quad
&\typerule{\Gamma}{\Sigma}{\Pi}{\Psi}{e}{\tau}{\eta} 
\wedge \mathsf{well\_formed\_state}(\Gamma, \Sigma, s) \wedge \eval{s}{e}{s'}{e'} \implies \\
&\exists \eta'. \ \typerule{\Gamma}{\Sigma}{\Pi}{\Psi}{e'}{\tau}{\eta'} \wedge \eta' \subseteq \eta \wedge \mathsf{well\_formed\_state}(\Gamma, \Sigma, s') 
\end{align*}
\end{lemma}
\begin{proof}
The proof follows by doing induction on the typing derivation $\typerule{\Gamma}{\Sigma}{\Pi}{\Psi}{e}{\tau}{\eta}$ of expression. Below, we present the proof for the case of a single expression. The case for a sequence of expressions follows similarly, by applying the induction hypothesis recursively at each step of the sequence. We discuss some of the interesting cases here: 

\textit{Var case.} The semantics of variable access distinguishes between local and global variables. This distinction determines how the variable \( x \) is looked up—either in the local environment \( s_\Omega \) or the global environment \( s_\Delta \)—to retrieve the memory location \( l \) to which it is bound. Once the location \( l \) is obtained, the variable is dereferenced to retrieve the corresponding value \( v \) from the memory state \( s_\Theta \), i.e., $s_\Theta[l] = v$
By the property of well-typed dereferencing (as guaranteed by the typing rule \textsc{(TVAR)} and the store typing context \( \Sigma \)), we have \( \Sigma(l) = \tau \) and \( \kw{typeof}(v) = \tau \). Therefore, the resulting expression has type \( \tau \).
The effect associated with variable access is \( \phi \), and since \( \phi \subseteq \eta \) for any effect \( \eta \), subsumption holds. Furthermore, the operational semantics for variable access does not modify the state:$s' = s$ and hence, the well-formedness of the state is preserved.

\textit{Ref case.} Consider the expression \( \kw{ref} \ e \) and the typing rule \textsc{(TREF)}, which assigns it the type \( \tau* \), assuming that \( e \) has type \( \tau \). We proceed by case analysis on the evaluation of \( e \).
\begin{itemize}
    \item \textit{Case 1:} \( e \rightarrow e' \). By the induction hypothesis, we obtain:$\typerule{\Gamma}{\Sigma}{\Pi}{\Psi}{e'}{\tau}{\eta'}\ \text{with} \ \eta' \subseteq \eta$
and the updated state \( s' \) satisfies \( \mathsf{well\_formed\_state}(\Gamma, \Sigma, s') \). We then reapply the typing rule \textsc{(TREF)} to derive:$
\typerule{\Gamma}{\Sigma}{\Pi}{\Psi}{\kw{ref} \ e'}{\tau*}{\kw{alloc :: \eta'}}$
which establishes the preservation of typing for \( \kw{ref} \ e' \) under evaluation.
\item \textit{Case 2:} \( e = v \), where \( v \) is a value. 
In this case, by the operational semantics rule \textsc{(REFV)}, the expression \( \kw{ref} \ v \) allocates a fresh memory location \( l \), initializes it with the value \( v \), and yields:$\kw{ref} \ v \rightarrow l$. We solve this goal by applying the rule \textsc{(TLOC)}, which generates the subgoal: $\Sigma(l) = \tau$.
From the assumption that the original state \( s \) is well-formed, and by the semantics of \( \kw{ref} \), we know that \( l \) is a newly allocated and initialized location. This means \( \kw{isValidAccess}(s_\Theta, l, \kw{Freeable}) \) holds. By the third clause of the well-formedness definition (Definition~\ref{well-formed-state}), such valid allocations are properly reflected in \( \Sigma \), ensuring: $\Sigma(l) = \tau$.
Moreover, by Lemma~\ref{ref_allocation_succeeds}, the updated state \( s' \), after allocation, remains well-formed. Hence, preservation holds in this case as well. In terms of effect, location does not produce any effect; hence it will always be sub-effect of the effect associated with original $\kw{ref}$ construct. 
\end{itemize}

\textit{Deref case.} Consider the dereference expression \( !e \). We proceed by case analysis on the evaluation of \( e \).
\begin{itemize} 
\item \textit{Case 1:} \( e \rightarrow e' \).
By the induction hypothesis, there exists an updated state \( s' \) such that:
$
\eval{s}{e}{s'}{e'} \ \text{and} \
\typerule{\Gamma}{\Sigma}{\Pi}{\Psi}{e'}{\tau^*}{\eta'} \ \text{with} \ \eta' \subseteq \eta,\ \text{and} \ \mathsf{well\_formed\_state}(\Gamma, \Sigma, s')$. Applying the rule \textsc{(TDEREF)}, we obtain:$\typerule{\Gamma}{\Sigma}{\Pi}{\Psi}{!e'}{\tau}{\kw{read ::} \ \eta'}$
Since \( \eta' \subseteq \eta \), we conclude: $
\kw{read ::} \ \eta' \subseteq \kw{read ::} \ \eta$. Thus, preservation holds in this case.
\item \textit{Case 2:} \( e = l \), where \( l \) is a valid memory location, and \( !e \rightarrow v \).
By the operational semantics of dereferencing, the value stored at \( l \) is retrieved: $s_\Theta[l] = v$.
By Lemma~\ref{safe-dereference} and the well-formedness of the state \( s \), we have:
$\Sigma(l) = \tau \ \text{and} \ \kw{typeof}(v) = \tau$
Hence, $\typerule{\Gamma}{\Sigma}{\Pi}{\Psi}{v}{\tau}{\phi}$
Since \( \phi \subseteq \kw{read ::} \eta \), the effect condition is satisfied.
Furthermore, since dereferencing does not modify the memory: $s' = s$,
$\mathsf{well\_formed\_state}(\Gamma, \Sigma, s')$is preserved.
\end{itemize}

\textit{Memory assignment.} Preservation for memory assignment expressions of the form \( e_1 := e_2 \) follows a similar structure. In cases where \( e_1 \) and \( e_2 \) are not yet fully evaluated, the proof proceeds analogously to Case~1 of the \texttt{deref} and \texttt{ref} constructs, relying on the induction hypothesis for the subexpressions. When \( e_1 \) evaluates to a location \( l_1 \) and \( e_2 \) evaluates to a value \( v_2 \), the resulting expression is \( \kw{unit} \). Type preservation holds in this case as well, since the type of the entire assignment expression \( e_1 := e_2 \) is \( \kw{unit}_\tau \), matching the type of the result. The effect will be $\phi$, which is a subeffect of the original effect. The well-formedness of the updated state is guaranteed through lemma~\ref{safe-assignment}.

\textit{Let-binding.} For the case, where $e_1$ is not yet a value and can still take a reduction step,
we proceed exactly as in the earlier cases, invoking the induction hypothesis to obtain the desired
preservation properties. For the case, where $e_1$ is fully evaluated and has reduced to a value $v$ in $\elet{x}{\tau}{e_1}{e_2}$. According to the semantic rule $\textsc{LETV}$, let-binding steps to $e_2[x \leftarrow v]$.
By applying the lemma \texttt{Subst\_preserve\_typing}~\ref{subst-preserve-typing} to the original typing derivation, we obtain that the substituted expression \(e_2[x \leftarrow v]\) is also well-typed.
Because this reduction leaves the store unchanged, the well-formedness of the current state is preserved.

\textit{For loop.} For the case of for-loop $\fore {e_1} {e_2} d {e}$, we discuss the whole proof in two categories:
\begin{itemize}
  \item \textit{Sub-expressions are not fully evaluated.} If one of the subexpressions has not reduced to a value, we invoke the induction hypothesis for that sub-expression, exactly as in the let-binding and other structural cases, to obtain both subject reduction and preservation of state well-formedness.
  \item \textit{All control parameters are values.} if $e_1$ and $e_2$ has reduced to $v_1$ and $v_2$, the semantic rule $\kw{FORV}$ ensures $\mathsf{range}(v_1,v_2,d)=n \wedge \eval{\,s\,}{\mathsf{repeat}(e,n)}{s'}{v}$. The premise \(\eval{s}{\mathsf{repeat}(e,n)}{s'}{v}\) together with the induction hypothesis for the body \(e\) (applied \(n\) times) implies that \(v\) has type \(\tau\) and the final store \(s'\) is well formed.  
\end{itemize}

\textit{Struct Initialization and Field Access.}
For both struct initialization and field access, the key invariant rely on is the well-formedness of struct declarations. By the definition of state well-formedness (Definition~\ref{well-formed-state}), any struct used in a program must be declared in the composite environment \(\Pi\), and its declaration must type check under the typing context \(\Gamma\) (as stated in the last clause of the definition). This ensures that the types of all struct fields are known and consistent throughout the program. As a result:
\begin{itemize}
  \item \textit{Struct Initialization.}  
        The initialization of a struct preserves the type, since it simply associates values with the statically declared field types and does not modify the declared structure or type. The initialization of a struct's fields results in updates to memory. Since the struct is already declared in the composite environment $\Pi$, the corresponding memory locations for its fields are guaranteed to be valid. By Lemma~\ref{safe-assignment}, any assignment to a valid memory location succeeds and preserves state well-formedness. Therefore, the field-wise initialization of the struct not only succeeds but also results in a well-formed updated state.
  \item \textit{Field Access.}  
        Accessing a field reduces the expression to the value stored in the corresponding field of the struct. The field’s type must match the declaration in \(\Pi\), ensuring that the result is well-typed. Field access does not change the state; hence the well-formedness of the state is preserved. 
\end{itemize}

\textit{Match.} For the match case, where the sub-expressions are not fully reduced to value, we invoke the induction hypothesis for that sub-expression, exactly as in the other structural cases, to obtain both subject reduction and preservation of state well-formedness. 

For the case where $e$ has reduced to $\kw{Some} \ v$, $\match e p {pe}$, we know according to the rule $\kw{MSOME}$ that it reduces to $e_2[x \leftarrow v]$, where $e_2$ is the corresponding follow-up expression for the pattern $\kw{Psome}$. By applying the lemma \texttt{Subst\_preserve\_typing}~\ref{subst-preserve-typing} to the original typing derivation, we obtain that the substituted expression \(e_2[x \leftarrow v]\) is also well-typed. Because this reduction leaves the store unchanged, the well-formedness of the current state is preserved. For the case where $e$ evaluates to $\kw{None}$, the evaluation of $\match e p {pe}$ proceeds to the expression $e_1$, which is the follow-up expression for the pattern $\kw{Pnone}$. We use the inductive hypothesis corresponding to $e_1$ to discharge the goal. 

For the case where the type of $e$ in $\match e p {pe}$ in $\kw{bytes}$ and is reduced to a value $v$, we have two cases to consider:
\begin{itemize}
\item Given that \(\mathsf{typeof}(v) = \kw{bytes}\) and its length is sufficient, the matched pattern \(\kpbytes{x}{\tau_x}{y}{\tau_y}\) extracts values \(v_x\) and \(\{ y_i \mapsto v_i \}\). The expression \(e_1\) is typed under a context extended with \(x : \tau_x, y_i : \tau_y\). By multiple applications of the substitution lemma, the substituted expression is well-typed. The state update is identity (\(s' = s\)), so it remains well-formed.
\item The length check fails, so the fallback branch \(e_2\) is selected. Since \(e_2\) is one of the match branches, it is well-typed. No substitution or state update occurs, so both typing and well-formedness are preserved.
\end{itemize}

For the match construct, we carefully instantiate the effect corresponding to the selected branch. Since the typing rules ensure that each branch is typed with an effect that is no greater than the overall effect of the match expression, the instantiated effect is guaranteed to be a sub-effect of the original. Thus, effect subsumption holds by construction.
\end{proof}

To prove soundness we define a semantic closure for BeePL program as follows:
\input{figures/expr_sem_closure}
The relation $\twoheadrightarrow$ and $\rightsquigarrow$ defines the multi-step evaluation (semantic closure) of expressions and expression lists, parameterized by a natural number $n$ that records the number of steps taken to reach the final state. The rules $\kw{SONE}$ and $\kw{SONES}$ capture a single-step evaluation, while $\kw{SMULTI}$ and $\kw{SMULTIS}$ recursively compose multiple steps, incrementing $n$ to reflect the total number of steps.

The soundness lemma~\ref{soundness} states that if an expression (or sequence) is well-typed and the state is well-formed, then its evaluation to a final expression (or sequence) must result in a value. This guarantees that well-typed programs do not get stuck during execution.
\begin{lemma}[Soundness]\label{soundness}
\footnotesize
\begin{multline*}
\forall \Gamma, \Sigma, \Pi, \Psi, es, es', s, s', \eta_s, \tau_s, n.\quad
\typerules{\Gamma}{\Sigma}{\Pi}{\Psi}{es}{\tau_s}{\eta_s} 
\wedge \text{well\_formed\_state}(\Gamma, \Sigma, s) \wedge 
\mevals{s}{es}{s'}{es'}{n} \wedge es' \centernot\rightarrowtail \\
\implies 
\kw{isVal}\ es' \wedge \exists \eta_s'.\ 
\typerules{\Gamma}{\Sigma}{\Pi}{\Psi}{es'}{\tau_s}{\eta_s'} 
\wedge \eta_s' \subseteq \eta_s \wedge 
\text{well\_formed\_state}(\Gamma, \Sigma, s')
\end{multline*}

\begin{multline*}
\forall \Gamma, \Sigma, \Pi, \Psi, e, e', s, s', \eta, \tau, n.\quad
\typerule{\Gamma}{\Sigma}{\Pi}{\Psi}{e}{\tau}{\eta} 
\wedge \text{well\_formed\_state}(\Gamma, \Sigma, s) \wedge 
\meval{s}{e}{s'}{e'}{n} \wedge e' \not\rightarrow \\
\implies 
\kw{isVal}\ e' \wedge \exists \eta'.\ 
\typerule{\Gamma}{\Sigma}{\Pi}{\Psi}{e'}{\tau}{\eta'} 
\wedge \eta' \subseteq \eta \wedge 
\text{well\_formed\_state}(\Gamma, \Sigma, s')
\end{multline*}
\end{lemma}

\begin{proof}
The proof follows by doing induction on the number of steps $n$. 
\begin{itemize}
    \item \textit{n = 0}: By definition of multi-step evaluation $\kw{SZERO}$, we know \( s = s' \) and \( e = e' \). Given: $\typerule{\Gamma}{\Sigma}{\Pi}{\Psi}{e}{\tau}{\eta} \ \text{and} \ \mathsf{well\_formed\_state}(\Gamma, \Sigma, s)$, we invoke the progress lemma~\ref{progress}, which states that either: $\mathsf{isVal}(e) \ \text{or}$ there exists $s_1$ and $e_1$ to which it steps to. If \( e \) is already a value, the result follows directly by instantiating \( \eta' = \eta \), and the type and well-formedness is preserved as the state does not change. Otherwise, if \( e \) can step, then we contradict the assumption \( \meval{s}{e}{s}{e}{0} \), which disallows any step. Thus only the first case applies, and the base case holds.
    \item \textit{$n = S \ n'$}: By the definition of multi-step evaluation $\kw{SMULTI}$, we know that there exist intermediate state \( s_1 \) and expression \( e_1 \) such that: $\meval{s}{e}{s_1}{e_1}{n} \ \text{and} \ \eval{s_1}{e_1}{s'}{e'}$. Now we apply preservation theorem inductively over each step in $n$ to obtain the two key results: $\typerule{\Gamma}{\Sigma}{\Pi}{\Psi}{e_1}{\tau}{\eta_1}$ and $\mathsf{well\_formed\_state}(\Gamma, \Sigma, s_1)$. Applying the termination lemma~\ref{termination} on the typing derivation of $e$ and well-formedness of the initial state $s$, we know that there exists $m$, such that $\meval{s}{e}{s'}{e'}{m}$ and $e'$ is a value. We show that $m = n + 1$ and conclude that $e'$ is a value. Now applying the preservation lemma $\typerule{\Gamma}{\Sigma}{\Pi}{\Psi}{e_1}{\tau}{\eta_1}$, $\mathsf{well\_formed\_state}(\Gamma, \Sigma, s_1)$ and $\eval{s_1}{e_1}{s'}{e'}$, we conclude the typing derivation and well-formedness of the resultant state $s'$ after $n + 1$ steps. From the transitivity relation over the effects $\eta' \subseteq \eta_1$ ($\eta_1$ is obtained after applying preservation for n steps) and $\eta_1 \subseteq \eta$ ($\eta_1$ is obtained for the last step), we conclude that $\eta' \subseteq \eta$.
\end{itemize}
\end{proof}

\begin{lemma}[Safe Read]\label{safe-read}
\small
\begin{align*}
\forall \Gamma, \Sigma, \Pi, \Psi, e, \tau, \eta, s, s', v, n.\ 
& \typerule{\Gamma}{\Sigma}{\Pi}{\Psi}{e}{\tau}{\eta} \wedge 
\mathsf{well\_formed\_state}(\Gamma, \Sigma, s) \wedge 
\mevals{s}{e}{s'}{v}{n} \wedge  
\kw{Read} \in \eta \\
&\implies\ \exists \eta'.\ \kw{Alloc} \in \eta' \wedge \eta' \subseteq \eta
\end{align*}
\end{lemma}

\begin{proof}
We proceed by induction on the structure of the typing derivation $
\typerule{\Gamma}{\Sigma}{\Pi}{\Psi}{e}{\tau}{\eta}$

\texttt{Var case.} From the typing rule for variables, we have: $
\Gamma(x) = \tau \Rightarrow \typerule{\Gamma}{\Sigma}{\Pi}{\Psi}{\texttt{Var}\ x\ \tau}{\tau}{\{ \kw{Read} \}}$.
Since the expression produces a \kw{read} effect, and the variable $x$ must be present in the typing context $\Gamma$, the well-formedness of the state $\mathsf{well\_formed\_state}(\Gamma, \Sigma, s)$ ensures that $x$ must be allocated in memory. Therefore, the allocation must have occurred prior to this point and be tracked in the effect list. So we conclude that there exists $\eta' \subseteq \eta$ such that $\kw{alloc} \in \eta'$.

\texttt{Let-binding case.}
From the typing derivation:$
\typerule{\Gamma}{\Sigma}{\Pi}{\Psi}{e_1}{\tau}{\eta_1}, \kw{and}
\typerule{(x \rightarrow \tau, \Gamma)}{\Sigma}{\Pi}{\Psi}{e_2}{\tau_2}{\eta_2}
\Rightarrow \typerule{\Gamma}{\Sigma}{\Pi}{\Psi}{\elet x \tau {e_1} {e_2}}{\tau_2}{\eta_1 \dotplus \eta_2}.$

From the hypothesis, $\kw{read} \in \eta_1 \dotplus \eta_2$. We consider two cases:
\begin{itemize}
  \item If $\kw{read} \in \eta_1$, we apply the induction hypothesis to $e_1$. This gives $\eta_1' \subseteq \eta_1$ such that $\kw{alloc} \in \eta_1'$, and thus $\eta_1' \subseteq \eta_1 \dotplus \eta_2$.

  \item If $\kw{read} \in \eta_2$, then the read occurs in the body $e_2$ under the extended context. If the read depends on the variable $x$, then $x$ must have been allocated during the evaluation of $e_1$, so $\kw{alloc} \in \eta_1$. If the read is on another variable from $\Gamma$, then $\mathsf{well\_formed\_state}(\Gamma, \Sigma, s)$ ensures it was allocated already. In either case, the corresponding $\kw{Alloc}$ appears in $\eta_1$ or $\eta_2$ and hence in $\eta$.
\end{itemize}
All other expression forms (e.g., \texttt{App}, \texttt{Prim}, \texttt{Cond}, \texttt{Sfield}, \texttt{Match}, etc.) either preserve or combine effects from subexpressions. If $\kw{read} \in \eta$, it must have originated from some subexpression with its own typing derivation. The inductive hypothesis applies, so the required $\kw{alloc} \in \eta'$ for some $\eta' \subseteq \eta$. 
By structural induction on the typing derivation, any presence of $\kw{read}$ in the effect set implies a prior $\kw{alloc}$ also appears in a subset of the effect list.
\end{proof}

\subsection{BeePL Compiler}
In this section, we describe how the BeePL compiler translates high-level BeePL constructs into C, highlighting key program transformations that ensure memory and type safety in the generated code. We focus only on the most interesting cases that illustrate how the compilation process enforces safety through careful code generation.

\paragraph{Translation of BeePL ref construct to C}
The BeePL compiler first translates the inner expression $e$ that will be stored in the reference. It then introduces a fresh local variable to hold this value and generates code to assign the value to that variable. Finally, it returns the address of the variable, effectively modeling the behavior of ref as a pointer in C. This allows BeePL's reference semantics to be represented using local variable allocation and address-taking in the CompCert C syntax.
$$
\llbracket \mathsf{ref}~e \rrbracket =
\begin{cases}
  \texttt{i = ce;} & \llbracket e\rrbracket = ce \wedge \kw{fresh} \ i \wedge \kw{typeof} \ i = \llbracket\kw{typeof} \ e\rrbracket_\tau\\
  \texttt{\&i} \\
\end{cases}
$$

\paragraph{Translation of BeePL for-loop to C}
A \texttt{for} loop in BeePL is translated into imperative C-style code with explicit loop bounds and direction-aware control flow. The source-level loop is of the form: $\mathsf{for}~({e_1}\ldots{e_2}, d)~\{e\}$, where $e_1$ and $e_2$ specify the lower and upper bounds, $d$ denotes the direction ($\mathsf{Up}$ or $\mathsf{Down}$), and $e$ is the loop body.
\begin{itemize}
  \item The expressions $e_1$ and $e_2$ are compiled to C expressions $ce_1$ and $ce_2$, which are stored in fresh local variables $\mathsf{l}$ and $\mathsf{h}$.
  \item A fresh variable $\mathsf{i}$ is introduced to serve as the loop index.
  \item Depending on the direction:
  \begin{itemize}
    \item If $d = \mathsf{Up}$, the generated code is: $\texttt{if (l <= h)} \quad \texttt{for (i = l; i <= h; i++) \{ ce \}}$
    \item If $d = \mathsf{Down}$, the generated code is: $\texttt{if (l >= h)} \quad \texttt{for (i = l; i >= h; i--) \{ ce \}}$
  \end{itemize}
  \item If the guard condition is false (i.e., the range is empty), the entire loop reduces to \texttt{Skip}.
\end{itemize}
This transformation ensures that BeePL’s high-level iteration is compiled into safe, structured loops in C. Since the loop bounds are computed before iteration starts, and the index progresses monotonically in the specified direction, the resulting loop is guaranteed to terminate, provided that the subexpressions $e_1$, $e_2$, and $e$ themselves terminate. This structure allows the BeePL compiler to preserve termination guarantees at the C level.
\small
$$\llbracket\fore {e_1} {e_2} d {e}\rrbracket =
    \begin{cases}
        \mathsf{l = ce_1;} \mathsf{h = ce_2;} & \mathsf{Dir == Up} \wedge \mathsf{fresh \ l}, \mathsf{fresh \ h}, \mathsf{fresh \ i} \wedge \mathsf{\llbracket e_1\rrbracket = ce_1} \\
        \mathsf{if (l <= h)} & \wedge \mathsf{\llbracket e_2\rrbracket = ce_2} \wedge \mathsf{\llbracket e\rrbracket = ce} \\
        \mathsf{then \ \ for (i = l; i <= h; i++)\{ ce\} } \\
         \mathsf{else \ \ Skip} \\ \\ 
                
        \mathsf{l = ce_1;} \mathsf{h = ce_2;} & \mathsf{Dir == Down} \wedge \mathsf{fresh \ l}, \mathsf{fresh \ h}, \mathsf{fresh \ i} \wedge \mathsf{\llbracket e_1 \rrbracket = ce_1} \\
        \mathsf{if (l >= h)} & \wedge \mathsf{\llbracket e_2 \rrbracket = ce_2} \wedge \mathsf{\llbracket e \rrbracket = ce} \\
        \mathsf{then \ \ for (i = l; i >= h; i--)\{ ce\} } \\
         \mathsf{else \ \ Skip} \\ \\ 
    \end{cases}$$

\paragraph{Translation of BeePL let-binding to C}
The let-binding construct in BeePL, written as \texttt{let $x$ : $t$ = $e$ in $e'$} is translated as follows:
\begin{itemize}
  \item It first compiles the continuation $e'$ into a C code fragment $ce'$.
  \item It then examines the structure of $e$:
  \begin{itemize}
    \item If $e$ is a statement-level construct (such as an assignment, \texttt{for}-loop, struct initialization, or a nested \texttt{Bind}), it is compiled into a statement $ce$ and simply sequenced before $ce'$ as $ce;~ce'$.
    \item If $e$ is a pure expression, it is compiled into a C expression $ce$, and the compiler emits a local variable assignment \texttt{x = ce;} followed by $ce'$.
  \end{itemize}
  \item This translation preserves the evaluation order of BeePL and ensures that side effects in $e$ occur before evaluating $e'$, as required by the semantics.
\end{itemize}
This strategy ensures that the generated C code faithfully reflects the sequential and scoping behavior of BeePL’s \texttt{let}-bindings, preserving both value and effect ordering as in the source program. It also distinguishes between expression-level and statement-level forms to enable correct and structured code generation.
$$
\llbracket \mathsf{let}~x : t = e~\mathsf{in}~e' \rrbracket =
\begin{cases}
  ce ; ce' & \llbracket e' \rrbracket = ce' \wedge \llbracket e \rrbracket = ce \wedge e \in \{\kw{let}, :=, \kw{for}, \kw{sinit}\}\\
  x = ce ;\ ce' & \llbracket e' \rrbracket = ce' \wedge \llbracket e \rrbracket = ce \wedge e \notin \{\kw{let}, :=, \kw{for}, \kw{sinit}\}\\
\end{cases}
$$

\paragraph{Translation of BeePL match constructs to C}
The BeePL \texttt{match} expression on an \texttt{option} type get translated to a \texttt{NULL} check in C. If $\llbracket e \rrbracket = \texttt{NULL}$ (i.e., \texttt{None}), it evaluates to $\llbracket e_1 \rrbracket$. Otherwise, for \texttt{Some}~$x$, it binds $x = ce$ and evaluates $\llbracket e_2 \rrbracket$. This ensures safe and direct handling of optional pointers.

$$\llbracket \mathsf{match}~e~\mathsf{with}~p_1 \Rightarrow e_1~|~p_2 \Rightarrow e_2 \rrbracket =
    \begin{cases}
        ce_1 & \mathsf{p_1 == None} \wedge \mathsf{p_2 == Some \ x} \ \wedge \ \kw{typeof} \ e = \kw{option}(\tau*) \\ &\llbracket e_1 \rrbracket = ce_1 \wedge \llbracket e_2 \rrbracket = ce_2 \wedge \llbracket e \rrbracket = ce \wedge ce = \kw{NULL} \\

        x = ce; ce_2 & \mathsf{p_1 == None} \wedge \mathsf{p_2 == Some \ x} \ \wedge \ \kw{typeof} \ e = \kw{option}(\tau*) \\ &\llbracket e_1 \rrbracket = ce_1 \wedge \llbracket e_2 \rrbracket = ce_2 \wedge \llbracket e \rrbracket = ce \wedge ce \neq \kw{NULL} \\
    \end{cases}$$
    \begin{figure}[h]
\centering
\begin{minipage}[t]{0.50\textwidth}
\begin{lstlisting}[basicstyle=\scriptsize,mathescape=true]
#section "xdp"
fun bprog4(option(struct xdp_md$\rptr$) ctx) : int {
    match ctx.data with 
        | eth, struct ethhdr: (h_proto, uint16) => 
            if (h_proto == htons(ETH_P_IPV6)) 
            then XDP_DROP 
            else XDP_PASS
        | _ => XDP_DROP
}
char LICENSE[] SEC("license") = "GPL";
\end{lstlisting}
\caption{BeePL program using the match construct to ensure safe memory access within bounds}
    \vspace{-1em}
\label{fig:ebpfprogb4}
\end{minipage}
\hfill
\begin{minipage}[t]{0.45\textwidth}
\begin{lstlisting}[basicstyle=\scriptsize,mathescape=true]
SEC("xdp_drop")
int cprog4(struct xdp_md *ctx)
{
    void* data_end = (void star)(long)ctx->data_end;
    void* data = (void star)(long)ctx->data;
    struct ethhdr *eth = data;
    uint16 h_proto;
    if (data + sizeof(struct ethhdr) > data_end)
        return XDP_DROP;
    h_proto = eth->h_proto;
    if (h_proto == htons(ETH_P_IPV6))
        return XDP_DROP;
    return XDP_PASS;
}
char LICENSE[] SEC("license") = "GPL";
\end{lstlisting}
\caption{C program obtained after compiling bprog4 using the BeePL compiler}
    \vspace{-1em}
\label{fig:ebpfprogc4}
\end{minipage}
\end{figure}

The BeePL \texttt{match} expression on \texttt{Pbytes} (illustrated in \figref{ebpfprogb4}) is translated to C (illustrated in \figref{ebpfprogc4}) by first evaluating the expression $e$ to obtain a byte range \texttt{ce} of type \texttt{bytes\_t} defined as \texttt{struct bytes\_t \{unsigned char *start; unsigned char* end\}}, where the fields \texttt{start} and \texttt{end} represent the lower and upper bounds of the byte range, respectively.. The generated C code includes a conditional check:
\begin{center}
\texttt{if (ce.start + sizeof(struct s) > ce.end)}
\end{center}
to ensure there is enough memory to safely read a \texttt{struct s} from the byte stream. If the check fails, the fallback branch $e_2$ is evaluated. Otherwise, the code casts the byte pointer to a \texttt{struct s*}, extracts fields into variables (e.g., ${x.f_i := \kw{tmp} \rightarrow f_i}$), updates \texttt{ce.start}, and continues with $e_1$. This strategy prevents out-of-bounds memory access by ensuring that field accesses only occur within valid memory bounds. We focus here on the case of reading a struct from the buffer, though the approach generalizes to all supported types. However, if a struct field is itself a pointer, additional care is required. In such cases, simply verifying the bounds of the enclosing buffer is insufficient—we must also reason about the memory region pointed to by that field. One way to handle this safely is to disallow pointer fields during deserialization, or alternatively, to require an additional \texttt{match} expression to validate the bounds of the pointed-to memory region before accessing it. Without such precautions, the safety guarantees provided by the initial match on the enclosing struct could be compromised. At this moment, we focus only on structs whose fields are of non-pointer types.
\[
\llbracket
\mathsf{match}~e~\mathsf{with}
\begin{array}[t]{l}
\mathsf{Pbytes}(x, s, \{f_i : \tau_i\}) \Rightarrow e_1 \\
\mid\ \_ \Rightarrow e_2
\end{array}
\rrbracket =
\begin{cases}
  \llbracket e \rrbracket = \mathit{ce},\ \mathit{ce} : \mathsf{bytes\_t} \\[0.4em]
  \texttt{if (ce.start + sizeof(struct } s \texttt{) > ce.end)} \\[0.3em]
  \quad \llbracket e_2 \rrbracket \\[0.6em]
  \texttt{else \{} \\[0.3em]
  \quad \texttt{struct } s\texttt{* tmp = (struct } s\texttt{*) ce.start;} \\[0.3em]
  \quad \texttt{for each } f_i: x.f_i := \texttt{tmp} \rightarrow f_i \\[0.3em]
  \quad \texttt{ce.start += sizeof(struct } s\texttt{);} \\[0.3em]
  \quad \llbracket e_1 \rrbracket \\[0.3em]
  \texttt{\}}
\end{cases}
\]

To prove the compiler correctness, we first define the equivalence relation between the BeePL and C state:
\[
s_B \approx s_C \quad \triangleq \quad
\left(
\begin{array}{l}
\mathsf{\Delta}_B = \mathsf{\Delta}_C \wedge 
\mathsf{\Omega}_B = \mathsf{\Omega}_C \wedge
\mathsf{\Theta}_B = \mathsf{\Theta}_C \ \wedge \\
\forall l, o, \tau_B.\ \Sigma_B(l) = \tau_B \implies
\exists \tau_C.\ \llbracket\tau_B\rrbracket_\tau = \tau_C \wedge 
\Theta_C[\tau_C, l, o] = \Theta_B[\tau_B, l, o] 
\end{array}
\right)
\]
\textit{The equivalence relation } $\approx$ \textit{ asserts that BeePL and C program states are identical in their memory representation—both use the same variable map } ($\Omega$), \textit{ global environment } ($\Delta$), \textit{ and memory } ($\Theta$), \textit{ since BeePL and CompCert are built on a common implementation of these components. The only distinction lies in their type systems. A memory access at } $(l, o)$ \textit{ yields the same value in both languages, provided the BeePL type } $\tau_B$ \textit{ translates to the corresponding C type } $\llbracket \tau_B \rrbracket_\tau$.

\begin{lemma}[Compiler Correctness: BeePL to Csyntax]
\label{compiler-correctness}
\small
\begin{align*}
\forall \Gamma, \Sigma, \Pi, \Psi, e, e', s, s', \tau, \eta, n, c, v_c, v_b, c_s, c_s'.\quad
& \typerule{\Gamma}{\Sigma}{\Pi}{\Psi}{e}{\tau}{\eta} \ \wedge \
\mathsf{well\_formed\_state}(\Gamma, \Sigma, s) \ \wedge \
\llbracket e \rrbracket = c \ \wedge \\
& \meval{s}{e}{s'}{v_b}{n} \wedge  s \approx c_s \wedge 
c_s, c \Downarrow c_s', v_c \ \implies \\
& v_b = v_c \wedge s' \approx c_s'
\end{align*}
\end{lemma}
\begin{proof}
We discuss the proofs for the translations of the $\kw{ref}$, $\kw{for}$-loop, and $\kw{match}$ expressions as defined above. For the remaining expressions, the proof proceeds in a similar style. In particular, several expressions—such as constants, variables, and binary operations—have direct one-to-one mappings to Csyntax, making their correctness proofs straightforward. Other expressions involving memory or control flow are handled analogously using the inductive hypothesis and consistent memory manipulations.

\textit{$\mathsf{ref}~e$:} In BeePL, evaluating $\mathsf{ref}~e$ proceeds by first evaluating $e$ ($\meval{s}{e}{s_1}{v}{n_1}$), then
allocating a fresh memory location $\ell$ and storing $v$ at $\ell$ and finally returning the location $\ell$ as the result
Thus, $\meval{s}{\mathsf{ref}~e}{s'}{\ell}{n_1 + 1}$, where $s'$ is $s_1$ extended with the binding $\ell \mapsto v$, $\ell$ is fresh and $s_1$ is the intermediate state.

The BeePL compiler translates $\mathsf{ref}~e$ to: $\llbracket \mathsf{ref}~e \rrbracket = 
\texttt{i = } \llbracket e \rrbracket;~\texttt{\&i}$, where $\texttt{i}$ is a fresh local variable in C syntax, and the expression returns its address. The compiled C program performs evaluation of compiled C expression: $c_s$, $\llbracket e \rrbracket$ $\Downarrow$ $c_s'$, $v_c$, followes by storing $v_c$ in a fresh local variable \texttt{i}, and finally returning \texttt{\&i}.
So, the final result is the address of \texttt{i} holding $v_c$, and the final C state $c_s'$ contains this binding.

By inductive hypothesis on the evaluation of $e$:
\[
\meval{s}{e}{s_1}{v_b}{n_1} \wedge c_s, \llbracket e \rrbracket \Downarrow c_s^1, v_c \Rightarrow v_b = v_c \wedge s_1 \approx c_s^1
\]
After this step:
\begin{itemize}
  \item BeePL allocates a fresh location $\ell$ and binds $\ell \mapsto v_b$
  \item C creates a fresh local variable \texttt{i} and stores $v_c$ into it
\end{itemize}
By the freshness condition and we know that the BeePL’s and C’s environments grow consistently, we extend the relation:
$s' \approx c_s'$, where $s'$ extends $s_1$ with $\ell \mapsto v_b$, and $c_s'$ extends $c_s^1$ with \texttt{i} holding $v_c$. Finally, the result values: $v_b = \ell \ \text{and} \ v_c = \&\texttt{i}$ are considered equal under the equivalence relation: $\ell \approx \&\texttt{i}$. Thus, both BeePL and C perform the same evaluation steps, allocate storage, store the result, and return a pointer to it. The resulting states and values remain related:$v_b = v_c \wedge s' \approx c_s'$ as required.

\textit{$\kw{for}(e_1; e_2; d; e)$:} In BeePL, the loop expression evaluates the bounds $e_1$ and $e_2$ to values $v_1$ and $v_2$ respectively: $\meval{s}{e_1}{s_1}{v_1}{n_1} \ \text{and} \ \meval{s_1}{e_2}{s_2}{v_2}{n_2}$. Then, depending on the direction $d$, the loop body $e$ is evaluated iteratively:
\begin{itemize}
  \item If $d = \kw{Up}$ and $v_1 \leq v_2$, the loop runs from $v_1$ to $v_2$
  \item If $d = \kw{Down}$ and $v_1 \geq v_2$, the loop runs from $v_1$ down to $v_2$
\end{itemize}
producing a final state $s'$ and result $v$: $\meval{s}{\kw{for}(e_1; e_2; d; e)}{s'}{\kw{unit}}{n}$. 

By inductive hypotheses: $
\meval{s}{e_1}{s_1}{v_1}{n_1} \wedge c_s, ce_1 \Downarrow c_s^1, v_1 \Rightarrow s_1 \approx c_s^1$ and $
\meval{s_1}{e_2}{s_2}{v_2}{n_2} \wedge c_s^1, ce_2 \Downarrow c_s^2, v_2 \Rightarrow s_2 \approx c_s^2$. Then, for each iteration: $
\meval{s_i}{e}{s_{i+1}}{v}{n_i} \wedge c_{s_i}, ce \Downarrow c_{s_{i+1}}, v \Rightarrow s_{i+1} \approx c_{s_{i+1}}$

Thus, state equivalence is preserved throughout the loop. Hence; we conclude:
$v_b = v_c \wedge s' \approx c_s'$ as required.

\textit{Match on Bytes.} We proceed by case analysis on the runtime check generated by the compiler from the BeePL expression:
\[
\mathsf{match}~e~\mathsf{with}~\mathsf{Pbytes}(x, s, \{f_i : \tau_i\}) \Rightarrow e_1~|~\_ \Rightarrow e_2
\]

Let $\llbracket e \rrbracket = ce$, and suppose it evaluates in both BeePL and C to a byte region $\mathsf{ce} = (\mathsf{start}, \mathsf{end})$. The compiler translates this expression into the following C code:
\[
\texttt{if (ce.start + sizeof(struct } s \texttt{) > ce.end)}~\llbracket e_2\rrbracket~\texttt{else }\{...; \llbracket e_1 \rrbracket \texttt{\}}
\]

We analyze the two branches of this conditional.

\paragraph{Case 1: \texttt{ce.start + sizeof(struct s) > ce.end}}

In this case, the BeePL semantics selects the fallback branch $e_2$. Formally, we have:
$\meval{s}{e}{s_1}{\mathsf{ce}}{n_1} \ \text{and} \ \meval{s_1}{e_2}{s'}{v_b}{n_2}$
and by assumption, \\
$c = \llbracket \mathsf{match}~e~\mathsf{with}~... \rrbracket \ \text{evaluates in C as} \ c_s, c \Downarrow c_s', v_c$
Since the compiled C code takes the same branch (the condition fails), it executes $\llbracket e_2 \rrbracket$, producing $v_c$ and ending in state $c_s'$. By the inductive hypothesis on $e_2$, and since $s_1 \approx c_s''$ for some intermediate C state $c_s''$, we get:
\[
v_b = v_c \quad \text{and} \quad s' \approx c_s'
\]

\paragraph{Case 2: \texttt{ce.start + sizeof(struct s) $\leq$ ce.end}}

Here, BeePL succeeds in extracting a struct of type $s$ from the byte buffer. The semantics proceeds by casting $\mathsf{start}$ to a struct $s$ and then binding the corresponding parts to their respective fields. After successful copying of the buffer, the pointer is updated and the corresponding follow-up expression $e_1$ is evaluated in the extended state.
In C, the compiler emits the same logic as explained in the translation of match construct over bytes. By the inductive hypothesis applied to $e_1$, and assuming that field access and pointer arithmetic behave consistently in BeePL and C (as guaranteed by $s \approx c_s$), we conclude:
\[
\meval{s}{e}{s_1}{\mathsf{ce}}{n_1}, \quad \mathsf{start}' = \mathsf{start} + \mathsf{sizeof(struct~s)}, \quad \meval{s_1'}{e_1}{s'}{v_b}{n_2}
\]
and similarly, $c_s, c \Downarrow c_s', v_c, \ \text{where } c_s' \text{ reflects updated } \mathsf{start} \text{ and same field bindings}$

Therefore: $v_b = v_c \quad \text{and} \quad s' \approx c_s'$. In both cases, the BeePL and C semantics agree on the branch taken, the fields extracted, and the pointer updates. The evaluation of the continuation expression ($e_1$ or $e_2$) is correct by the inductive hypothesis. Hence, the compiled C code produces the same result as the BeePL code, and the final states remain related.

\end{proof}

%% file: figures/expr_sem_closure.tex
\begin{figure}[H]
 \footnotesize
  \[
\begin{array}{@{}c@{}}
    \inferrule*[left=\kw{SZERO}]{\eval s e {s} {e}}{\meval s e {s} {e} 0}\ \ \ \ \ \
    \inferrule*[left=\kw{SMULTI}]{\meval {s} {e} {s'} {e'} {n} \wedge \eval {s'} {e'} {s''} {e''}}{\meval s e {s''} {e''} {n+1}}\\[2ex]
    \inferrule*[left=\kw{SZEROS}]{\evals s {es} {s} {es}}{\mevals s {es} {s} {es} 0}\ \ \ \ \ \
    \inferrule*[left=\kw{SMULTIS}]{\mevals {s} {es} {s'} {es'} {n} \wedge \evals {s'} {es'} {s''} {es''}}{\mevals s {es} {s''} {es''} {n+1}}
  \end{array}
 \]   
 \caption{Semantics closures for expressions.}
\label{fig:sem_closure}
\end{figure}